%% file: main.tex
\documentclass[11pt]{article} 

\usepackage{fullpage}

\usepackage{minitoc}
\usepackage{amsmath,amsthm, amssymb}
\usepackage{latexsym}
\usepackage{graphicx}
\usepackage{ifthen}
\usepackage{subcaption}
\usepackage{multicol}
\usepackage{algorithm,algorithmic}
\usepackage[numbers]{natbib}
\usepackage{mathtools}
\usepackage{dsfont}
\usepackage{nicefrac}
\usepackage{soul}
\usepackage[normalem]{ulem}

\usepackage{stmaryrd}
\usepackage{bbm}

\usepackage{hyperref}
\hypersetup{colorlinks=true,linkcolor=blue,citecolor=blue,urlcolor=blue}

\usepackage{multirow}
\usepackage{makecell}

\newtheorem{theorem}{Theorem}

\newtheorem{lemma}[theorem]{Lemma}

 \newtheorem*{pb*}{Problem}
\newtheorem*{rem*}{Remark}

\newtheorem{definition}[theorem]{Definition}

\usepackage{thm-restate}

\usepackage{etoc}

\renewcommand{\arraystretch}{1.5}

\newcommand\numeq[1]%
  {\stackrel{\scriptscriptstyle(\mkern-1.5mu#1\mkern-1.5mu)}{=}}
\newcommand\numleq[1]%
  {\stackrel{\scriptscriptstyle(\mkern-1.5mu#1\mkern-1.5mu)}{\leq}}

\DeclareMathOperator*{\argmin}{arg\,min}
\DeclareMathOperator*{\argmax}{arg\,max}

\usepackage{algorithm}
\usepackage{algorithmic}
\usepackage{enumitem,kantlipsum}

\usepackage{xcolor}
\usepackage{color-edits}
\addauthor{eb}{orange}
\addauthor{yf}{red}
\addauthor{np}{blue}

\newcommand{\G}{\mathcal{G}}

\renewcommand{\H}{\mathcal{H}}
\newcommand{\B}{\mathcal{B}}
\renewcommand{\S}{\mathcal{S}}
\renewcommand{\P}{\mathbb{P}}

\newcommand{\A}{\mathcal{A}}
\newcommand{\hp}{n^{-\Omega(\log(n))}}
\newcommand{\E}{\mathbb{E}}
\newcommand{\io}{d_1\log(n)}
\newcommand{\reg}{\text{R is regular}}
\newcommand{\xlr}{x_{(\ell,m)}}
\newcommand{\ylr}{y_{(\ell,m)}}
\newcommand{\excess}{4\log(n)^2\sqrt{n}}
\newcommand{\indep}{\perp \!\!\! \perp}

\usepackage{authblk}

\newcommand{\OPT}{\texttt{OPT}}

\title{The Power of Greedy for Online Minimum Cost Matching \\ on the Line}

\author[]{Eric Balkanski}
\author[]{Yuri Faenza}
\author[]{Noemie Perivier}
\affil[]{Columbia University}

\date{}

\begin{document}
\pagenumbering{gobble}

\maketitle

\input{abstract}

\newpage

\tableofcontents
\newpage
\pagenumbering{arabic}
 \input{intro}

\input{preliminaries}

\input{section3}

\input{section5}

\smallskip

{\bf Acknowledgments.} This research was supported by the National Science Foundation through the grant \emph{CAREER: An algorithmic theory of matching markets}, by a Columbia Center of AI Technology (CAIT) in collaboration with Amazon faculty research award, and by a Columbia Center of AI Technology (CAIT) PhD Fellowship.

\newpage

\bibliographystyle{plainnat}
\bibliography{bib}

\newpage

\appendix

\section*{Appendix}

\input{appendix_extensions}

\input{appendix_strategyproof}
\input{appendix_auxilliary_lemmas}
\input{appendix_hybrid_lemma}

  \input{appendix_random_model}
\input{appendix_UB_semi_random}

\input{appendix_HG}

\input{appendix_overview_LB}

\input{appendix_LB_full_proof}

\end{document}

%% file: abstract.tex
In the online minimum cost matching problem, there are $n$ servers and, at each of $n$ time steps, a request arrives and  must be irrevocably matched to a server that has not yet been matched, with the goal of minimizing the sum of the distances between the matched pairs. Online minimum cost matching is a central problem in applications such as ride-hailing platforms and food delivery services. Despite achieving a worst-case competitive ratio that is exponential in $n$ even on the line, the simple greedy algorithm, which matches each request to its nearest available server, performs  very well in practice. A major question is thus to explain greedy’s strong empirical performance. In this paper, we aim to understand the performance of greedy on the line over instances that are at least partially random.

When both the requests and the servers are drawn uniformly and independently from  $[0,1]$, we obtain a constant competitive ratio for greedy, which improves over the previously best-known  bound of $O(\sqrt{n})$ for greedy in this setting, and  also show that this constant competitive ratio holds in the excess supply setting where there is a linear excess of servers. In the semi-random model where the requests are still drawn uniformly and independently but where the servers are chosen adversarially, we show that greedy achieves an $\Theta(\log{n})$ competitive ratio.
These results invite further investigation about how much randomness is necessary and sufficient to obtain guarantees for the greedy algorithm, on the line and beyond.

%% file: intro.tex


\section{Introduction}

Matching problems are a core area of discrete optimization. In the 90s, a seminal paper by~\citet{karp1990optimal} introduced online bipartite maximum matching problems and showed that, in the worst-case scenario, no deterministic algorithm can beat a simple greedy procedure, and no randomized algorithm can beat \emph{ranking}, which is a greedy procedure preceded by a random shuffling of the order of the nodes. These elegant results and their natural application to online advertising spurred much research, especially from the late 2000s on (see, e.g.,~\cite{mehta2013online} and the references therein for a survey). While more complex algorithms have been devised for models other than  worst-case analysis, greedy techniques are often used as a competitive benchmark for comparisons, see, e.g.,~\cite{feldman2010online,li2020two,xu2019unified}.

In the last few years, motivated by the surge of ride-sharing platforms, a second online matching paradigm has received  much attention: online (bipartite) minimum cost matching. In this class of problems, one side of the market is composed of \emph{servers} (sometimes called \emph{drivers}) and is fully known at time $0$. Nodes from the other side, often called \emph{requests} or \emph{customers}, arrive one at a time. When request $i$ arrives, we must match it to one of the servers $j$, and incur a cost $c_{ij}$. Server $j$ is then removed from the list of available servers, and the procedure continues. The goal is to minimize the total cost of the matching. 

Given the motivating application to ride-sharing, it is natural to impose the condition that both servers and requests belong to some metric space (e.g.,~\cite{KalyanasundaramP93,Kanoria21,Raghvendra16,TsaiTC94}).  Many  algorithms in this area involve  non-trivial, in some cases computationally expensive, procedures like randomized tree embeddings~\citep{Myerson, bansal}, iterative segmentation of the space~\citep{Kanoria21} or primal-dual arguments based on the computation of  offline optimal matchings at each time step~\citep{raghvendra2018optimal}. Other algorithms use randomization to bypass worst-case scenarios for deterministic algorithms~\citep{GuptaL12}.

The predominant objective of this line work has been to design algorithms that achieve the strongest possible performance guarantees in terms of quality of the solution found, which is measured by an algorithm's competitive ratio. However, there are other important considerations  when deploying systems that match individuals in real time, such as simplicity, strategyproofness, running time,  and explainability. An extremely simple algorithm that is highly desirable with respect to all these factors is the greedy algorithm, also called \emph{nearest neighbor}, that matches each incoming request to the closest available server. But does it perform well?

Somehow surprisingly, this algorithm often works very well in practice: experiments have shown that greedy was more effective than other existing algorithms in most tests and has outstanding scalability~\citep{TongSDCWX16}. This performance substantiates the choice of many ride-sharing platforms to actually implement greedy procedures, in combination with other techniques
~\citep{Lyft,UberEats}. 
However, current theory exhibits a mismatch with such strong computational results: if we assume that $n$ servers and $n$ requests are adversarially placed on a line, the greedy algorithm only achieves a  $2^n - 1$ competitive ratio~\citep{KalyanasundaramP93,khuller1994line, TsaiTC94}. It is therefore important to develop a theory that closes the gap with practice and gives solid ground to  the use of the greedy algorithm. This motivates the first guiding question of the paper.

\begin{quote}\emph{Can we find a theoretical justification for the strong practical performance of the greedy algorithm for online minimum cost matching problems?} \end{quote}

A standard approach to the question above is to make a distributional assumption on the input. Obviously, stronger assumptions may lead to stronger positive results -- but such assumptions may not be verified in practice. Ideally, we would like to identify the hypotheses that are necessary to guarantee a strong performance for the greedy procedure. These results can provide important guidance to practitioners: depending on whether or not they believe such hypothesis to be verified by their data, they can choose to either apply the greedy algorithm or to resort to a more refined procedure. This discussion motivates the second guiding question of the paper. 

\begin{quote}  
\emph{What are necessary and sufficient assumptions to guarantee that the greedy procedure outputs a solution whose quality is asymptotically optimal?}
\end{quote}

These questions have  important implications since they aim to characterize the scenarios where a simple greedy algorithm can be used instead of significantly more complex algorithms for a problem central to multiple large modern markets such as ride-sharing and food delivery.

Understanding the strong practical performance of simple algorithms  has motivated a lot of work on beyond the worst-case analysis of algorithms. Some examples include  using properties such as curvature, stability, sharpness, and smoothness to obtain improved guarantees for greedy for submodular maximization \citep{conforti1984submodular,chatziafratis2017stability,pokutta2020unreasonable,rubinstein2022budget} and different semi-random models for analyzing   $k$-means for clustering \citep{arthur2009k,manthey2013worst}, local search for the traveling salesman problem \citep{englert2014worst,kunnemann2015towards,englert2016smoothed,balkanski2022simultaneous}, and greedy for online \emph{maximum} matching \citep{goel2008online,devanur2011near,mastin2013greedy,arnosti2022greedy}. In the context of online \emph{minimum} cost matching, our understanding of the performance of greedy is very limited. Despite its simplicity, greedy is hard to analyze because a greedy match at some time step can have complex consequences on the available servers in a different region at a much later time step. In other words,  ``the state of the system under the standard greedy algorithm is
hard to keep track of analytically" \citep{Kan_arxiv}.

As a step towards understanding the power and limits of greedy, we focus on a fundamental, deceptively simple setting, which is in fact one of the most studied in the area: online minimum cost matching on the line. Despite much work~\citep{Akbarpour21,GuptaL12,peserico2021matching,Gupta2020PERMUTATIONSB,Megow,raghvendra2018optimal, Koutsoupias,Nayyar, Fuchs2003}, the performance of simple algorithms for this model are far from being understood.

 We first consider the \emph{fully random} model, where  the $n$ servers and $n$ requests are all drawn uniformly and independently from $[0,1]$. In this model, the best known bound on the competitive ratio of greedy is a trivial $O(\sqrt{n})$ bound,\footnote{There is a known $\Omega(\sqrt{n})$ bound on the optimal cost  (see, e.g., \cite{TsaiTC94}) and  the cost of any algorithm is trivially upper bounded by $n$.} and  there are more sophisticated algorithms such as hierarchical greedy~\citep{Kanoria21} and fair-bias~\citep{GuptaGPW19} that  are constant competitive in Euclidean spaces  and on the line, respectively.\footnote{Note that hierarchical greedy is constant competitive  only for $d=1$ or $d\geq 3$. For $d=2$, \citet{Kanoria21} shows that an adapted version of the gravitational matching algorithm by \citet{Holden21} is constant competitive.} Our first main result settles  the asymptotic performance of greedy  for matching on the line  in the fully random model by showing that greedy achieves a constant competitive ratio.

\begin{restatable}{rThm}{thmrandombalanced}
\label{thm:random_balanced}
For online matching on the line in the fully random model, the greedy algorithm achieves a constant competitive ratio.
\end{restatable}

A main benefit of greedy is that it is customer-strategyproof, meaning that the customers arriving online have no incentive to misreport the location of their requests. We note that this result improves the best-known competitive ratio of any  mechanism that is customer-strategyproof from $O(\sqrt{n})$ to constant for this setting (in fact, we are not aware of any non-trivial customer-strategyproof mechanism besides greedy). We refer to Appendix~\ref{sec:appstrategyproof} for a discussion and a formal definition of customer-strategyproofness.

We  show that this constant competitiveness of greedy also holds in the \emph{fully random $\epsilon-$excess model}, for every constant $\epsilon  \in [0,1]$. This is a modification of the fully random model where there is a linear excess of servers, i.e., $(1+\epsilon) n$ servers. This results improves over the previously best-known  competitive ratio for greedy of $O(\log^3 n)$ in this setting, which was a byproduct of a result by~\cite{Akbarpour21}. 

\begin{restatable}{rThm}{thmrandomunbalanced}
\label{thm:random_unbalanced}
For  any constant $\epsilon \in [0,1]$,  greedy is constant competitive in the fully random $\epsilon$-excess model.
\end{restatable}

It is widely acknowledged (see, e.g.,~\cite{feige2021introduction})  that i.i.d.~instances often do not resemble ``real'' instances. We next therefore consider whether strong guarantees for greedy can also be obtained in a semi-random model. 
In particular, we consider a model that we call the \emph{random requests model} where the $n$ servers are adversarially chosen and the requests are, as in the fully random model, drawn uniformly and independently. Our next result shows that greedy is logarithmic competitive in the random requests model.

\begin{restatable}{rThm}{thmsemirandomup}
\label{thm:semi_random_up}
For online matching on the line in the random requests model, the greedy algorithm achieves an $O(\log n)$-competitive ratio.
\end{restatable}

In the model where the servers and requests are chosen adversarially but where the arrival order is random,  $O(n)$ and  $\Omega(n^{0.26})$ upper and lower bounds are known for the competitive ratio of greedy~\citep{caragiannis2016truthful}.
Combined with this  $\Omega(n^{0.26})$ lower bound, our result shows that the performance of greedy improves exponentially when the locations of the requests are also random. Interestingly, hierarchical greedy only achieves a polynomial competitive ratio in the random requests model (see Appendix~\ref{app_HG}). Our last main result shows that this competitive ratio of greedy in the random requests model is tight. 

\begin{restatable}{rThm}{thmsemirandomlower}
\label{thm:semi_random_lower}
For online  matching on the line in the random requests model, the greedy algorithm achieves an $\Omega(\log n)$-competitive ratio.
\end{restatable}

Combined with Theorem~\ref{thm:semi_random_up}, we obtain that greedy is $\Theta(\log n)$-competitive in the random requests model. The combination of our four results give a first partial characterization of the scenarios in which greedy is guaranteed to perform well for online minimum cost matching.
However, there remain multiple intriguing and well-motivated extensions of the two models we consider where the performance of greedy is poorly understood and where strong competitive ratio guarantees might be achievable. These extensions and their challenges are discussed in Appendix~\ref{app_extension} and include 
\begin{itemize}
\item more general metric spaces beyond the line, such as the unit hypercube of arbitrary dimension $d$ for which we provide numerical simulations where greedy achieve a competitive ratio that is always at most $1.4$ (Appendix~\ref{app_dimensions}),
\item a relaxation of the uniform assumption where the requests are instead drawn i.i.d. from an arbitrary distribution (Appendix~\ref{app_iid}),
\item the random servers model where the servers are adversarially chosen and the requests are drawn uniformly and independently (Appendix~\ref{app_randomservers}), and
\item the sublinear excess supply setting where there is a sublinear excess number of servers compared to the number of requests (Appendix~\ref{app_sublinear}). 
\end{itemize}
 Since these extensions better capture the ride-hailing  and food-delivery applications, bounds on the performance of greedy under such extensions would provide stronger justifications for such platforms to use simple greedy algorithms.

\subsection{Technical overview}
\label{sec:techoverview}

The main difficulty in analyzing the greedy algorithm is that there can be complex dependencies between a greedy match that occurred at some time step in some region of the line  and the set of remaining  servers that are available at a later time step in a completely different region of the line. In other words, a single greedy match at some time step can have a butterfly effect on the servers that will be available in the future in different regions. Algorithms such as hierarchical greedy that partition the interval in different regions have been designed to prevent matching decisions in one region from impacting the future available servers in another region. This does not necessarily lead to algorithms that are better than greedy, but does give algorithms that are simpler to analyze.

A high-level contribution of our paper is to develop a general framework for analyzing the greedy algorithm, for both upper and lower bounds, that, we believe, also provides foundations for analyzing greedy in higher dimensions and other partially random models. The starting point of our analysis is to consider a \emph{hybrid algorithm} $\H_\A^m$ that matches the first $m$ requests according to an  algorithm $\A$ and then greedily  matches each of the remaining requests to the closest available server. The algorithm $\A$ is different for each of our results. To derive our upper bound results, we first show a \emph{hybrid lemma} that upper bounds, for any algorithm $\A$ that satisfies some fairly general properties, the difference $\mathbb{E}[cost(\H_\A^{m-1}) - cost(\H_\A^{m})]$  (i.e., between the expected total costs incurred by $\H_\A^{m-1}$ and $\H_\A^{m}$) as a function of the cost incurred by $\A$ to match the $m^{th}$ request. 
 This hybrid algorithm idea  was also used in~\cite{GuptaL12} to show a $O(\log(n))$  upper bound on the competitive ratio of a randomized greedy algorithm for online matching, but with three main differences. The first is that their hybrid algorithm is used to analyze a randomized algorithm on a deterministic instance (instead of a deterministic algorithm on a randomized instance). The second is that their hybrid algorithm uses an optimal offline algorithm $\A$, which we cannot use because we need to exploit the randomness of the instance, so we instead use existing online algorithms. The third is that our bound on $\mathbb{E}[cost(\H_\A^{m-1}) - cost(\H_\A^{m})]$ is tighter, which was a necessary improvement to obtain a constant competitive ratio in the fully random model.

The second  part of the analysis of the upper bounds leverages the hybrid lemma. For the fully random model, we consider  the hybrid algorithm $\H_\A^m$ where $\A$ is the constant-competitive hierarchical greedy algorithm by \citet{Kanoria21}. We note that a direct application of the hybrid lemma with this hybrid algorithm would only give an $O(\log n)$ competitive ratio  for greedy. Instead, we also show that   the total cost of the hierarchical greedy algorithm $\A$ is dominated by the cost of requests that are matched to servers at a constant distance away, which is needed to show that the difference between the expected costs of  greedy and hierarchical greedy is $O(\sqrt{n})$.
Since the expected optimal total cost is known to be $\Theta(\sqrt{n})$ and hierarchical greedy is constant competitive, we get that greedy is also constant competitive.   For the random requests model, we again use the hybrid lemma but with a different algorithm $\A$, which is  a simple modification of the fair-bias algorithm by \citet{GuptaGPW19}, to show that greedy achieves an $O(\log n)$ competitive ratio.

For the $\Omega(\log n)$ lower bound in the random requests model, we consider an instance where there is a large number of servers at location $0$, no servers in  $(0, n^{-1/5}]$, and the remaining $1-o(1)$ servers uniformly spread in  $(n^{-1/5}, 1].$ We again analyze the difference $\mathbb{E}[cost(\H_\A^{m-1}) - cost(\H_\A^{m})]$, but where $\A$ is the tailored algorithm that matches any request in $[0, n^{-1/5}]$ to a server at $0$ and greedily matches any other request to the closest available server. We show that at any time step $t$, the set of available servers for $\H_\A^{m-1}$ and $\H_\A^{m}$ differ in at most one server. We then consider the distance $\delta_t$ at time $t$ between these two different servers that are available to only one of the algorithms and we show that $\mathbb{E}[cost(\H_\A^{m-1}) - cost(\H_\A^{m})]$ can be lower bounded as a function of $\max_{t\geq m} \delta_t$. Due to the randomness of the requests, the main difficulty is to lower bound  $\max_{t\geq m} \delta_t$ (e.g., the gap $\delta_t$ can either shrink or expand at each time step), which we do by giving a careful partial characterization of the remaining  servers $(S_0,\ldots, S_n)$ for $\H_\A^{m}$ at each time $t$ that allows to analyze  $(S_0,\ldots, S_n)$ and  $(\delta_0,\ldots, \delta_n)$ separately.

\subsection{Additional related work}

In general metric spaces with adversarial requests and servers, \citet{KalyanasundaramP93} and \citet{khuller1994line}   gave a $2n - 1$ deterministic competitive algorithm and proved that this competitive ratio is optimal for deterministic algorithms. On the line, \citet{KalyanasundaramP93} and \citet{khuller1994line}  showed that the competitive ratio of greedy is at least $2^n - 1$. A deterministic algorithm with a sublinear competitive ratio  was presented in~\cite{Antoniadis}. A few years later, \citet{Nayyar} gave a $O(\log^2 n)$ competitive deterministic algorithm, which was then shown to be $O(\log n)$-competitive in \cite{raghvendra2018optimal}.  Regarding lower bounds, \citet{Fuchs2003} showed that no deterministic algorithm can achieve a competitive ratio strictly less than $9.001$ on the line. 

For randomized algorithms, still for adversarial requests and servers,  \citet{Myerson} and \citet{Csaba} obtained a $O(\log^3 n)$ competitive ratio in general metric spaces using randomized tree embeddings, which was later improved to $O(\log^2n)$  by \citet{bansal}.   
On the line, and for doubling metrics, \citet{GuptaL12} showed that a randomized greedy algorithm is $O(\log n)$ competitive.  Recently, \citet{peserico2021matching} improved the lower bound from \cite{Fuchs2003} to obtain an $\Omega(\sqrt{\log n})$ lower bound for the line that also holds for randomized algorithms. For general metrics, it was previously known that no randomized algorithm can achieve a  competitive ratio better than $\Omega(\log{n})$ \cite{Myerson}.

\begin{table}[]
\centering
\resizebox{\columnwidth}{!}{%
\begin{tabular}{|l|l||l|l|l|l|l|l|}
\cline{3-8}
\multicolumn{1}{c}{} & & \multicolumn{2}{ c |}{Greedy algorithm} & \multicolumn{2}{ c |}{Deterministic algorithms} & \multicolumn{2}{ c |}{Randomized algorithms}    \\ \hline
   \multicolumn{2}{|c||}{Arrival order}    & random   & adversarial& random     & adversarial& random& adversarial \\\hline \hline
 \multirow{2}{*}{Line }  & UB  &     $n$    & $2^n-1$   &         $O(\log n)$       &   $O(\log n)$ \cite{raghvendra2018optimal} &      $O(\log n)$       &   $O(\log n)$ \cite{GuptaL12} \\ \cline{2-8}
     & LB  &   $n^{0.26}$      &  $2^n-1$  &               & $\Omega(\sqrt{\log n})$   &            &   $\Omega(\sqrt{\log n})$ \cite{peserico2021matching} \\ \hline
 \multirow{2}{*}{\makecell{General \\ metric \\ space} }  & UB &    $n$  \cite{caragiannis2016truthful}   &   $2^n-1$ \cite{KalyanasundaramP93, khuller1994line}  &    $O(\log n)$       \cite{Raghvendra16}    &  $2n - 1$ \cite{KalyanasundaramP93, khuller1994line}  &   $O(\log n)$         &   $O(\log^2 n)$ \cite{bansal} \\ \cline{2-8}
     & LB &     $n^{0.26}$  \cite{caragiannis2016truthful}  &  $2^n-1$ \cite{KalyanasundaramP93, khuller1994line}   &       $\Omega(\log n)$         &  $2n - 1$  \cite{KalyanasundaramP93, khuller1994line}   &     $\Omega(\log n)$  \cite{Raghvendra16}      &  $\Omega(\log n)$ \cite{Myerson} \\ \hline 
\end{tabular}
}
\vspace{.15cm}
\caption{Summary of known competitive ratios for online minimum cost  matching.}
\label{tab:summary}
\end{table}

When the arrival order of the requests is random,  \citet{caragiannis2016truthful} showed that greedy is $O(n)$ and $\Omega(n^{0.26})$ competitive. \citet{Raghvendra16} gave a deterministic algorithm that achieves a $O(\log n)$ competitive ratio, which is optimal even for randomized algorithms. When the requests are drawn i.i.d.~from any distribution over the set of servers, \citet{GuptaGPW19}  gave a $O((\log\log\log{n})^2)$ competitive algorithm in general metric spaces that is also constant competitive on the line and for tree metrics. When the servers and requests are uniformly and independently distributed, \citet{TsaiTC94} showed that greedy achieves an $2.3\sqrt{n}$ competitive ratio on the unit disk and \citet{Kanoria21} showed that an algorithm called hierarchical greedy is constant competitive  on the unit hypercube (and also analyzed the more challenging fully dynamic setting where the servers also arrive online). A summary of the best-known bounds for the competitive ratio in different settings is provided in Table~\ref{tab:summary}.

Empirical evaluations of different algorithms on real spatial data have shown that greedy performs  well in practice~\citep{TongSDCWX16}. The excess supply setting  was studied by \citet{Akbarpour21}, who showed that the expected total optimal cost is constant  and the total cost of greedy is $O(\log^3 n)$ when the number of excess servers is linear and when the requests and servers are random (but the arrival order can be adversarial). The results for hierarchical greedy from \cite{Kanoria21} also extends to the excess supply setting. \citet{KalyanasundaramP00} showed a $O(\min(m, \log(n)))$ bound on the “double-competitive ratio” of greedy in an adversarial model with resource augmentation where there are $m$ possible server locations and the adversary has only half as many servers at each location as greedy. Recourse, i.e. allowing matching decisions to be revoked to some extent, has been considered in \cite{Megow, Gupta2020PERMUTATIONSB}. In the offline non-bipartite version of the problem with $2n$ point drawn uniformly from $[0,1]$, \citet{Frieze} showed that  greedy   achieves a $\Theta(\log{n})$ approximation.

%% file: preliminaries.tex
\section{Preliminaries}
\label{sec:preliminaries}

In the online matching on the line problem, there are $n_s$ \emph{servers} $S = \{s_1, \ldots, s_{n_s}\}$ and $n=n_r$ \emph{requests} $R = (r_1, \ldots, r_{n})$ such that $s_i, r_i \in [0,1]$ for all $i$. 
Hence, an instance is given by a pair $(S,R)$. The servers are known to the algorithm at time $t = 0$. For all $t \in [n]$, the algorithm  observes request $r_t$ and must irrevocably  match it to a server that has not yet been matched. We denote by $s_\A(r_t)$ the server that gets matched to request $r_t$ by (the current execution of) algorithm $\A$ and by $S_{\A,0}\supseteq  \cdots\supseteq S_{\A,n}$  the  sets of free servers obtained through the execution of $\A$, where $S_{\A,0}$ is the initial set of servers, and for all $t\in [n]$,  $S_{\A,t}$ is the set of remaining free servers just after matching $r_t$. The cost incurred from matching $r_t$ to $s_{\mathcal{A}}(r_t)$ is  $\text{cost}_t(\mathcal{A}, r_t) = |r_t-s_{\mathcal{A}}(r_t)|$ and  the total cost of the matching produced by $\mathcal{A}$ on instance $I$ is $\text{cost}(\mathcal{A}, I) = \sum_{t=1}^n \text{cost}_t(\mathcal{A}, r_t)$. We often abuse notation and write $\text{cost}_t(\mathcal{A}), \text{cost}(\mathcal{A}),$ and   $S_{t}$ instead of $\text{cost}_t(\mathcal{A}, r_t), \text{cost}(\mathcal{A}, I)$, and   $S_{\A,t}$. Unless specified otherwise, ``time step $t$" refers to the time just after matching $r_t$.

All models studied in the paper can be represented by a triple $(n^u, n^d, n)$. Here, $n^u$ (resp.~$n$) is the cardinality of the set $S^u$ of servers  (resp.~of the set $R$ of requests) sampled independently from the uniform distribution ${\mathcal{U}}_{[0,1]}$. $n^d$ is the number of adversarily placed servers (hence, $n^u + n^d =n_s$). The performance of an algorithm $\A$ is measured by its \emph{competitive ratio}:

\[\max_{S^d \in [0,1]^{n^d}} \frac{\E_{S^u,R \sim {\mathcal{U}}_{[0,1]},  \A}[\text{cost}(\mathcal{A}, (S^d \cup S^u, R))]}{ \E_{S^u,R \sim {\mathcal{U}}_{[0,1]}} [\text{cost}(OPT, (S^d \cup S^u , R))]}.\]

where $OPT$ is the offline optimal matching when the requests are known at time $t = 0$. We say that an algorithm is $\alpha$-competitive if its competitive ratio is upper bounded by $\alpha$. Although some papers in online optimization use a different notion of competitive ratio  (see, e.g., the survey \cite{Mehta}), in the context of online matching on the line, most literature we are aware of use the same definition as ours. This is true, in particular, for papers over which we build \citep{GuptaGPW19,Kanoria21} or whose results we improve \citep{Akbarpour21, TsaiTC94}.

The three models investigated in this paper can then be formalized as follows.

\begin{itemize}
    \item  In the \emph{fully random model},  $(n^u, n^d, n) = (n,0,n)$, i.e., all servers $S$ and requests $R$ are drawn uniformly and independently from $[0,1]$ and there is an equal number of servers and requests.
    \item For a constant $\epsilon>0$, we define the \emph{fully random $\epsilon-$excess model}, in which $(n^u, n^d, n) = ((1+\epsilon)n , 0, n)$, i.e., all servers $S$ and requests $R$ are drawn uniformly and independently from $[0,1]$ and there is a linear excess of $\epsilon n$ servers.
    \item In the \emph{random requests model}, $(n^u, n^d, n) =(0, n, n)$, i.e., the requests $R$ are still drawn uniformly and independently from $[0,1]$ but the servers are now chosen adversarially over all potential sequence of $n$ requests in $[0,1]$. 
\end{itemize}

 The greedy algorithm, denoted by $\G$, is the algorithm that matches each request $r_t$ to the closest available server, i.e., $s_\G(r_t) = \argmin_{s \in S_{\G, t-1}} |s - r_t|$. We assume that greedy breaks ties arbitrarily but consistently. We say that an algorithm $\A$ \emph{makes neighboring matches} if it matches every request $r_t$ either to the closest available server to its left or  to its right. 
 For any algorithm $\A$ (possibly randomized) and  $m\in \{0,\ldots, n\}$, we define the hybrid algorithm $\H_\A^m$  that matches the first $m$ requests according to $\A$ and then greedily matches the remaining requests to the closest available server. The following key lemma (proved in Appendix~\ref{app_hybrid}) bounds 
$\mathbb{E}\big[cost(\H_\A^{m-1}) - cost(\H_\A^{m})]$ as a function of $\mathbb{E}[\text{cost}_m(\A)]$ -- that is, the expected cost for algorithm $\A$ to match the $m^{th}$ request.

 \begin{restatable}{rLem}{lemhybrid}
\label{lem:hybrid}  \textbf{\textit{(The Hybrid Lemma)}.} There exists a constant $C>0$ such that for any online algorithm $\A$ that makes neighboring matches, for any instance with $n$ arbitrary servers $S = \{s_1,\ldots,s_n\}$, $n$ requests $R = (r_1,\ldots,r_n)$ uniformly and independently drawn from $[0,1]$, for any $m \in [n]$, we have
\begin{equation*}
     \mathbb{E}\big[cost(\H_\A^{m-1}) - cost(\H_\A^{m})|S_{m-1}, r_m]\leq C \cdot \mathbb{E}\left[\big(1+\log\big(\tfrac{1}{\text{cost}_m(\A)}\big)\big)\text{cost}_m(\A)\Big| S_{m-1}, r_m\right].
\end{equation*}
\end{restatable}

Note that the expectation is taken over the randomness in the requests sequence as well as any possible source of randomization in the algorithm $\mathcal{A}$. The idea of using  hybrid algorithms for analyzing online matching algorithms was used in \cite{GuptaL12}, who also introduce a hybrid lemma (see Section~\ref{sec:techoverview} for additional discussion). A key component of the proof of Lemma \ref{lem:hybrid} relies on Lemma \ref{lem:structural_hybrid_main} given below, that describes, for a fixed $m \in [n]$, the difference between the executions of $\H_{\A}^m$ and $\H_{\A}^{m-1}$ on the same sequence $R$. Lemma~\ref{lem:structural_hybrid_main} in fact shows that, at every step $t$ (i.e., just after matching request $r_t$); the free servers for both algorithms coincide, with the exception of at most one pair of servers, that we denote by $g^L_t < g^R_t$ (see Figure~\ref{fig:set_of_servers_main}) ; there is no other free server in between $g^L_t$ and $g^R_t$ ; and that strong bounds can be obtained on $\delta_t:=g^L_t-g^R_t$. These properties, in turns, will allow us to control the difference in the costs incurred by the two algorithms, eventually leading to the bound from Lemma~\ref{lem:hybrid}.

To ease the exposition, we drop the reference to the algorithms in the indices and  write $S_{t}$ and $s(r_{t})$ instead of $S_{\H_{\A}^m, t}$ and $s_{\H_{\A}^m}(r_t)$ to denote, respectively, the set of free servers for $\H_{\A}^m$ just after matching $r_t$ and the server to which $\H_{\A}^m$ matches $r_t$. Similarly, we write $S_t'$ and  $s'(r_t)$ instead of $S_{\H_{\A}^{m-1}, t}$ and $s_{\H_{\A}^{m-1}}(r_t)$ for the equivalent objects for $\H_{\A}^{m-1}$.

\begin{figure}
    \centering
    \includegraphics[scale=0.25]{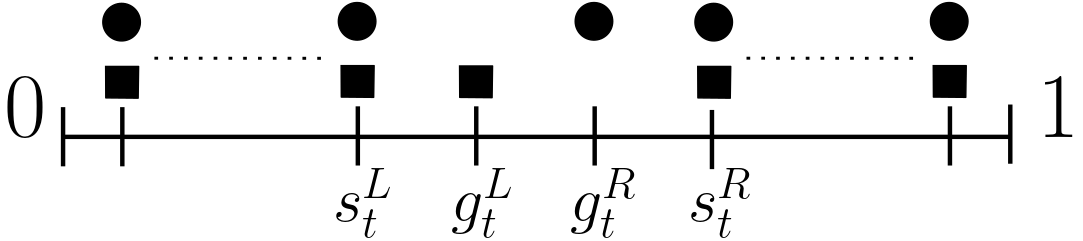}
    \caption{Set of servers $S_t$ (free servers at time $t$ for $\H_A^m$) and $S_t'$ (free servers at time $t$ for $\H_A^{m-1}$) in the case where $S_t\neq S_t'$, where the squares are the servers in  $S_t$ and the circles the servers in $S_t'$.}
    \label{fig:set_of_servers_main}
\end{figure}

 If $S_t = S_t'$, then we write $g^L_t =g^R_t = \emptyset$ and $\delta_t = 0$. We also define
       $s_t^L = \max\{s\in S_t\cup S_t'\setminus\{g_t^L, g_t^R\}: s\leq g_t^L\}$ and $s_t^R = \min\{s\in S_t\cup S_t'\setminus\{g_t^L, g_t^R\}: s\geq g_t^R\}$ (with the convention that $s_t^L=\emptyset$ if $\{S_t\cup S_t'\setminus\{g_t^L, g_t^R\} : s \leq g_t^L\}= \emptyset$ or if $g_t^L=\emptyset$, and similarly for $s_t^R$), which are the nearest servers of $S_t$ (or equivalently, of $S_t'$) on the left of $g_t^L$ and on the right of $g_t^R$.  \

\begin{restatable}{rLem}{structuralhybrid}
\label{lem:structural_hybrid_main}
Let $\A$ be any online algorithm that makes neighboring matches, $S_0$ be $n$ arbitrary servers and $R$ be $n$ arbitrary requests. Let $(S_0,\ldots, S_n)$ and $(S_0',\ldots, S_n')$ denote the set of free servers for $\H_{\A}^{m}$  and $\H_{\A}^{m-1}$ at each time steps. Then, the following propositions hold for all $t\in \{m, \ldots, n\}$:
\begin{enumerate}

    \item  \textbf{Difference in at most one server.} $|S_t\setminus S_t'|=|S_t'\setminus S_t|\leq 1$. 
        \item \textbf{Consecutiveness of the different servers.} If $g^L_t,g^R_t\neq \emptyset$, there is no server $s \in S_t \cup S_t'$ such that $g^L_t < s < g^R_t$.
        \item \textbf{Gap remains zero after disappearing.} If $\delta_t = 0$, then $\delta_{t'}=0$ for all $t'\geq t$.
        \end{enumerate}
\end{restatable}

The proof is given in Appendix~\ref{app_hybrid}. For all $t<n$ and $S_t\neq S_t'$, we also characterize the values of $s(r_{t+1}), s'(r_{t+1}), \delta_{t+1}, g^L_{t+1}, g^R_{t+1}$, and give an upper bound on $\Delta\text{cost}_{t+1} := |\text{cost}_{t+1}(\H^{m-1})- \text{cost}_{t+1}(\H^m)|$, which is key in the proof of the Hybrid Lemma (see Appendix~\ref{app_hybrid}).

%% file: section3.tex
\section{Greedy is Constant Competitive in the Fully Random Model}
\label{section:upper_bound_iid}

In this section, we show that  greedy achieves a constant competitive ratio in the fully random model where both the servers and requests are drawn uniformly and independently from $[0,1]$. In addition, we show that this result also holds when  there is a linear excess supply of servers.

\paragraph{The setting with $n$ servers.} We recall that in this setting, the competitive ratio of any algorithm $\A$ is given by:

\[
 \frac{\E_{(R,S)\sim \mathcal{U}(0,1)^n\times \mathcal{U}(0,1)^n, \A}[\text{cost}(\mathcal{A},(S,R))]}{ \E_{(R,S)\sim \mathcal{U}(0,1)^n\times \mathcal{U}(0,1)^n}[\text{cost}(OPT, (S, R))]}.
\]

\vspace{0.2cm}

The main idea of the analysis is to apply Lemma~\ref{lem:structural_hybrid_main} with $\A=\A^H$ being the hierarchical greedy algorithm from \citet{Kanoria21}. We first present the hierarchical greedy algorithm (note that \cite{Kanoria21} considers two models: a \textit{semi-dynamic} model similar to ours, and a \textit{fully-dynamic} model where the servers also arrive online. We only present here the algorithm corresponding to the \textit{semi-dynamic} model).
To describe it, we need to define the sequence $\mathcal{I}_{\ell_0}, ..., \mathcal{I}_0$, where $\ell_0 = \log(n)$, which are increasingly refined partitions of $[0,1]$.  More precisely, $\mathcal{I}_{\ell_0} = \{[0,1]\}$ and for each $\ell\leq \ell_0-1$, $\mathcal{I}_{\ell}$ is the partition obtained by dividing each interval in $\mathcal{I}_{\ell+1}$ into two intervals of equal length, i.e., $\mathcal{I}_{\ell} =\left( \cup_{[0,y] \in \mathcal{I}_{\ell+1}} \{[0, y/2], ]y/2, y]\}\right)\cup\left( \cup_{]x,y] \in \mathcal{I}_{\ell+1}} \{]x, (x+y)/2], ](x+y)/2, y]\}\right)$. The partitions obtained through this process can be organized in a binary tree, where the nodes at level $\ell$  are the intervals of $\mathcal{I}_{\ell}$  and  the leafs are the intervals of $\mathcal{I}_0$.

Given a request $r_t$, let $I(r_t)$ be the  leaf interval to which $r_t$ belongs and $J(r_t)$ be the lowest-level ancestor interval of $I(r_t)$ in the tree such that $J(r_t)\cap S_{t-1}\neq \emptyset$, i.e., such that $J(r_t)$ contains some free servers when request $r_t$ arrives. The hierarchical greedy algorithm matches $r_t$ to any free server in $J(r_t)$. For our purposes, we assume that it matches $r_t$ to the closest free server in $J(r_t)$.  A request $r_t$ is said to be \textit{matched at level $\ell$} if $J(r_t)\in \mathcal{I}_{\ell}$. There are two known results about hierarchical greedy that are important for our analysis. The first one upper bounds the number of requests matched at each level.

\begin{lemma}[\citet{Kanoria21}]
\label{lem:excess_servers} There is a constant $C'>0$ such that, for all $\ell \in \{0,\ldots, \ell_0\}$,  we have $\mathbb{E}[|\{r_t : J(r_t) \in \mathcal{I}_\ell\}|]\leq C'\sqrt{n2^{\ell-\ell_0}}2^{\ell_0-\ell}$.
\end{lemma}

The second important result about hierarchical greedy is its constant competitiveness.

\begin{theorem}[\citet{Kanoria21}]
\label{lem:HGCR}
In the fully random model, we have that $\E[\text{cost}(\A^H)] = O(\sqrt{n})$.
\end{theorem}

Next, we show the following bound on the cost incurred by hierarchical greedy when matching a request at level $\ell$.

\begin{restatable}{rLem}{factlevell}
\label{fact:2} 
For all $t\in [n]$, if $r_t$ is matched at level $\ell$, then we have $$\text{cost}_t(\A^H)\log(1/\text{cost}_t(\A^H))\leq 2^{\ell-\ell_0}(\log(2)(\ell_0-\ell) + 1).$$
\end{restatable}
\begin{proof}
Let $\ell\in \{0,\ldots,\ell_0\}$. First note that  the cost incurred by $\A^H$ when matching a request $r_t$ at level $\ell$ satisfies $\text{cost}_t(\A^H)\leq 2^{\ell-\ell_0}$ since the intervals of $\mathcal{I}_{\ell}$ have length at most $2^{\ell-\ell_0}$ by definition of $\mathcal{I}_{\ell}$. Next,  if $2^{\ell-\ell_0}\in (0,1/e]$, then  $$\text{cost}_t(\A^H)\log(1/\text{cost}_t(\A^H)) \leq 2^{\ell-\ell_0}\log(1/2^{\ell-\ell_0})$$ since $x \log(1/x)$ is non-decreasing on $(0, 1/e]$ and $\text{cost}_t(\A^H)\leq 2^{\ell-\ell_0}$. If $2^{\ell-\ell_0}\in [1/e,1]$, then $$\text{cost}_t(\A^H)\log(1/\text{cost}_t(\A^H)) \leq 1/e \leq 2^{\ell-\ell_0}$$ since $\argmax_{x \in (0,1]} x\log(1/x)  = 1/e$ and 
$\frac{1}{e}\log(\frac{1}{1/e}) = \frac{1}{e}$.  We conclude that if $r_t$ is matched at level $\ell$,
\begin{align*} \text{cost}_t(\A^H)\log(1/\text{cost}_t(\A^H))
    \leq  2^{\ell-\ell_0}\log(1/2^{\ell-\ell_0})+ 2^{\ell-\ell_0} \leq 2^{\ell-\ell_0}(\log(2)(\ell_0-\ell) + 1). & \qedhere
\end{align*}
\end{proof}

 The next lemma is the main lemma of this section and shows that the difference between the total cost of greedy and hierarchical greedy is $O(\sqrt{n})$.

\begin{lemma}
\label{lem:mainrandom}
In the fully random model, we have that $$\mathbb{E}[\text{cost}(\G) - \text{cost}(\A^H)] = O(\sqrt{n}).$$
\end{lemma}
\begin{proof} We first note that since the hierarchical greedy algorithm matches every request $r_t$  to the closest free server in $J(r_t)$, and since $r_t \in J(r_t)$ by definition of $J(r_t)$, hierarchical greedy makes neighboring matches, which is the condition needed to apply the hybrid lemma to the hybrid algorithm $\H^m$. We get that 
\begingroup
\allowdisplaybreaks
\begin{align*}
    & \mathbb{E}[cost(\mathcal{G}) - cost(\A^H)] &  \\
     = \hspace{.1cm} & \sum_{m=1}^n \mathbb{E}[cost(\H^{m-1}) - cost(\H^{m})] & \text{$\H^n = \A^H, \H^0 = \G$} \\
     \leq \hspace{.1cm} & C\sum_{m=1}^n\mathbb{E}[\big(1+\log\big(\tfrac{1}{cost_m(\A^H)}\big)\big)\text{cost}_m(\A^H)] & \text{Hybrid lemma} \\
      \leq \hspace{.1cm} &  C\sum_{m=1}^n\mathbb{E}[\log\big(\tfrac{1}{cost_m(\A^H)}\big)\text{cost}_m(\A^H)] + C \mathbb{E}[cost(\A^{H})]\\
      \leq \hspace{.1cm} &  C\sum_{m=1}^n\mathbb{E}[\log\big(\tfrac{1}{cost_m(\A^H)}\big)\text{cost}_m(\A^H)] + O(\sqrt{n}) & \text{Theorem \ref{lem:HGCR}}\\
     = \hspace{.1cm}& C\sum_{m=1}^n\sum_{\ell=0}^{\ell_0}\mathbb{P}(J(r_m)\in \mathcal{I}_{\ell})\mathbb{E}[\log\big(\tfrac{1}{cost_m(\A^H)}\big)\text{cost}_m(\A^H) | J(r_m)\in \mathcal{I}_{\ell}] + O(\sqrt{n})& \\
     \leq \hspace{.1cm} & C \sum_{m=1}^n\sum_{\ell=0}^{\ell_0}\mathbb{P}(J(r_m)\in \mathcal{I}_{\ell})\cdot 2^{\ell-\ell_0}(\log(2)(\ell_0-\ell)+1) + O(\sqrt{n})& \text{Lemma~\ref{fact:2}}  \\
     = \hspace{.1cm} & C \sum_{\ell=0}^{\ell_0} 2^{\ell-\ell_0}(\log(2)(\ell_0-\ell)+1) \cdot \sum_{m=1}^n \mathbb{P}(J(r_m)\in \mathcal{I}_{\ell})+ O(\sqrt{n}) &   \\
 = \hspace{.1cm} & C \sum_{\ell=0}^{\ell_0}2^{\ell-\ell_0}(\log(2)(\ell_0-\ell)+1) \cdot \mathbb{E}[|\{r_t : J(r_t) \in I_\ell\}|] + O(\sqrt{n})& \\
   \leq \hspace{.1cm} & CC'\sqrt{n}\sum_{\ell=0}^{\ell_0}2^{(\ell-\ell_0)/2}(\log(2)(\ell_0-\ell)+1)+ O(\sqrt{n}) & \text{Lemma~\ref{lem:excess_servers}}  \\
   = \hspace{.1cm} & CC'\sqrt{n}\sum_{j=0}^{\ell_0}2^{-j/2}(\log(2)j+1)+ O(\sqrt{n})& \\
   = \hspace{.1cm} & CC'\sqrt{n}\left(\log(2)\sum_{j=0}^{\ell_0} j \left(\frac{1}{\sqrt{2}}\right)^j + \sum_{j=0}^{\ell_0}  \left(\frac{1}{\sqrt{2}}\right)^j\right) + O(\sqrt{n})& \\
   = \hspace{.1cm} & O(\sqrt{n}). & \qedhere
\end{align*}
\endgroup
\end{proof}

 The last result needed  is that the optimal cost in the fully random model is known to be $\Theta(\sqrt{n})$.

\begin{lemma}[\cite{Kanoria21}]
\label{lem:offline}
In the fully random model, we have that $\E[\OPT] = \Theta(\sqrt{n})$.
\end{lemma}

By combining Theorem~\ref{lem:HGCR}, Lemma~\ref{lem:mainrandom}, and Lemma~\ref{lem:offline}, we obtain the main result of this section.

\thmrandombalanced*

\paragraph{The excess supply setting.} We consider here an extension of the previous model where there is a linear excess of servers. For any constant $\epsilon \in [0,1]$, we define the \textit{fully random $\epsilon$-excess model}, where an instance consist of $n$ requests and $n(1+\epsilon)$ servers all drawn uniformly and independently from $[0,1]$. The competitive ratio of any algorithm $\A$ is given by:

\[
 \frac{\E_{(R,S)\sim \mathcal{U}(0,1)^n\times \mathcal{U}(0,1)^{n(1+\epsilon)}, \A}[\text{cost}(\mathcal{A}, (S,R))]}{ \E_{(R,S)\sim \mathcal{U}(0,1)^n\times \mathcal{U}(0,1)^{n(1+\epsilon)}}[\text{cost}(OPT, (S, R))]}.
\]

\vspace{0.2cm}

In this setting, the  hybrid approach with  hierarchical greedy   used above does not give a constant competitive ratio. However, we are still able to prove that greedy is constant competitive with a different argument. Unlike the model with $n$ servers, the analysis for the excess supply setting does not rely on the hybrid lemma but on concentration arguments.  Missing proofs can be found in Appendix~\ref{app_random_model}.

The main technical contribution here lies in showing that, thanks to the excess of servers, there is an exponentially small probability that there is a large area around the $n$-th request  that contains no available servers. More formally, for $\ell, m\in [0,1]$, we let $\xlr = |\{t\in [n-1]: r_t\in (\ell,m)\}|$ be the number of requests out of the first $n-1$ that arrived in the interval $(\ell,m)$, and we let $\ylr = |\{t\in [n(1+\epsilon)]: s_t\in (\ell,m)\}|$ be the total number of servers that lie in the interval $(\ell,m)$. Then, the following lemma holds.

\begin{restatable}{rLem}{lemdiscretization}
\label{lem:discretization-main}
Let $\epsilon \in [0,1]$ be a constant. There are constants $C_\epsilon, C'_\epsilon$ such that, in the \textit{fully random $\epsilon$-excess} model, we have that  for all $z\in [ \tfrac{4+\epsilon}{\epsilon n},1]$,     
$$\P(\exists \ell, m\in [0,1] : x_{(\ell,m)} = y_{(\ell,m)}, ( r_n- \ell\geq z \text{ or  } \ell=0), (m - r_n\geq z\text{ or } m=1 )\;|\;r_n)\leq C_{\epsilon}' e^{- nzC_\epsilon}.$$
\end{restatable}

The proof is deferred to Appendix~\ref{app_random_model}. Using Lemma~\ref{lem:discretization-main}, we then upper bound the expected cost incurred by greedy at the last step. At a high level, we use in the proof that the free servers at each time step act as ``natural barriers" between different areas of the interval $[0,1]$ (in the sense that if there is a free server at location $x\in [0,1]$, no request arriving in $[0,x]$ can be matched to a server in $(x,1]$, and vice-versa). This allows to quantify precisely the total number of remaining servers in each of those areas. Note that in \citep{Kanoria21}, the analysis also relies on a division of space into distinct  regions, and on a quantification of remaining servers and requests in each region. However, in \citep{Kanoria21}, the division is fixed at the beginning of the time horizon (through the partition $\mathcal{I}_{\ell_0}, ..., \mathcal{I}_0$). The additional difficulty in our setting is that the ``barriers" we consider depend on all previously arrived requests and are thus random.

\begin{restatable}{rLem}{lemegreedynsmall}
\label{lem:e-greedy-n-small}
Let $\epsilon \in [0,1]$ be a constant. There is a constant $ C''_\epsilon$ such that, in the \textit{fully random $\epsilon$-excess} model, we have 
$\E[cost_{n}(\G)] \leq \frac{C_{\epsilon}''}{n}$.
\end{restatable}
\begin{proof}
To exclude any ambiguity, we  condition on the event that all servers are distinct and that no server or requests are at positions $0$ and $1$, which occurs almost surely. In the remainder of the proof, we condition on the variable $r_n$ and let $s_n^L =\max\{s\in S_{n-1}: s\leq r_n\}$ and $s^R_n = \min\{s\in S_{n-1}: s\geq r_n\}$ denote the nearest available servers on the left and on the right of $r_n$ when $r_n$ arrives; with the convention that $s_n^L=0$ and $s_n^L= 1$ if there are no such servers.

Now, let $z\in [ \tfrac{4(1+\epsilon/4)}{\epsilon n},1]$ and assume that $\text{cost}_n(\mathcal{G})\geq z$. Since $\mathcal{G}$ matches $r_n$ to the closest available server, we must have $r_n- s_n^L\geq z$ or $s_n^L=0$, and $s_n^R - r_n\geq z$ or $s_n^R=1$. In addition, by definition of $s_n^L$ and $s_n^R$, we have that $(s_n^L, s_n^R)\cap S_{n-1} = \emptyset$. Now, recall that all requests $r_1, \ldots, r_{n-1}$ have been matched each time to the closest available server. Moreover, for all $j\in [n-1]$, $s_n^L$ was either  available when $r_j$ arrives, but $r_j$ was not matched to it, or $s_n^L=0$; similarly for $s_n^R$.    
Hence, if $r_j\notin (s_n^L, s_n^R)$, then $s_{\G}(r_j)\notin (s_n^L, s_n^R)$. Similarly, if $r_j\in (s_n^L, s_n^R)$, then $s_{\G}(r_j)\in (s_n^L, s_n^R)$. Therefore,
\begin{equation*}
|\{ j\in [n-1] : s_{\G}(r_j)\in (s_n^L, s_n^R)\}| = |\{ j\in [n-1] : r_j\in (s_n^L, s_n^R)\}|.
\end{equation*}
In addition, since $(s_n^L, s_n^R)\cap S_{n-1} = \emptyset$, all servers in $(s_n^L, s_n^R)\cap S_0$ must have been matched to some request before time $n-1$, hence
\[
|\{ j\in [n-1] : s_{\G}(r_j)\in (s_n^L, s_n^R)\}| = |\{j\in [n(1+\epsilon)]: s_j\in (s_n^L, s_n^R)\}|.
\]

By combining the two previous equalities and by definition of $x_{(s_n^L, s_n^R)}$, and $y_{(s_n^L, s_n^R)}$, we get that 
\[x_{(s_n^L, s_n^R)} = |\{ j\in [n-1] : r_j\in (s_n^L, s_n^R)\}| = |\{j\in [n(1+\epsilon)]: s_j\in (s_n^L, s_n^R)\}| = y_{(s_n^L, s_n^R)}.\]

Since we have that $r_n- s_n^L\geq z$ or $s_n^L=0$, and $s_n^R - r_n\geq z$ or $s_n^R=1$, we thus have that
\begin{align}
\label{eq:cgcost_main}
   \P(\text{cost}_n&(\mathcal{G})\geq z\;|\;r_n)\nonumber\\
   &\leq \P(\exists \ell, m\in [0,1] : x_{(\ell,m)} = y_{(\ell,m)}, ( r_n- \ell\geq z \text{ or  } \ell=0), (m - r_n\geq z\text{ or } m=1 )\;|\;r_n). 
\end{align}
Now, by Lemma \ref{lem:discretization-main}, we have that for some constants $C_{\epsilon}, C_{\epsilon}'>0$:
\[
\P(\exists \ell, m\in [0,1] : x_{(\ell,m)} = y_{(\ell,m)}, ( r_n- \ell\geq z \text{ or  } \ell=0), (m - r_n\geq z\text{ or } m=1 )\;|\;r_n)\leq C_{\epsilon}' e^{- nzC_\epsilon}.
\]
Combining this with (\ref{eq:cgcost_main}) and by the law of total probability, we get $
\P(\text{cost}_n(\mathcal{G})\geq z)\leq C_{\epsilon}' e^{- nzC_\epsilon}
$. Hence, we obtain
\begin{align*}
    \label{eq:cost_exp}
        \E[cost_{n}(\G)] &\leq \tfrac{4(1+\epsilon/4)}{\epsilon n} +  \int_{z=\tfrac{4(1+\epsilon/4)}{\epsilon n}}^{1} \mathbb{P}(cost_{n}(\G)\geq z)\mathrm{d}z  & &\\
        &\leq \tfrac{4(1+\epsilon/4)}{\epsilon n} +  \int_{z=\tfrac{4(1+\epsilon/4)}{\epsilon n}}^{1} C_{\epsilon}' e^{- nzC_\epsilon}\mathrm{d}z   & &\\
        &\leq \tfrac{4(1+\epsilon/4)}{\epsilon n} +  \frac{C_{\epsilon}'}{C_{\epsilon}n}. &\text{for some $C_{\epsilon}>0$}  & \\
        &= \frac{C_{\epsilon}''}{n}&\text{for some $C_{\epsilon}''>0$.} & \qedhere
    \end{align*}
\end{proof}

 We underscore that a simple application of Chernoff bounds between all initial pairs of servers locations would only lead to a weaker version of the above lemma, involving poly-logarithmic terms. Since our objective was to present a sharp analysis of greedy, we introduced the refined analysis above.
 
 Last, we observe that, because of servers getting less and less dense as requests arrive, the expected cost of the greedy algorithm increases at each step. 

\begin{restatable}{rLem}{lemgreedycostincreases}
\label{lem:greedy-cost-increases}
Let $\epsilon \in [0,1]$ be a constant. Then, in the \textit{fully random $\epsilon$-excess} model, we have that for all $i \in [n-1]$,  $\E[cost_{i}(\G)]\leq \E[cost_{i+1}(\G)]$.
\end{restatable}

Using Lemma~\ref{lem:e-greedy-n-small} and Lemma~\ref{lem:greedy-cost-increases}, we conclude that $\E[cost(\G)] =\sum_{i=1}^n \E[cost_{i}(\G)] \leq n\cdot\E[cost_{n}(\G)] \leq C_{\epsilon}''.$ We have thus shown the following.

\begin{restatable}{rLem}{lemgreedyunbalanced}
\label{thm:greedy_unbalanced} 
Let $\epsilon \in [0,1]$ be a constant. There exists a constant $C''_{\epsilon}>0$ such that in the \textit{fully random $\epsilon$-excess model}, we have
$\mathbb{E}[\text{cost}(\G)] \leq C''_{\epsilon}$.
\end{restatable}

In order to conclude the proof of Theorem~\ref{thm:random_unbalanced}, it suffices to lower bound the cost of the optimal solution in the \textit{fully random $\epsilon$-excess model}.

\begin{lemma}[\cite{Kanoria21}]
\label{lem:offline_excess}
For  any constant $\epsilon \ in [0,1]$, we have that in the \textit{fully random $\epsilon$-excess model}, $\E[\OPT] = \Theta(\frac{1}{\epsilon})$.
\end{lemma}

We can then conclude the following result on the performance of the greedy algorithm. 

\thmrandomunbalanced*

%% file: section5.tex



\section{Greedy is $O(\log n)$-competitive in the Random Requests Model}
\label{section:sim_random}

In Section~\ref{section:sim_random} and Section~\ref{sec_semi_lower}, we show that greedy achieves an $\Theta(\log n)$ competitive ratio in the random requests model where  the servers are chosen adversarially  and the requests are drawn uniformly and independently from $[0,1]$. Thus, unlike in the fully random model,  servers and requests can be distributed in a significantly different manner in this model. 

In this section, we first show the $O(\log n)$ upper bound. We note that, even though hierarchical greedy and greedy are both constant-competitive in the fully random model, hierarchical greedy is only $\Omega(n^{1/4})$-competitive in the random requests model (see Appendix~\ref{app_HG}). 
The main lemma (Lemma~\ref{lem:bound_hybrid}) shows that greedy is at most a logarithmic factor away from any online algorithm that makes neighboring matches. To prove Lemma~\ref{lem:bound_hybrid}, we first need to lower bound the probability that the cost incurred by any online algorithm at any time step is  small. First, for all $t\in [n]$, we let $\mathcal{G}(S_{t-1}) = \{r\in [0,1]: \exists s\in S_{t-1}, |r-s|<\frac{1}{n^4}\}$ be the set of points in $[0,1]$ that are close to servers in $S_{t-1}$.



\begin{restatable}{rLem}{lemunifcost}
\label{lem:unif_cost}
In the random requests model, for any online algorithm $\mathcal{A}$ and for all $t\in [n]$, we have that
$\E[\text{cost}_t(\mathcal{A})] \geq \frac{1}{2(n+1)}$ and that $\P(r_t\in \mathcal{G}(S_{t-1}))\leq \frac{2}{n^3}$.
\end{restatable}

The proof is in Appendix~\ref{app_semi_random_UB}. Next, to show that Lemma~\ref{lem:bound_hybrid} holds for any online algorithm $\A$ that makes neighboring matches, we use the hybrid lemma on the hybrid algorithm $\H_\A^m$ (and we abuse notation by writing $\H^m$).

\begin{lemma}
\label{lem:bound_hybrid}   In the random requests model, there exists a constant $C>0$ such that for any online algorithm $\mathcal{A}$ that makes neighboring matches,  $$\mathbb{E}[cost(\mathcal{G})] \leq C\log(n)\mathbb{E}[cost(\mathcal{A})].$$
\end{lemma}
\begin{proof} We write 
\begin{align*}
    & \E[cost(\H^{m-1}) - cost(\H^m)]& \\
    = \ & \E[cost(\H^{m-1}) - cost(\H^m)\;|\;r_m \notin \mathcal{G}(S_{m-1})]\cdot \P(r_m \notin \mathcal{G}(S_{m-1})) & \\
    & \qquad + \E[cost(\H^{m-1}) - cost(\H^m)\;|\;r_m \in \mathcal{G}(S_{m-1})]\cdot \P(r_m \in \mathcal{G}(S_{m-1})) \\
    \leq \ &  \E[cost(\H^{m-1}) - cost(\H^m)\;|\;r_m \notin \mathcal{G}(S_{m-1})] \cdot \P(r_m \notin \mathcal{G}(S_{m-1}))
    + n\cdot 2/n^3  \\
    \leq \ &  C\cdot \mathbb{E}[\big(1+\log\big(\tfrac{1}{cost_m(\mathcal{A})}\big)\big)\text{cost}_m(\mathcal{A})\;|\;r_m \notin \mathcal{G}(S_{m-1})] \cdot \P(r_m \notin \mathcal{G}(S_{m-1}))
    + 2n^{-2} \\
    \leq \ & C (1+4
\log(n))\cdot \mathbb{E}[\text{cost}_m(\mathcal{A})\;|\;r_m \notin \mathcal{G}(S_{m-1})] \cdot \P(r_m \notin \mathcal{G}(S_{m-1}))
    + 2n^{-2} & \\
    \leq \ &  C(1+4\log(n)) \cdot \E[\text{cost}_m(\mathcal{A})] + 2n^{-2} & \\
    = \ &  C' \log(n)  \cdot \E[\text{cost}_m(\mathcal{A})], &
\end{align*}
where the first inequality is by Lemma~\ref{lem:unif_cost}, the second one by the Hybrid Lemma (Lemma~\ref{lem:hybrid} ; noting that $\{r_m \notin \mathcal{G}(S_{m-1})\}$ is an event that depends only on $S_{m-1}$ and $r_m$ and that $\A$ makes neighboring matches) and the third one is since for any algorithm $\A$, $\text{cost}_m(\mathcal{A})\geq 1/n^4$ when $r_m \notin \mathcal{G}(S_{m-1})$. The last equality is by Lemma~\ref{lem:unif_cost}.

Since $\H^n = \mathcal{A}$ and $\H^0 = \G$, we conclude that 
\begin{align*}
  \mathbb{E}[cost(\mathcal{G}) - cost(\mathcal{A})] & =  
    \sum_{m=1}^n \mathbb{E}[cost(\H^{m-1}) - cost(\H^m)] \\
    & \leq C'\log(n)\sum_{m=1}^n\mathbb{E}[cost_m(\mathcal{A})] \\
    & = C'\log(n)\cdot \mathbb{E}[cost(\mathcal{A})].  
\end{align*}
\end{proof}

It remains to show the existence of a constant competitive online algorithm that makes neighboring matches in the random requests model, which is the case for  a simple modification of the algorithm fair-bias from \cite{GuptaGPW19}. The proof is deferred to Appendix~\ref{app_semi_random_UB}.

\begin{restatable}{rLem}{corggp}
\label{cor:GGP}
In the random requests model, there exists a constant competitive algorithm that makes neighboring matches.
\end{restatable}

We are now ready to prove the main result of Section~\ref{section:sim_random}.

\thmsemirandomup*

\begin{proof} By Lemma \ref{cor:GGP}, there exists an algorithm $\mathcal{A}$ that is constant competitive algorithm in the random requests model and makes neighboring matches. We have, by Lemma \ref{lem:bound_hybrid}, that $\mathbb{E}[cost(\mathcal{G})] \leq C\log(n)\mathbb{E}[cost(\mathcal{A})]$. We conclude that greedy is $O(\log{n})$-competitive.
\end{proof}

\begin{figure}
    \centering
    \includegraphics[scale =0.25]{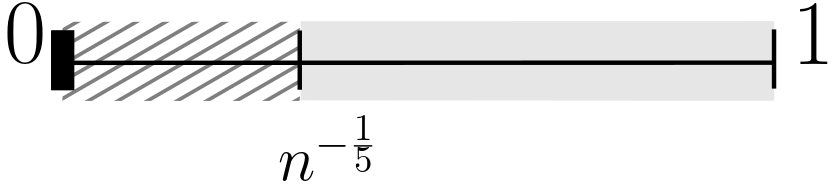}
    \caption{The lower bound instance. There are $n^{4/5}+4 \log^2(n) \sqrt{n}$ servers at $0$, no server in the dashed area, and $n-(n^{\frac{4}{5}}+4\log^2(n)\sqrt{n})$ servers uniformly distributed  in the gray area.}
    \label{fig:lb_instance}
\end{figure}

\section{Greedy is $\Omega(\log{n})$-competitive in the Random request Model: Overview of the Proof}
\label{sec_semi_lower}

The $\Omega(\log n)$ lower bound is the main technical  contribution of this paper.  The main steps of our proofs are as follows:
\begin{itemize}
\item We create an instance with a mass of servers at $0$, no server in $(0,n^{-1/5}]$ (dashed area), and uniformly distributed servers in $(n^{-1/5}]$ (gray area). See Section~\ref{sec:lb-instance-description} and Figure~\ref{fig:lb_instance}.
\item When applied to this instance, with high probability the greedy algorithm matches some demands in $(0,n^{-1/5}]$ (dashed area) to servers in the gray area. This is suboptimal, since the expected total number of requests in $[0,n^{-1/5}]$ is $n^{-1/5}\cdot n = n^{4/5}$, which is less than the number of servers at position $0$. Thus, we define algorithm $\mathcal{A}$ that, for all $t\in [n]$, matches $r_t$ to a free server at location $0$ if $r_t\in [0,n^{-1/5}]$ and $S_{\mathcal{A},t-1}\cap\{0\}\neq \emptyset$, and, otherwise, matches $r_t$ greedily. See Section~\ref{sec:lb-analysis-instance}.
\item Recall that a main building block for the upper bound on the competitive ratio of greedy is the Hybrid Lemma, Lemma~\ref{lem:hybrid}, giving an \emph{upper bound} on the expected value of the difference between the cost of the hybrid algorithms $\H_{\A}^{m-1}$ and  $\H_{\A}^m$ (recall that $\H_{\A}^m$ is defined as the algorithm that  matches the first $m$
requests according to ${\cal A}$ and then greedily matches the remaining requests to the closest available
server). Here, we want to \emph{lower bound}  the same quantity, since our final goal is to lower bound the competitive ratio of greedy. We start by lower bounding the difference of the related quantities $\sum_{t= m+1}^n cost_t(\H^{m-1}_{\cal A})$ and $\sum_{t= m+1}^n cost_t(\H^m_{\cal A})$, see Lemma~\ref{cor:costLdelta_gamma}. A detailed lower bound on $\mathbb{E}[cost(\H^{m-1}) - cost(\H^m)]$ is then given in Lemma~\ref{lem:m_cn2}, which can therefore be interpreted as a counterpart (for the specific algorithm ${\cal A}$ described above) of the Hybrid Lemma, Lemma~\ref{lem:hybrid}. See again Section~\ref{sec:lb-analysis-instance}.
\item Lemma~\ref{lem:m_cn2}, together with an upper bound on the expected cost of the offline optimum (Lemma~\ref{lem:m_cn2}) implies the claimed $\Omega(\log n)$ bound on the competitive ratio of greedy. See Section~\ref{sec:main_lb_result}.
\end{itemize}

Below we give some details on each of the steps above. The complete analysis and proofs of all lemmas can be found in Appendix~\ref{appendix:full_LB}.

\subsection{Description of the instance}\label{sec:lb-instance-description}

  We formalize here the description of the instance from Figure~\ref{fig:lb_instance}. We define the set of $n$ servers $S_0$ as follows: for all $j\in [n^{4/5} + \excess]$, we set $s_j= 0$. Then, we let $\tilde{n}:= n - \excess/(1-n^{-1/5})$, and for all $j\in [ n-(n^{4/5}-\excess)]$, we set  $s_{(n^{4/5} + \excess)+j} = n^{-1/5} + \frac{j}{\tilde{n}}$.

 We note that, interestingly, the  servers are almost uniform since a $1-o(1)$ fraction of the servers are uniformly spread in an interval $(o(1), 1]$.

\subsection{Analysis of the instance}\label{sec:lb-analysis-instance}

 Recall that for all $t\in [n]$, algorithm ${\cal A}$ matches  $r_t$ to a free server at location $0$ if $r_t\in [0,n^{-1/5}]$ and $S_{\mathcal{A},t-1}\cap\{0\}\neq \emptyset$, and, otherwise, matches $r_t$ greedily.  The main part of the proof is to lower bound $\mathbb{E}[cost(\H_{\A}^{m-1}) - cost(\H_{\A}^m)]$, i.e., the increase in cost from switching from algorithm $\A$ to the greedy algorithm $\G$ one step earlier in hybrid algorithm $\H_\A^{m-1}$ compared to $\H_\A^m$. As we will show, matching a request in $[0,n^{-1/5}]$ greedily at time $t = m$ instead of matching it to a server at location $0$ causes a cascading increase in costs at future time steps for $\H_{\A}^{m-1}$ compared to $\H_{\A}^m$ due to the different available servers, even though these two algorithms both match requests greedily at time steps $t > m$.

\paragraph{Structural properties.} The first lemma shows that at every time step $t$, there are at most two servers in the  symmetric difference between the sets of free servers $S_{\H_{\A}^m, t}$ and  $S_{\H_{\A}^{m-1},t}$, and that the potential extra free server in $S_{\H_{\A}^{m-1},t}$ is always located at $0$ whereas the potential extra free server in $S_{\H_{\A}^m, t}$ is the leftmost free server that is not at location $0$ (see Figure~\ref{fig:config_LB}). To ease  notation, we  write $\H^{m}$ and $\H^{m-1}$ instead of $\H_{\A}^{m}$ and $\H_{\A}^{m-1}$ and $S_t$ and $S_t'$ instead of $S_{\H^m,t}$ and $S_{\H^{m-1},t}$.

\begin{restatable}{rLem}{lemconfigSS}
\label{lem:configSS'} 
For any arbitrary sequence $R$ of $n$ requests, we have that
for all $t\in \{0,\ldots, m-1\}$, $S_t = S_t'$, and that for all $t\geq m$, either $S_t =  S_t'$ or   $ S_t' = S_t\cup \{0\}\setminus\{\min\{s\in S_t: s>0\}\}$.
\end{restatable}

\begin{figure}
    \centering
    \includegraphics[scale=0.33]{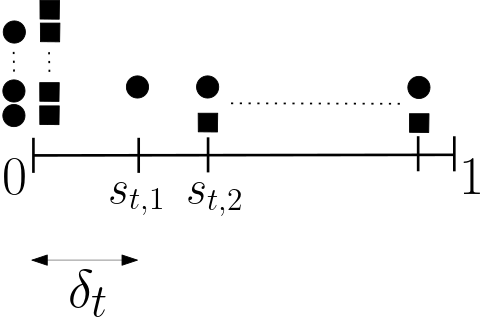}
    \caption{Sets of free servers for $\H^m$ and $\H^{m-1}$ at all time steps (with the circles denoting servers in $S_t$ and the squares denoting servers in  $S_t'$).}
    \label{fig:config_LB}
\end{figure}

\paragraph{Lower bounding the cost by the maximum gap $\delta_t$.} To bound $\mathbb{E}[cost(\H^{m-1}) - cost(\H^m)]$, we analyze the gap $\delta_t := \min\{s \in S_t : s > 0\}$ between the unique available server in $S'_t \setminus S_t \subseteq \{0\} $ and the unique available server in $S_t \setminus S_t' \subseteq \{\min\{s \in S_t : s > 0\}\}$. If $S_t = S_t'$, then there is no gap and we define $\delta_t=0$. The next lemma formally bounds $\mathbb{E}[ \sum_{t= m+1}^n (cost_t(\H^{m-1}) - cost_t(\H^m)|\delta_m,S_m]$ as a function of the gap $\delta_t$.

\begin{restatable}{rLem}{lemcostdeltagamma}
\label{cor:costLdelta_gamma} For all $m \in [n]$, we have that
\begin{align*}
    \mathbb{E}\Big[ \sum_{t= m+1}^n (cost_t(\H^{m-1}) - cost_t(\H^m)|\delta_m,S_m\Big]\geq \frac{1}{2}\mathbb{E}&\Big[\max_{t \in \{0,\ldots, \min(t_{\{0\}}, t_w)-m\}}\delta_{t+m}- \delta_{m} |\delta_m,S_m\Big]\\
    & - \P(t^d>t_{\{0\}}|\delta_m,S_m), 
\end{align*}
where $s_{t,1} := \min\{s>0:s\in S_{t}\}$ and $s_{t,2} := \min\{s>s_{t,1}:s\in S_{t}\}$; $t_w := \min\{t\geq m: s_{t,2} - s_{t,1}> s_{t,1}, \text{ or } s_{t,2}= \emptyset\}$, $t^d = \min \{t\geq m: \delta_t = 0\}$ and $t_{\{0\}}:= \min \{t\geq m|\;S_{t}\cap \{0\} = \emptyset\}$.
\end{restatable}

\paragraph{Lower bounding the maximum gap $\delta_t$.} By Lemma~\ref{cor:costLdelta_gamma}, it remains to lower bound the maximum gap $\delta_t$, for $t \geq m$. To analyze this gap, we first need to introduce some additional notation and terminology.  We consider a partition $I_0, I_1, \ldots$ of $(0,1]$ into intervals of geometrically increasing size, where $I_i = (y_{i-1}, y_i]$ and $y_i = (3/2)^{i}n^{-1/5}$ (with the convention $y_{-1}=0$).
     In addition, we say that a sequence of requests is \textit{regular} if, for any $i \in [n]$, the number of requests between any time steps $t$ and $t'$ that are in the interval $[(i-1)/n, i/n]$ "sufficiently concentrates". More formally, we start by discretizing the interval $[0,1]$ as $\mathcal{D} = \{\tfrac{i}{n}: i\in \{0,\ldots, n\}\}$.

 \begin{definition}
\label{def:reg_main}
We say that a realization $R$ of the sequence of requests is regular if for all $d,d'\in \mathcal{D}$ such that $d<d'$, and for all $t,t' \in [n]$ such that $t<t'$,
\begin{enumerate}
    \item $|\{j\in \{t,\ldots, t'\}|\;r_j \in [d,d']\}|\geq (d'-d)(t'-t) - \log(n)^2\sqrt{(d'-d)(t'-t)}$,
    \item and if $(d'-d)(t'-t) = \Omega(1)$, then \[|\{j\in \{t,\ldots, t'\}|\;r_j \in [d,d']\}|\leq (d'-d)(t'-t) + \log(n)^2\sqrt{(d'-d)(t'-t)}.\]
\end{enumerate}
\end{definition}
     
     By standard concentration bounds, a sequence of requests is regular with high probability.

\begin{restatable}{rLem}{lemreg}
\label{lem:reg}
With probability at least $1-\hp$, the  sequence of requests is regular.
\end{restatable}


Once the requests of sequence is assumed regular, all  events that can be derived by successive applications of simple Chernoff bounds become deterministic events. In particular, when a sequence of requests is regular, we can bound, for algorithm $\H^m$, the gap $s_{t,j+1} - s_{t,j}$ between the $j^{th}$ and $j+1^{th}$ free servers $s_{t,j}$ and $s_{t,j+1}$ with positive location at time $t \in [(1-o(1)) n]$. 

The main technical lemma of the proof of the $\Omega(\log(n))$-competitive ratio is to lower bound the maximum gap $\delta_t$ over all $t \geq m$, which we do in the next lemma, where $c_1,d_1,c_3$ are positive constants.

\begin{figure}
    \centering
    \includegraphics[scale=0.21]{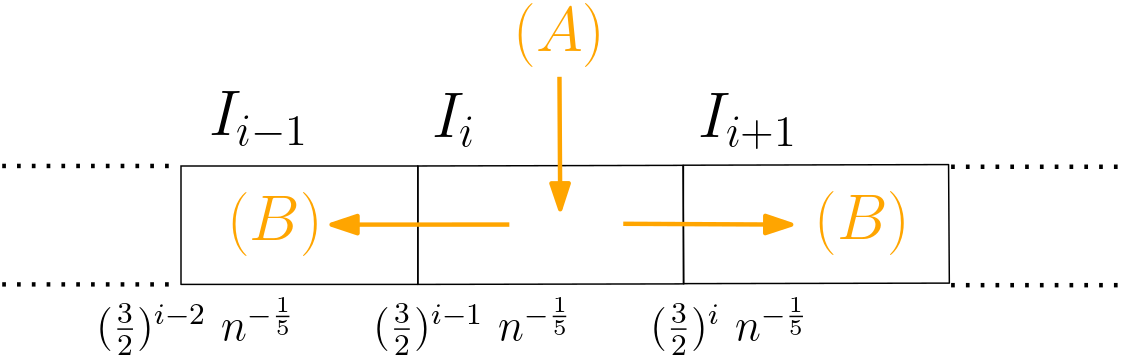}\\
    \includegraphics[scale=0.25]{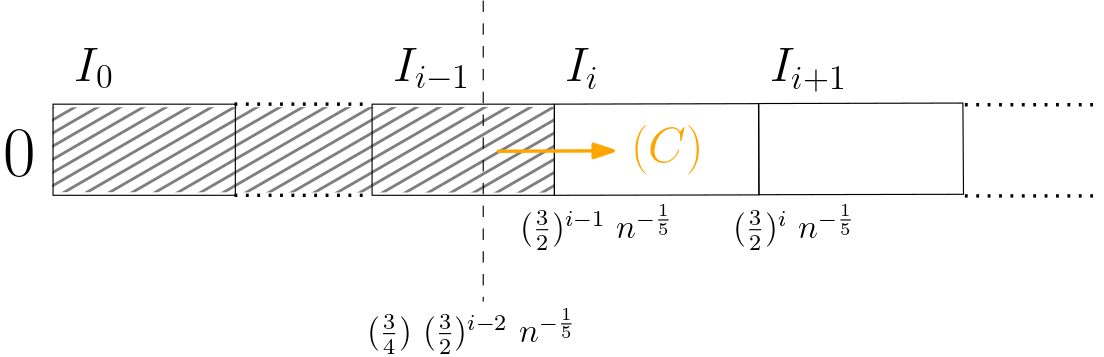}
    \caption{Requests in and out of $I_i$ up to time $\overline{t}_i:= \min(t_i, t_{i-1} + c_2(n-t_{i-1}))$, with (A) the total number of requests that arrived in $I_i$ from time 0 to $\overline{t}_i$, (B) the total number of requests that arrived in $I_i$ and were matched outside $I_i$ from time $0$ to $\overline{t}_i$, and $(C)$ the total number of requests that arrived in $[\tfrac{3}{4}y_{i-1}, y_{i-1}]$ and were matched inside $I_i$ from time $t_{i-1} + 1 +c_1(n - t_{i-1})$ to time $\overline{t}_i$ (note that there are no free servers in the dashed area for times $t\geq t_{i-1}$).}
    \label{fig:flow_main}
\end{figure}

\begin{restatable}{rLem}{lemPdeltamax}
\label{lem:P_delta_max} 
For all $i\in [d_1\log(n)]$ and $m\leq c_1n$,
\begin{equation*}
    \P\Big(\max_{t\in \{m,\ldots,  \min(n-n^{c_3}, t_{\{0\}})\}} \delta_{t} \geq y_{i-1}|\reg, \delta_m, S_m\Big)\geq \frac{\delta_m}{y_{i}}-n^{-\Omega(\log(n))}.
\end{equation*}
\end{restatable}

\paragraph{Challenges to prove Lemma~\ref{lem:P_delta_max}.} The main difficulty in proving Lemma \ref{lem:P_delta_max} is that the value of $\delta_t$ at each time step $t$ is  dependent on the value of $S_t$. However, $S_t$ lies in an exponentially-sized state space and it is difficult to compute the exact distribution of $S_t$ at all time steps. The key idea is to  separate the analysis of $(\delta_1,\ldots, \delta_n)$ and $(S_1,\ldots, S_n)$. We first show that with high probability, the servers in $(S_1,\ldots, S_n)$ become globally unavailable from left to right (see below an overview of the proof for a more precise statement). Then, we lower bound the probability that for any $y$ and any arbitrary sequence of sets $(S_1\supseteq\ldots \supseteq S_n)$,  $\delta = 0$ before all servers in the interval $(0,y]$ have become unavailable. Combining these two properties leads to the desired result.  

\paragraph{Overview of the proof of Lemma~\ref{lem:P_delta_max}.} The proof consists of three main parts. The first one analyzes the sets of free servers $S_0\supseteq \ldots\supseteq S_n$ obtained with algorithm $\H^m$ at each time step, the second one partially characterizes the values of $(\delta_t, S_t)$ and studies the first time $t\geq m$ such that $\delta_t = 0$. The last ones combines the first two parts.

\paragraph{Part 1 of the proof of Lemma~\ref{lem:P_delta_max}.} We say that an interval $I$  is \textit{depleted} at time $t$ if $S_t \cap I = \emptyset.$ We let $t_{I} := \min \{t\geq 0|  S_t \cap I= \emptyset\}$, i.e., $t_I$ is the time at which $I$ is depleted. For simplicity, we write $t_i$ instead of $t_{I_i}$.   We first show that (A) there exists a constant $c_2 \in (1/2,1)$ such that if $t_{i-1}\leq n - (1-c_2)^{i-1}n$, then, $t_{i-1}< t_{i}$. 
Then, we show that (B) if $t_0<\ldots< t_{i-1}\leq n- (1-c_2)^{i-1}n$ and  $t_{i-1}<t_{i}$, then, $t_i\leq n- (1-c_2)^in$. To show this last result, we lower bound the number of requests matched in $I_i$ until time $\overline{t_i} = \min (t_i,t_{i-1} + c_2(n-t_{i-1}))$. We first show (see Figure~\ref{fig:flow_main}) that 
\begin{align*}
|\{j\in [\overline{t_i}]\;| s_{\H^m}(r_j) \in I_i\}| & \geq \Big[|\{j\in [\overline{t_i}]:r_j\in I_i\}| - |\{j\in [\overline{t_i}]:r_j\in I_i, s_{\H^m}(r_j)\notin I_i\}|\Big]\\
    & +|\{j\in \{t_{i-1} + 1 +c_1(n - t_{i-1}),\ldots ,\overline{t_i}\} :r_j\in [\tfrac{3}{4}y_{i-1}, y_{i-1}], s_{\H^m}(r_j)\in I_i\}|.
\end{align*}

We then lower bound each of these terms separately, using in particular the regularity of the requests sequence. We deduce from this lower bound that if $t_i> n- (1-c_2)^in$, then the number of requests matched in $I_i$ exceeds the initial number of free servers in $I_i$, which is a contradiction. Hence the bound $t_i\leq n- (1-c_2)^in$. Finally, by combining properties (A) and (B), we show inductively that there is a constant $d_1>0$ such that the intervals $\{I_i\}_{i\in [d_1\log(n)]}$ are depleted in increasing order, i.e. that $m< t_{1}<\ldots< t_{d_1\log(n)}\leq n-n^{c_3}$ and that $m<t_{\{0\}}$, which is the main result of this first part.

\paragraph{Part 2 of the proof of Lemma~\ref{lem:P_delta_max}.} We start by a partial characterization of the value of $(\delta_t, S_t)$ and of the difference of cost $\Delta \text{cost}_{t+1}:= \text{cost}_{t+1}(\H^{m-1}) - \text{cost}_{t+1}(\H^m)$  between the costs incurred by $\H^{m-1}$ and $\H^{m}$ at time step $t$ as a function of $\delta_t$ and  $S_t$.

\begin{restatable}{rLem}{lemMCdef}
\label{lem:MC_def}
All the following properties hold at any time $t\in \{m,\ldots, n-1\}$:

\begin{enumerate}
    \item if $\delta_t= 0$, then for all $t'\geq t$, we have $\delta_{t'} = 0$ and $\Delta\text{cost}_{t+1} = 0$,
    \item if $S_t\cap \{0\}\neq \emptyset$, then $\Delta\text{cost}_{t+1}\geq 0$.
    \item if $S_t\cap \{0\}\neq \emptyset$, $\delta_t\neq 0$ and $|S_t\cap (\delta_t,1]|\geq 1$, then the values of $(\delta_{t+1},S_{t+1})$ and the expected value of  $\Delta\text{cost}_{t+1}$ conditioning on $(\delta_t, S_t)$ and on $r_{t+1}$ are as given in Table \ref{tab:MC_gamma_S_main}, where $w_t:=s_{t,2}-s_{t,1}$ and where we write $\E[\Delta \text{cost}_{t+1}|...]$ instead of $\E[\Delta \text{cost}_{t+1}|(\delta_t,S_t), S_t\cap \{0\}\neq \emptyset$, $\delta_t\neq 0$, $|S_t\cap (\delta_t,1]|\geq 1, r_{t+1}\in \ldots]$.
    \item if $\delta_{t+1}\neq \delta_t$, then $S_{t+1} = S_t\setminus\{\delta_t\}$.
    \item $\E[\mathbf{1}_{S_t\cap \{0\}=\emptyset, \delta_t\neq 0}\cdot\Delta\text{cost}_{t+1}|(\delta_t, S_t)] \geq-\mathbf{1}_{S_t\cap \{0\}=\emptyset, \delta_t\neq 0}\cdot\P(\delta_{t+1}=0|(\delta_t, S_t))$.
\end{enumerate}
\end{restatable}

\begin{center}
\begin{table}[]
    \centering
    \begin{tabular}{c||c| c| c| c| c} 

$r_{t+1}\in \ldots$ & $[0,\tfrac{\delta_t}{2}]$ &  $[\tfrac{\delta_t}{2}, \tfrac{\delta_t+w_t}{2}]$ &   \small{$[\tfrac{\delta_t+w_t}{2}, \delta_t +\tfrac{w_t}{2}]$} &\small{$[\delta_t +\tfrac{w_t}{2}, \delta_t+w_t]$} & \small{$[\delta_t+w_t,1]$}\\
 \hline
 $S_{t+1}$ & $S_t\setminus \{0\}$  &   $S_t\setminus \{\delta_t\}$ &$S_t\setminus \{\delta_t\}$ & $S_t\setminus \{\delta_t+w_t\}$ & \small{$\exists s\in [\delta_t+w_t,$ }\\
  &   &  & & &\small{$1]\cap S_t: S_t\setminus \{s\}$}\\ 
 \hline
 $\delta_{t+1}$ &  $\delta_{t}$ & $0$  &  $\delta_t+w_t$ & $\delta_t$ & $\delta_t$\\
 \hline
 $\mathbb{E}[\Delta\text{cost}_{t+1}|\ldots]$ &  $\geq 0$ & $\geq 0$ & $\geq
\begin{cases}
\frac{w_t}{2} &\small{\text{ if }  w_t\leq \delta_t}\\
 0 &\small{\text{otherwise.}}
\end{cases}
$ & $\geq 0$ & $\geq 0$\\
\end{tabular}
    \caption{Values of $(\delta_{t+1},S_{t+1})$ and expected value of  $\Delta\text{cost}_{t+1}$ conditioning on $(\delta_t, S_t)$ and on $r_{t+1}$, assuming that $S_t\cap \{0\}\neq \emptyset$, $\delta_t\neq 0$ and $|S_t\cap (\delta_t,1]|\geq 1$, and where $w_t:=s_{t,2}-s_{t,1}$.}
    \label{tab:MC_gamma_S_main}
\end{table}
\end{center}

 We recall that for any interval $I\subseteq [0,1]$, $t_{I}:= \min \{t\geq m|\;S_{t}\cap I = \emptyset\}$ is the time at which $I$ is depleted, and that $t^d := \min\{t\geq m: \delta_t = 0\}$ is the time at which the gap disappears. Using the properties given in Lemma \ref{lem:MC_def},  we  next show the following lemma.

\begin{restatable}{rLem}{lemsimpleversionprocess}
\label{lem:simple_version_process}
Conditioning on the gap $\delta_m$ and available servers $S_m$, and for all $y\in [\delta_m,1]$, we have
    \begin{equation*}
\mathbb{P}\Big(\min(t_{(0,y]}, t_{\{0\}}) \leq \min(t^d, t_{\{0\}})\Big|\delta_m, S_m\Big)\geq \frac{\delta_m}{y}.
\end{equation*}
\end{restatable}

 In other words, starting from a gap $\delta_m$, the probability that the gap has not yet disappeared at the time all the servers in $(0,y]$ have been depleted, or that all the servers at location $0$ are depleted before either of these events occurs, is lower bounded by $\frac{\delta_m}{y}$.

\paragraph{Part 3 of the proof of Lemma~\ref{lem:P_delta_max}.} Since we have shown in the first part that the intervals $\{I_j\}$ are depleted in increasing order of $j$, we have that just before the time $t_{y_i}$ where $(0,y_i] =\cup_{j\leq i} I_j$ is depleted, none of the intervals $I_j$ for $j< i$ have free servers left, hence $\min\{s>0: s\in S_{t_{y_i}-1}\}\in I_i$. Hence, if $\delta_{t_{y_i}-1}\neq 0$, we have by the definition of $\delta_t$ that $\delta_{t_{y_i}-1} = \min\{s>0: s\in S_{t_{y_i}-1}\}\in I_i = (y_{i-1}, y_i]$, which, in particular, implies $\delta_{t_{y_i}-1}\geq y_{i-1}$. Thus, to prove the desired result, it suffices to lower bound the probability that $\delta_{t_{y_i}-1}\neq 0$ and that $t_{y_i}\leq t_{\{0\}}$ and $t_{y_i}\leq  n-n^{c_3}$. By using the second part, we show that it is lower bounded by $\frac{\delta_m}{y_i}-\hp$.

\subsection{The main lower bound result} 
\label{sec:main_lb_result}

By combining the main lemma (Lemma \ref{lem:P_delta_max}) with Lemma~\ref{cor:costLdelta_gamma}, we can show the following bounds on $\mathbb{E}[cost(\H^{m-1}) - cost(\H^m)]$.
\begin{restatable}{rLem}{lemmnc}
\label{lem:m_cn2}
\hfill
\begin{enumerate}
    \item For any $m> c_1n$, we have: $
     \mathbb{E}[cost(\H^{m-1}) - cost(\H^m)|r_m \in [0,y_0]] = -O(n^{-1/5}).$
     \item For any $m\leq c_1n$, we have: $
     \mathbb{E}[cost(\H^{m-1}) - cost(\H^m)|r_m \in [0,y_0]] = \Omega(\log(n)n^{-1/5}).
    $
     \item For any $m\in [n]$, we have: $        \mathbb{E}[cost(\H^{m-1}) - cost(\H^m)|r_m \in (y_0,1]] = 0.$
\end{enumerate}
\end{restatable}

The last lemma needed is the following bound on $\OPT$.

\begin{restatable}{rLem}{optoffline}
\label{lem:cost_opt} For any $n\in \mathbb{N}$, the expected cost $\OPT$ of the optimal offline matching for our lower bound instance satisfies: $E[\OPT] = O(n^{3/5})$.
\end{restatable}

By doing a telescoping sum over all $m\in [n]$ and using that $\H^n = \mathcal{A}$ and $\H^0 = \G$, we obtain from Lemma \ref{lem:m_cn2} and \ref{lem:cost_opt} the lower bound.

\thmsemirandomlower*

\begin{proof} Since $\A^0 = \mathcal{G}$ and $\A^n = \mathcal{A}$, we have that
\allowdisplaybreaks{
\begin{align*}
    &\mathbb{E}[cost(\mathcal{G})] - \mathbb{E}[cost(\mathcal{A})]\\
    &= \sum_{m=1}^n\mathbb{E}[cost(\H^{m-1}) - cost(\H^m)]\\ 
    &= \sum_{m=1}^n \mathbb{E}[cost(\H^{m-1}) - cost(\H^m)|r_m \in (y_0,1]]\mathbb{P}(r_m \in (y_0,1])\\
    &+ \sum_{m=1}^{c_1n} \mathbb{E}[cost(\H^{m-1}) - cost(\H^m)|r_m \in [0,y_0]]\mathbb{P}(r_m \in [0,y_0])\\
    &+ \sum_{m=c_1n+1}^n \mathbb{E}[cost(\H^{m-1}) - cost(\H^m)|r_m \in [0,y_0]]\mathbb{P}(r_m \in [0,y_0])\\
    &\geq 0 + \sum_{m=1}^{c_1n} C'\log(n)n^{-1/5} n^{-1/5} - \sum_{m=c_1n+1}^{n} Cn^{-1/5} n^{-1/5}\qquad\tag{for some constants $C,C'>0$}\\
     &= n^{3/5}\Big(C'(\log(n)(c_1 -\tfrac{1}{n}) -C(1-c_1- \tfrac{1}{n})\Big)\\
     &=\Omega(\log(n)n^{3/5}),
\end{align*}
}
where the inequality is by Lemma \ref{lem:m_cn2} and since $\P(r_m \in [0,y_0]) = \P(r_m \in [0,n^{-1/5}]) = n^{-1/5}$. Thus, $\mathbb{E}[cost(\mathcal{G})]\geq\mathbb{E}[cost(\mathcal{A})]+ \Omega(\log(n)n^{3/5}) =\Omega(\log(n)n^{3/5})$. Since by Lemma \ref{lem:cost_opt} we have $\mathbb{E}[OPT] =O(n^{3/5})$, we conclude that $\frac{\mathbb{E}[cost(\mathcal{G})]}{\mathbb{E}[OPT]} = \Omega(\log(n))$.
\end{proof}

%% file: appendix_extensions.tex
\section{Extensions}
\label{app_extension}

\subsection{Beyond the line}
\label{app_dimensions}

An obvious natural extension  to consider is the setting where the servers and requests are located in more general metric spaces beyond the line, such as the unit hypercube of arbitrary dimension $d$. The two-dimensional setting is especially natural due to the ride-hailing platforms and food delivery services applications. In this section, we first present experimental results on the average performance of the greedy algorithm over instances that are drawn according to the fully random model in higher dimensions. We then discuss which parts of our analysis extend to more general metric spaces, which parts of our analysis do not extend to the unit square, and what the challenges are to extend our results to more general metric spaces.

\begin{figure}[t]
    \centering
    \includegraphics[width=0.9\linewidth]{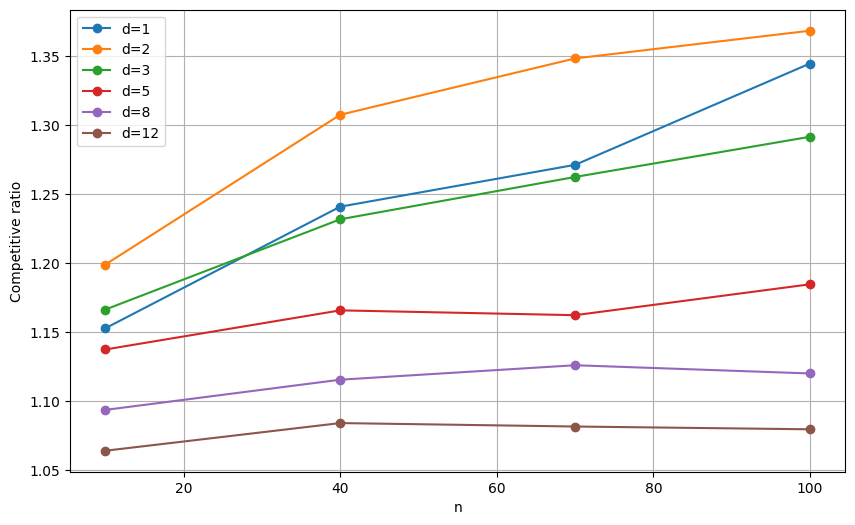}
    \caption{The competitive ratio achieved by the greedy algorithm for different dimensions $d$ and number of uniformly random servers and requests $n$.}
    \label{fig:higher_dim}
\end{figure}

\paragraph{Simulations.}  We experimentally evaluate the average performance of the greedy algorithm over instances  that are drawn according to the fully random model in higher dimensions. The goal of these simulations is to get some indications about whether the competitive ratio of the greedy algorithm might also be constant in higher dimensions and not just on the line.

 Figure~\ref{fig:higher_dim} shows the average competitive ratio achieved by the greedy algorithm for different dimensions $d$ and number $n$ of uniformly random servers and requests. For each point, we averaged the competitive ratio over $20$ random instances. The precise mean and standard deviation of the competitive ratio computed from these simulations are provided in Table~\ref{fig:higher_dim}. The optimal solution to an instance was computed by solving an integer linear programming formulation of the problem.

In all our simulations, the average competitive ratio achieved by the greedy algorithm is at most $1.4$. Interestingly, the competitive ratio is larger in two dimensions than in one dimension and then for $d>2$ the competitive ratio improves as the dimension increases. These simulations suggest that there could be some phenomena that are unique to the $d= 1$ and $d=2$ dimensions settings and that the competitive ratio then smoothly improves as a function of $d$. 
\renewcommand{\arraystretch}{0.9}
\begin{table}[t]
\centering
\begin{tabular}{|c|c|c|c|c|}
\hline
$\mathbf{n}$ & $\mathbf{d}$ & \textbf{Mean} & \textbf{Std. Dev.} \\
\hline
10  & 1 & 1.153 & 0.201 \\ \hline
40  & 1 & 1.241 & 0.134 \\ \hline
70  & 1 & 1.271 & 0.214 \\ \hline
100 & 1 & 1.345 & 0.113 \\ \hline
10  & 2 & 1.199 & 0.095 \\ \hline
40  & 2 & 1.308 & 0.102 \\ \hline
70  & 2 & 1.348 & 0.078 \\ \hline
100 & 2 & 1.369 & 0.069 \\ \hline
10  & 3 & 1.167 & 0.123 \\ \hline
40  & 3 & 1.232 & 0.072 \\ \hline
70  & 3 & 1.263 & 0.045 \\ \hline
100 & 3 & 1.292 & 0.049 \\ \hline
10  & 5 & 1.138 & 0.078 \\ \hline
40  & 5 & 1.166 & 0.042 \\ \hline
70  & 5 & 1.162 & 0.033 \\ \hline
100 & 5 & 1.185 & 0.036 \\ \hline
10  & 8 & 1.094 & 0.065 \\ \hline
40  & 8 & 1.116 & 0.031 \\ \hline
70  & 8 & 1.126 & 0.021 \\ \hline
100 & 8 & 1.120 & 0.013 \\ \hline
10  & 12 & 1.064 & 0.035 \\ \hline
40  & 12 & 1.084 & 0.026 \\ \hline
70  & 12 & 1.082 & 0.011 \\ \hline
100 & 12 & 1.080 & 0.013 \\ \hline
\end{tabular}
\label{tab:higher_dim}
\vspace{.4cm}
\caption{The mean and standard deviation of the competitive ratio achieved by the greedy algorithm over $20$ random instances for different dimensions $d$ and number of uniformly random servers and requests $n$.}
\end{table}

\renewcommand{\arraystretch}{1}

\paragraph{Analysis.} The current analysis of the hybrid lemma (Lemma~\ref{lem:hybrid}) does not extend from the line to the unit square. As we have shown, ``the difference in at most one server" and ``gap remains zero after disappearing" properties that hold for the two hybrid algorithms $\H^{m-1}$ and $\H^{m}$ (properties 1 and 4 in Lemma~\ref{lem:structural_hybrid}) hold in any metric space. However,
the analysis of Lemma~\ref{lem:diff_cost_delta}, which provides the bound $\text{cost}(\H^{m-1}) - \text{cost}(\H^{m}) \leq 2\max_{t\in \{m,\ldots, n-1\}}\delta_{t}$ on the difference in the total costs  of $\mathcal H^m$ and $\mathcal H^{m-1}$ as a function of the maximum gap $\delta _t$, only works for the line. The main reason is that the analysis crucially relies on the gap $\delta_t$ having a monotone increase until there is no more gap. However, over the unit square, it is possible to construct instances where the gap can decrease at some time step without disappearing. This monotonicity is needed to show that $$\text{cost}_t(\H^{m-1}) - \text{cost}_t(\H^{m})\leq \delta_{t} - \delta_{t-1},$$ i.e., to show that the difference $\text{cost}_t(\H^{m-1}) - \text{cost}_t(\H^{m})$ in the costs incurred at time $t$ between the two hybrid algorithms is at most the increase $\delta_t - \delta_{t-1}$ in the gap  $\delta$. Showing this inequality is a crucial step of the proof of Lemma~\ref{lem:diff_cost_delta}.

To achieve positive results in more general metric spaces beyond the line, we believe that the ``difference in at most one server" and ``gap remains zero after disappearing" properties that hold for the two hybrid algorithms $\H^{m-1}$ and $\H^{m}$ might still be important properties to show a more general version of the hybrid lemma (Lemma~\ref{lem:hybrid}). However, a different approach is needed to bound the difference $\text{cost}_t(\H^{m-1}) - \text{cost}_t(\H^{m})$ in the costs of the hybrid algorithms as a function of the increase $\delta_{t} - \delta_{t-1}$ of the gap $\delta$.

\subsection{Beyond the uniform assumption}
\label{app_iid}

Another natural and interesting extension is to relax the assumption that the requests are drawn uniformly and independently and  consider the weaker assumption where the requests are drawn i.i.d. from an arbitrary distribution. The current analysis of the hybrid lemma (Lemma~\ref{lem:hybrid}) does not extend to this setting. We note that Lemma~\ref{lem:diff_cost_delta}, which provides the bound $\text{cost}(\H^{m-1}) - \text{cost}(\H^{m}) \leq 2\max_{t\in \{m,\ldots, n-1\}}\delta_{t}$ on the difference in the total costs  of $\mathcal H^m$ and $\mathcal H^{m-1}$ as a function of the maximum gap $\delta _t$, still holds in the i.i.d. setting (in fact it even holds in the adversarial setting). We recall that $\delta _t$ is the distance between the two free servers that differ for algorithms $\mathcal H^m$ and $\mathcal H^{m-1}$ at time $t$. In addition, the structural properties provided by Lemma~\ref{lem:structural_hybrid}, and in particular the values of the gap $\delta_{t+1}$ at time $t+1$ as a function of the gap $\delta_t$ at time $t$ and the location $r_{t+1}$ of the request at time $t+1$, also still hold.

The part of the analysis of the hybrid lemma that does not extend is the proof of Lemma~\ref{lem:worst_case} that provides a bound on the maximum gap $\delta_t$. This proof bounds the expected size increase $\delta_{t+1} - \delta_t$ of the gap at any time $t$ and for any free server locations. Showing this bound  crucially relies on 1) the values of  $\delta_{t+1}$  as a function of  $\delta_t$ and $r_{t+1}$ and 2) the uniform distribution so that the probability that  request $r_{t+1}$ lies in a given interval is equal to the length of that interval. Since there are intervals $I$ such that $r_{t+1} \in I$ leads to a large size increase for the gap, it is possible to construct distributions with a large probability mass over such an interval $I$ and it is no longer possible to always obtain, for every time $t$, a sufficiently strong bound on the expected size increase $\delta_{t+1} - \delta_t$. 

Despite the current analysis not extending, it might still be possible to provide a bound on the maximum gap $\delta _t$ in the i.i.d. setting. The problematic case is when there exists an interval $I$ at time $t$ such that  a request $r_{t+1} \in I$ causes a large  gap size increase $\delta_{t+1} - \delta_t$. A promising observation is that in such a scenario, a request $r_{t+2} \in I$ in this same interval $I$ at the following time step does not lead to an increase in the gap $\delta$ between time $t+1$ and time $t+2$ and  even leads to $\delta_{t+2} = 0$ with some probability. This observation follows  by inspection of the values for $\delta_{t+1}$  that are provided in Lemma~\ref{lem:structural_hybrid} and are as a function  of $\delta_t$ and $r_{t+1}$.

In summary, the expected increase $\E[\delta_{t+1} - \delta_t]$ in the gap $\delta$ can always be bounded at any time $t$ in the uniform setting. Although this increase cannot always be bounded for every time $t$ in the i.i.d. setting, an expected increase at a time $t$ implies a potential expected decrease at time $t+1$, which could be exploited to bound the expected maximum gap.   

\subsection{The random servers model} 
\label{app_randomservers}

In the random requests model, the $n$ servers are adversarially chosen and the requests are drawn uniformly and independently. In this section, we discuss the random servers model where the $n$ requests are adversarially chosen and the servers are drawn uniformly and independently. As for the i.i.d. relaxation discussed in Section~\ref{app_iid},  the structural properties provided by Lemma~\ref{lem:structural_hybrid}, which include the values of the gap $\delta_{t+1}$  as a function of the gap $\delta_t$ and the location $r_{t+1}$ of the request at time $t+1$, and Lemma~\ref{lem:diff_cost_delta}, which provides the bound $\text{cost}(\H^{m-1}) - \text{cost}(\H^{m}) \leq 2\max_{t\in \{m,\ldots, n-1\}}\delta_{t}$, both hold in the random servers model.

However, the next step of the analysis of the hybrid lemma (Lemma~\ref{lem:hybrid}) does not extend to this model. This next step is Lemma~\ref{lem:worst_case} that provides a bound on the maximum gap $\delta _t$. In fact, if the $n$ servers are located at locations $i / (n+1)$ for $i \in \{0, \ldots, n-1\}$,   it is easy to construct adversarial requests that can cause the gap $\delta_t$ to be arbitrarily large. Note that in the random servers model, $i / (n+1)$ is the expected location of the $i^{th}$ leftmost server. 

Even though the gap $\delta_t$ can be arbitrarily large, the adversarial sequence of requests that causes such a large gap also causes a high cost for the optimal solution. This high cost for the optimal solution when trying to make the gap grow large explains why such  adversarial requests do not easily lead to a lower bound on the performance of greedy in the random servers model. We believe that an interesting direction for positive results for greedy in the random servers model is to bound the increase $\delta_{t+1} - \delta_t$, at each time step $t$, in the gap  as a function of the increase in the cost of the optimal solution.

\subsection{The sublinear excess of servers regime}
\label{app_sublinear}

In the fully random model, 
Theorem~\ref{thm:random_balanced} shows that greedy achieves a constant approximation when the number of servers is equal to the number of requests and Theorem~\ref{thm:random_unbalanced} shows that it also achieves a constant approximation when there is a linear excess supply, i.e., when the number of servers is $(1+\epsilon)n$, for some constant $\epsilon >0$, and the number of requests is $n$.

It is an interesting open question to analyze the performance of greedy in the fully random model when there is a sublinear excess of servers. There is a sharp transition in the average cost incurred by greedy for matching the last $m$ requests  from when there are $n$ servers to when there are $(1+\epsilon)n$ servers for some constant $\epsilon > 0$. This average cost is $\Omega(1/\sqrt{n})$ when there are $n$ servers and it is $O(1/n)$ when there is a linear excess supply of servers. Analyzing this sharp transition in the average cost incurred by greedy, as well as the average cost incurred by the optimal solution, between these two regimes is the main challenge.

%% file: appendix_strategyproof.tex
\section{Strategyproofness in Online Minimum Cost Matching}
\label{sec:appstrategyproof}

 In the context of metric matching, strategyproofness has been considered for the problem of \emph{distortion} in metric matching. In metric matching distortion, the customers (referred to as agents in distortion problems) are all known offline but the algorithm is only given ordinal information about the preferences of the customers over the servers (referred to as items in distortion problems). The information limitation of the algorithm is thus not due to the unknown future customers but due to having only access to the ordinal preferences of customers and the locations of both customers and servers being unknown. \citet{caragiannis2016truthful} discusses the connection between online minimum cost matching and metric matching distortion. In particular, they note that the greedy algorithm in online minimum cost matching is known as the serial dictatorship mechanism in metric matching distortion and that the greedy algorithm under random order arrival correponds to the randomized serial dictatorship mechanism.  \citet{anari2023distortion} consider mechanisms for metric matching distortion that are strategyproof, i.e., where the customers can never benefit from misreporting their ordinal preferences. They show that the $2^{n} - 1$ lower bound for greedy (serial dictatorship) extends to a broad family of mechanisms called serializable mechanisms.

To formally define strategyproofness, we first define the cost of a customer to be the distance between its true location and the location of the server it is matched to. We note that the cost incurred by a customer is equal to the cost incurred by the algorithm for matching that customer.

In the mechanism design version of the online minimum cost problem, the true locations of the customers are private, and can be strategically reported to the matching mechanism. Informally, we say that a mechanism is customer-strategyproof if customers have no incentive to misreport their true location, regardless of the reports of other customers. More formally, a mechanism is customer-strategyproof if, for any instance of the problem and any reports of all customers except some customer $t$, the cost that customer $t$ with true location $r_t$ incurs by reporting a location $r'_t \neq r_t$ is at least the cost that it incurs by reporting its true location $r_t$.

Note that, by definition, the greedy procedure is customer-strategyproof since it matches each customer to the server that minimizes the cost of that customer. Except for trivial algorithms that ignore the locations of the customers, other existing algorithms for online minimum cost matching are, to the best of our knowledge, not strategyproof.  For example, the hierarchical greedy algorithm \cite{Kanoria21} introduced in Section~\ref{section:upper_bound_iid} partitions the interval into different regions and then matches a request to the closest available server in its region. This is not customer-strategyproof, since a customer who is close to a server in another region might be better off by misreporting its location in this other region, which we next show formally.

\begin{lemma}
    The hierarchical greedy algorithm is not customer-strategyproof.
\end{lemma}

\begin{proof}
    Consider an instance with two servers $s_1 = 0$ and $s_2 = 1/2+\epsilon$ (for some small constant $\epsilon >0$), and a request $r_1 = 1/2-\epsilon$. By definition of hierarchical greedy, if the customer reports its true location $r_1$, it will be matched to server $s_1$  and incur a cost $1/2-\epsilon$. However, by reporting any location $r>1/2$, it will be matched to server $s_2$ and incur a cost $2\epsilon < 1/2-\epsilon$. Hence, hierarchical greedy is not customer-strategyproof.
\end{proof}

We next show a similar result for the current-best algorithm for the i.i.d.~model in general metric spaces.

\begin{lemma}
    The algorithm fair-bias \cite{GuptaGPW19} is not customer-strategyproof.
\end{lemma}

Before presenting the proof, we recall the main steps of the algorithm fair-bias. For all $t\in [n]$, before the arrival of request $r_t$, the algorithm computes an optimal fractional  matching $x^{S_{t-1}}$ between $S_{t-1}$ and $S$. More precisely,  $x^{S_{t-1}}$ is an optimal solution of the following linear program:

\begin{align*}\label{eq:fairbias}
    {M(S_{t-1})\;:=}\qquad \text{min}& \sum_{s\in S_{t-1} ,r\in S} |r-s| \cdot x_{s,r}\\
    \text{s.t }& \quad \sum_{r\in S} x_{s,r} = \frac{1}{|S_{t-1}|} \qquad\qquad \forall s\in S_{t-1} \\
             &\quad  \sum_{s\in S_{t-1}} x_{s,r} = \frac{1}{|S|} \qquad\qquad \;\forall r\in S \\
             &\quad  x \geq 0.
\end{align*}

When request $r_t = r$ arrives, it is then assigned to a randomly sampled server $s$ from $S_{t-1}$, where each server $s\in S_{t-1}$ is chosen with probability $n\cdot x^{S_{t-1}}_{s,r}$.

We are now ready to present the proof.

\begin{proof}
Consider an instance with three servers $s_1 = 0$, $s_2 =1$ and $s_3 = 4$, and two customers $r_1=r_2=1$. 

Note that the optimal solution to the LP $M(S_0)$ is $x^{S_0}_{0,0}=x^{S_0}_{1,1}= x^{S_0}_{4,4} = 1/3$ and $x^{S_0}_{s,r} = 0$ when $s\neq r$. Hence, when $r_1=1$ arrives, it is matched to server $s_2$ with probability $3\cdot 1/3 = 1$.

Now, we have $S_1 = \{0,4\}$, and the optimal solution to the LP $M(S_1)$ is $x^{S_1}_{0,0} = 1/3$, $x^{S_1}_{0,1} = 1/6$, $x^{S_1}_{4,1} = 1/6$, $x^{S_1}_{4,4} = 1/3$ and $x^{S_1}_{s,r} = 0$ otherwise. If the customer reports its true location $r_1=1$, it will be matched to server $s_1 = 0$ with probability $3\cdot 1/6 = 1/2$ and to $s_3 = 4$ with probability $3\cdot 1/6 = 1/2$, incurring costs $1$ and $3$, respectively. However, by reporting location $r = 0$, it will be always matched to $s_1=0$ and incur cost $1$.  Hence, fair-bias is not customer-strategyproof.
\end{proof}

%% file: appendix_auxilliary_lemmas.tex
\section{Auxiliary Lemmas}

Throughout the paper, we will use the following version of Chernoff bounds. 

\begin{lemma}(Chernoff Bounds)
\label{prop:chernoff_bounds}
    Let $X = \sum^n_{i=1} X_i$, where $X_i = 1$ with probability $p_i$ and $X_i = 0$ with probability $1 - p_i$, and all $X_i$ are independent. Let $\mu = \mathbb{E}[X] = \sum^n_{i=1} p_i$. Then
    \begin{itemize}
        \item Upper tail: $P(X \geq (1+\delta)\mu)\leq e^{-\delta^2\mu/(2+\delta)}$ for all $\delta>0$.
        \item Lower tail: $P(X \leq (1-\delta)\mu)\leq e^{-\delta^2\mu/2}$ for all $\delta \in [0,1]$.
    \end{itemize}
\end{lemma}

In particular, we will repeatedly use the following lemma, which immediately follows from Chernoff bounds.
\begin{lemma}
\label{cor:CB}
Let $X\sim\mathcal{B}(n,p)$ be a binomially distributed random variable with parameters n and p. Then, 
\[\mathbb{P}\big(X\geq \mathbb{E}[X]-\log(n)^2\sqrt{\mathbb{E}[X]}\big)\geq 1 - \hp,\] 
and if $np = \Omega(1)$,
\begin{equation*}
    \mathbb{P}\big(X\leq \mathbb{E}[X]+\log(n)^2\sqrt{\mathbb{E}[X]}\big) \geq 1 - \hp.
\end{equation*}
\end{lemma}

\begin{proof}
This results from a direct application of Chernoff bounds as stated in Lemma \ref{prop:chernoff_bounds} with  $\delta = \log(n)^2/\sqrt{\mathbb{E}[X]}$:
\begin{align*}
    \mathbb{P}(X\leq \mathbb{E}[X](1-\log(n)^2/\sqrt{\mathbb{E}[X]}))&\leq 
     e^{-\log(n)^4/2}=n^{-\Omega(\log(n))},
\end{align*}
and if $np = \Omega(1)$, since $1/\sqrt{\mathbb{E}[X]} = 1/\sqrt{np} = O(1)$, we have
\begin{align*}
    \mathbb{P}(X\geq \mathbb{E}[X](1+\log(n)^2/\sqrt{\mathbb{E}[X]}) &\leq 
     e^{-\log(n)^4/(2+\log(n)^2/\sqrt{\mathbb{E}[X]}))}
    =n^{-\Omega(\log(n))}. \qedhere
\end{align*}
\end{proof}

Finally, we recall the following classical inequality, which follows immediately from Jensen's inequality. 
\begin{lemma}
\label{lem:jensen}
For any random variable $Y$: $
    \mathbb{E}[|Y-\mathbb{E}[Y]|] \leq std(Y)$, where $std$ denotes the standard deviation.
\end{lemma}


%% file: appendix_hybrid_lemma.tex
\section{Proof of the Hybrid Lemma (Lemma \ref{lem:hybrid})}
\label{app_hybrid}

The objective of this section is to prove Lemma \ref{lem:hybrid}, that we restate below. 
\lemhybrid*

 The general structure follows that of the proof of the hybrid lemma (Lemma 5.1) in  \cite{GuptaL12}. However, \cite{GuptaL12} considers a fixed deterministic sequence of requests and uses a coupling argument between a randomized greedy algorithm and an optimal offline matching, whereas we directly leverage the randomness of the input sequence to analyze the performance of hybrid algorithms between an online algorithm $\A$ and the standard deterministic greedy algorithm.


\noindent\textbf{Overview of the proof.} A key component of the proof of Lemma~\ref{lem:hybrid} relies on Lemma \ref{lem:structural_hybrid} given below, that describes, for a fixed $m \in [n]$, the difference between the executions of $\H_{\A}^m$ and $\H_{\A}^{m-1}$ on the same sequence $R$. Lemma~\ref{lem:structural_hybrid} in fact shows that, at every step $t$, the free servers for both algorithms coincide, with the exception of at most a pair of servers, that we denote by $g^L_t < g^R_t$ (see Figure~\ref{fig:set_of_servers_main}), that there is no other free server in between $g^L_t$ and $g^R_t$ ; and that strong bounds can be obtained on $\delta_t:=g^R_t-g^L_t$. These properties, in turns, will allow us to control the difference in the costs incurred by the two algorithms, by first upper bounding it by the value of the maximum gap $\delta_t$ (Lemma~\ref{lem:diff_cost_delta}), then by upper bounding the probability that this gap $\delta_t$ grows large (Lemma~\ref{lem:worst_case}). This eventually leads to the bound from Lemma~\ref{lem:hybrid}.

\noindent\textbf{Additional notations.} To ease the exposition, we abuse notations by writing $\H^m$ and $\H^{m-1}$ instead of $\H_{\A}^m$ and $\H_{\A}^{m-1}$. We also drop the reference to the algorithms in the indices and  write $S_{t}$ and $s(r_{t})$ instead of $S_{\H^m, t}$ and $s_{\H^m}(r_t)$ to denote, respectively, the set of free server for $\H^m$ just after matching $r_t$ and the server to which $\H^m$ matches $r_t$. Similarly, we write $S_t'$ and  $s'(r_t)$ instead of $S_{\H^{m-1}, t}$ and $s_{\H^{m-1}}(r_t)$ for the equivalent objects for $\H^{m-1}$. We also define $ \mathcal{N}(r_t)$ and $ \mathcal{N}'(r_t)$ the set of available servers neighboring $r_t$ for $\H^m$ and $\H^{m-1}$ when $r_t$ arrives, where two points neighbor each other if there is no other point between them.

 If $S_t = S_t'$, then we write $g^L_t =g^R_t = \emptyset$ and $\delta_t = 0$. We also define
       $s_t^L = \max\{s\in S_t\cup S_t'\setminus\{g_t^L, g_t^R\}: s\leq g_t^L\}$ and $s_t^R = \min\{s\in S_t\cup S_t'\setminus\{g_t^L, g_t^R\}: s\geq g_t^R\}$ (with the convention that $s_t^L=\emptyset$ if $\{S_t\cup S_t'\setminus\{g_t^L, g_t^R\} : s \geq g_t^L\}= \emptyset$ or if $g_t^L=\emptyset$, and similarly for $s_t^R$), which are the nearest servers of $S_t$ (or equivalently, of $S_t'$) on the left of $g_t^L$ and on the right of $g_t^R$. When it is clear from context, we drop the dependency on $t$. Finally, for any two consecutive sets of servers $S,S'$, i.e., sets differing by a pair of consecutive servers $g<g'$, we write $\delta(S,S'):= g'-g$.


\begin{lemma}{(Expanded version of Lemma~\ref{lem:structural_hybrid_main})}
\label{lem:structural_hybrid}
Let $\A$ be any online algorithm that makes neigbhoring matches, $S_0$ be $n$ arbitrary servers and $R$ be $n$ arbitrary requests. Let $(S_0,\ldots, S_n)$ and $(S_0',\ldots, S_n')$ denote the set of free servers for $\H_{\A}^{m}$  and $\H_{\A}^{m-1}$ at each time steps. Then, the following propositions hold for all $t\in \{m, \ldots, n\}$:
\begin{enumerate}

    \item  \textbf{Difference in at most one server.} In any metric space,  $|S_t\setminus S_t'|=|S_t'\setminus S_t|\leq 1$. 
        \item \textbf{Consecutiveness of the different servers.} On the line, if $g^L_t,g^R_t\neq \emptyset$, there is no server $s \in S_t \cup S_t'$ such that $g^L_t < s < g^R_t$.
        \item \textbf{The values.} On the line, if $t<n$ and $S_t\neq S_t'$ (and assuming without loss of generality that $S_t = S_t' \cup\{g^L_t\}\setminus \{g^R_t\}$), then the values of $s(r_{t+1}), s'(r_{t+1}), \delta_{t+1}, g^L_{t+1}, g^R_{t+1}$ and an upper bound on $\Delta\text{cost}_{t+1} := |\text{cost}_{t+1}(\H^{m-1})- \text{cost}_{t+1}(\H^m)|$ are given in Tables \ref{tab:delta_general}, \ref{tab:delta_special} and \ref{tab:delta_special2}:
        \begin{itemize}
        \item[(a)] if $s_t^L \neq \emptyset, s_t^R \neq \emptyset$, the values are given in  Table \ref{tab:delta_general}, where $d_t^L := g_t^L - s_t^L$ and $d_t^R := s_t^R - g_t^R$,
            \item[(b)] if $s_t^L = \emptyset, s_t^R \neq \emptyset$, the values are given in  Table \ref{tab:delta_special}, where $d_t^R := s_t^R - g_t^R,$
            \item[(c)] if $s_t^R = \emptyset, s_t^L \neq \emptyset$, the values are given in  Table \ref{tab:delta_special2}, where $d_t^L := g_t^L - s_t^L$,
            \item[(d)] if $s_t^L = \emptyset, s_t^R = \emptyset$ then $S_{t+1} = S_{t+1}' = \emptyset, \delta_{t+1}=0$, and $|\text{cost}_{t+1}(\H^{m-1})- \text{cost}_{t+1}(\H^m)|\leq \delta_t$.
        \end{itemize}
        \item \textbf{Gap remains zero after disappearing.} In any metric space,  if $\delta_t = 0$, then $\delta_{t'}=0$ for all $t'\geq t$.
        \end{enumerate}
\end{lemma}

\renewcommand{\arraystretch}{1.2}

\begin{center}
\begin{table}[]
    \centering
    \begin{tabular}{c||c| c| c| c| c|c} 
\scriptsize{$r_{t+1}-s_t^L\in \ldots$} & \scriptsize{$[0,\tfrac{d_t^L}{2}]$} &  \scriptsize{$[ \tfrac{d_t^L}{2}, \tfrac{d_t^L+\delta_t}{2}]$ }&   \scriptsize{$[\tfrac{d_t^L+\delta_t}{2}, d_t^L + \tfrac{d_t^R+\delta_t}{2}]$ }&\scriptsize{$[d_t^L +\tfrac{d_t^R+\delta_t}{2},$ }&\scriptsize{$[d_t^L + \delta_t + \tfrac{d_t^R}{2},$} &\scriptsize{ $[-s_t^L,0)$}\\
 &  &  &   &\scriptsize{$d_t^L + \delta_t + \tfrac{d_t^R}{2}]$}&\scriptsize{$d_t^L + \delta_t + d_t^R]$}& \scriptsize{$\cup (s_t^R-s_t^L,1-s_t^L]$} \\
 \hline
 \hline
 \scriptsize{$s(r_{t+1})$} & $s_t^L$ & $g_t^L$  &  $g_t^L$  & $s_t^R$ & $s_t^R$ & \scriptsize{$\in [0, s_t^L)\cup (s_t^R,1]$}\\
 \hline
 \scriptsize{$s'(r_{t+1})$} & $s_t^L$ & $s_t^L$  &  $g_t^R$ & $g_t^R$ & $s_t^R$ & \scriptsize{$\in [0, s_t^L)\cup (s_t^R,1]$}\\
 \hline
 \scriptsize{$g_{t+1}^L$} & $g_t^L$ & $s_t^L$  &  $\emptyset$  & $g_t^L$ & $g_t^L$ &$g_t^L$\\
 \hline
 \scriptsize{$g_{t+1}^R$} & $g_t^R$ & $g_t^R$ & $\emptyset$ & $s_t^R$ &  $g_t^R$ &$g_t^R$ \\
 \hline
 \scriptsize{$\delta_{t+1}$} &  $\delta_{t}$ & $\delta_{t}+d_t^L$  &  $0$ & $\delta_t+d_t^R$ & $\delta_t$& $\delta_t$\\
 \hline
 \scriptsize{$\Delta\text{cost}_{t+1} \leq$} & $0$  & $d_t^L$ &$\delta_t$ & $d_t^R$ & $0$ & $0$ \\
\end{tabular}
    \caption{Values of $\delta_{t+1}, g^L_{t+1}, g^R_{t+1}$, and upper bound on $\Delta\text{cost}_{t+1}$ when $s_t^L, s_t^R\neq \emptyset$.}
    \label{tab:delta_general}
\end{table}
\end{center}

\begin{center}
\begin{table}[]
    \centering
    \begin{tabular}{c||c| c| c| c} 

\scriptsize{$r_{t+1}\in \ldots$} & \scriptsize{$[0,g^L_t +\tfrac{d_t^R+\delta_t}{2}$]} &\scriptsize{$[g^L_t +\tfrac{d_t^R+\delta_t}{2}, g^L_t + \delta_t + \tfrac{d_t^R}{2}]$ }&\scriptsize{$[g^L_t + \delta_t + \tfrac{d_t^R}{2}, g^L_t + \delta_t + d_t^R]$}&\scriptsize{ $(s_t^R,1]$} \\
 \hline
  \hline
 \scriptsize{$s(r_{t+1})$} &  $g_t^L$  & $s_t^R$ & $s_t^R$ & $\in (s_t^R,1]$\\
 \hline
 \scriptsize{$s'(r_{t+1})$} &   $g_t^R$ & $g_t^R$ & $s_t^R$ & $\in (s_t^R,1]$\\
 \hline
 \scriptsize{$g_{t+1}^L$} &   $\emptyset$  & $g_t^L$ & $g_t^L$ &$g_t^L$\\
 \hline
 \scriptsize{$g_{t+1}^R$} &  $\emptyset$ & $s_t^R$ &  $g_t^R$ &$g_t^R$ \\
 \hline
 \scriptsize{$\delta_{t+1}$}  &  $0$ & $\delta_t+d_t^R$ & $\delta_t$& $\delta_t$\\
 \hline
 \scriptsize{$\Delta\text{cost}_{t+1} \leq$}  &$\delta_t$ & $d_t^R$ & $0$ & $0$ \\
\end{tabular}
    \caption{Values of $\delta_{t+1}, g^L_{t+1}, g^R_{t+1}$, and upper bound on $\Delta\text{cost}_{t+1}$ when $s_t^L= \emptyset, s_t^R \neq \emptyset$.}
    \label{tab:delta_special}
\end{table}
\end{center}

\begin{center}
\begin{table}[]
    \centering
    \begin{tabular}{c||c| c| c| c} 

\scriptsize{$r_{t+1}\in \ldots$}&\scriptsize{ $[0,s_t^L)$}  & \scriptsize{$[g_t^R - ( \delta_t+d_t^L),g_t^R - (\delta_t + \frac{d_t^L}{2})]$} &  \scriptsize{$[g_t^R - (\delta_t + \frac{d_t^L}{2}),g_t^R-\tfrac{d_t^L+\delta_t}{2}]$ }&   \scriptsize{$[g_t^R-\tfrac{d_t^L+\delta_t}{2}, 1]$ }\\
 \hline
 \hline
 \scriptsize{$s(r_{t+1})$} & $\in [0, s_t^L)$ & $s_t^L$  &  $g_t^L$  &  $g_t^L$ \\
 \hline
 \scriptsize{$s'(r_{t+1})$} & $\in [0, s_t^L)$ & $s_t^L$  &  $s_t^L$ & $g_t^R$\\
 \hline
 \scriptsize{$g_{t+1}^L$} &$g_t^L$ & $g_t^L$ & $s_t^L$  &  $\emptyset$  \\
 \hline
 \scriptsize{$g_{t+1}^R$}&$g_t^R$ & $g_t^R$ & $g_t^R$ & $\emptyset$  \\
 \hline
 \scriptsize{$\delta_{t+1}$}& $\delta_t$ &  $\delta_{t}$ & $\delta_{t}+d_t^L$  &  $0$ \\
 \hline
 \scriptsize{$\Delta\text{cost}_{t+1} \leq$} & $0$ & $0$  & $d_t^L$ &$\delta_t$  \\
\end{tabular}
    \caption{Values of $\delta_{t+1}, g^L_{t+1}, g^R_{t+1}$, and upper bound on $\Delta\text{cost}_{t+1}$ when $s_t^R= \emptyset, s_t^L \neq \emptyset$}
    \label{tab:delta_special2}
\end{table}
\end{center}

\begin{proof}
First, note that since $\H^m$ and $\H^{m-1}$ both match  $r_1, \ldots, r_{m-1}$ to exactly the same servers that $\A$ matches them to, we have that $S_t = S_t'$ for all $t\in [m-1]$.

We now show facts $1,2,3$ by induction on $t\in \{m,\ldots, n\}$. We first show fact 1 (difference in at most one server).  Since $S_{m-1} = S_{m-1}'$, we have  $|S_{m}\setminus S_{m}'|=|S_{m}'\setminus S_{m}|\leq 1$. Next, assume that $|S_{t}\setminus S_{t}'|=|S_{t}'\setminus S_{t}|\leq 1$ is satisfied at time $t \in \{m, \ldots, n-1\}$.  

Consider the case where  $\H^m$ and $\H^{m-1}$ match the request at time $t+1$ to  different servers, i.e., $s(r_{t+1}) \neq s'(r_{t+1})$. Then, since $\H^m$ and $\H^{m-1}$ both match request $r_{t+1}$ greedily and ties are broken consistently, we have that one of these two servers must not have  been available to one of the algorithms, i.e.,  $s(r_{t+1}) \not \in S_{t}'$  or $s'(r_{t+1}) \not \in S_{t}$. Consider the case $s(r_{t+1}) \not \in S_{t}'$. Then $$|S_{t+1}\setminus S_{t+1}'| =  |S_{t}\setminus S_{t+1}'| - 1 \leq  |S_{t}\setminus S_{t}'| \leq 1$$ where the equality is since $s(r_{t+1}) \not \in S_{t}'$, the first inequality is since $|S_{t+1}'| = |S_{t}'|-1$, and the second inequality is by the inductive hypothesis.
Next, we have $$|S_{t+1}'\setminus S_{t+1}|\leq |S_{t}'\setminus S_{t+1}| = |S_{t}'\setminus S_{t}| \leq 1$$ where the first inequality is since $S_{t+1}' \subseteq S'_t $, the equality is since $s(r_{t+1}) \not \in S_{t}'$, and the second inequality is by the inductive hypothesis. We have that $|S_{t+1}\setminus S_{t+1}'|=|S_{t+1}'\setminus S_{t+1}|$  since $|S_{t+1}| = |S_{t+1}'|$. Combining the three previous series of (in)equalities, we obtain the inductive claim  $|S_{t+1}\setminus S_{t+1}'|=|S_{t+1}'\setminus S_{t+1}| \leq 1$. The case $s'(r_{t+1}) \not \in S_{t}$ follows by a symmetric argument.

The remaining case is where $\H^m$ and $\H^{m-1}$ match the request at time $t+1$ to  the same server, i.e., $s(r_{t+1}) = s'(r_{t+1})$. In this case, we get $|S_{t+1}\setminus S_{t+1}'| = |S_{t}\setminus S_{t}'|$ and $|S'_{t+1}\setminus S_{t+1}| = |S'_{t}\setminus S_{t}|$, which implies the inductive claim and concludes the proof of fact 1.

We next show that fact 2  is satisfied for $t = m$. First, define $\Tilde{s}^L_{m-1} = \max\{s\in S_{m-1} : s\leq r_m\}$, or $\Tilde{s}^L_{m-1} = \emptyset$ if there is no such server, $\Tilde{s}^R_{m-1} = \min\{s\in S_{m-1} : s\geq r_m\}$, or $\Tilde{s}^R_{m-1} = \emptyset$ if there is no such server.  Now, recall that $\H^m$ matches $r_m$ to the same server as $\A$, while $\H^{m-1}$ matches $r_m$ greedily.  
Since both $\A$ and greedy make neighboring matches, we have $s(r_m), s(r_m)'\in \{\Tilde{s}^L_{m-1}, \Tilde{s}^R_{m-1}\}$.  Assume first that $\H^m$ and $\H^{m-1}$ make the same matching decision for $r_m$. Then bullet point $2$ clearly holds. 

Suppose now that one algorithm matches $r_m$ to $\Tilde{s}^L_{m-1}$ whereas the other matches it to $\Tilde{s}^R_{m-1}$. By definition of $g_{m}^L$ and $g_{m}^R$, we have $g_{m}^L = \Tilde{s}^L_{m-1}$ and $g_{m}^R = \Tilde{s}^R_{m-1}$. Now, by definition of $\Tilde{s}^L_{m-1}, \Tilde{s}^R_{m-1}$, there is no server $s\in S_m\cup S_m'$ such that $g_{m}^L< s< g_{m}^R$, which shows that bullet point $2$ is satisfied at time $m$.

Next, let $t\in \{m, \ldots, n-1\}$ and assume that $1,2$ are satisfied at time $t$. We show that $3$ is satisfied at time $t$ and that  $2$ is satisfied at time $t+1$. This concludes the proof of 1,2,3 by induction. Recall that both algorithms match $r_{t+1}$ greedily; hence, if $S_t = S_t'$, the result follows immediately. We now assume that $S_t\neq S_t'$. By the inductive hypothesis, we thus have that $|S_t\setminus S_t'| = |S_t\setminus S_t'| = 1$, with $S_t\Delta S_t' = \{g^L_t, g^R_t\}$,  and that there is no server $s \in S_t \cup S_t'$ such that $g^L_t < s < g^R_t$. We assume without loss of generality that $S_t = S_t'\cup \{g_t^L\}\setminus\{g_t^R\}$.  We consider different cases depending on whether there is a free server $s_t^L$ on the left of $g_t^L$ and a free server $s_t^R$ on the right of $g_t^R$. Recall that we defined $d_t^L = g_t^L - s_t^L$ when $s_t^L\neq \emptyset$ and $d_t^R = s_t^R - g_t^R$ when $s_t^R\neq \emptyset$.

\vspace{0.2cm}
\noindent $\mathbf{s_t^L,s_t^R\neq \emptyset:}$  We consider all possibles cases depending on the location of request $r_{t+1}$. Recall that $ \mathcal{N}(r_{t+1})$ and $ \mathcal{N}'(r_{t+1})$ are the set of available servers neighboring $r_{t+1}$ for $\H^m$ and $\H^{m-1}$ when $r_{t+1}$ arrives.

\begin{itemize}
    \item Case 1 (see Figure~\ref{fig:c1}): $r_{t+1}\in s_t^L +[0,\tfrac{d_t^L}{2}]$. In this case, we have $\mathcal{N}(r_{t+1}) = \{s_t^L, g_t^L\}$, $\mathcal{N}'(r_{t+1}) = \{s_t^L, g_t^R\}$, and it is immediate that $|r_{t+1}-s_t^L|\leq |g_t^L-r_{t+1}|$ and that $|r_{t+1}-s_t^L|\leq |g_t^R-r_{t+1}|$. Hence we get $s(r_{t+1}) = s'(r_{t+1}) = s_t^L$.
    
    Combining this with the  induction hypothesis, we get: $S_{t+1} = S_{t}\setminus\{s_t^L\} = (S_{t}'\cup \{g_t^L\}\setminus\{g_t^R\})\setminus\{s_t^L\} =(S_{t}'\setminus\{s_t^L\})\cup \{g_t^L\}\setminus\{g_t^R\} = S_{t+1}' \cup\{g_t^L\}\setminus\{g_t^R\}$, 
    which immediately implies that $g_{t+1}^L = g_t^L$, $g_{t+1}^R = g_t^R$ and $\delta_{t+1} = \delta_t$. In addition, since $s(r_{t+1}) = s'(r_{t+1})$, we have $\Delta\text{cost}_{t+1}=0$.
    
\begin{figure}
    \centering
    \includegraphics[scale = 0.32]{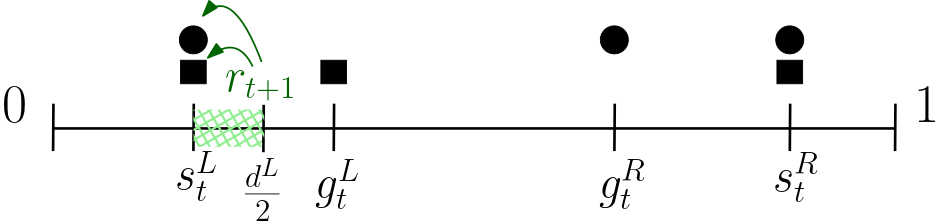}
    \caption{Illustration of Case 1 in the proof of Lemma \ref{lem:structural_hybrid}.}
    \label{fig:c1}
\end{figure}

    \item Case 2 (see Figure~\ref{fig:c2}): $r_{t+1}\in s_t^L + [ \tfrac{d_t^L}{2}, \tfrac{d_t^L+\delta_t}{2}]$. In this case, we have $\mathcal{N}(r_{t+1}) \subseteq \{s_t^L, g_t^L, g_t^R\}$, $\mathcal{N}'(r_{t+1}) = \{s_t^L, g_t^R\}$.  Since $r_{t+1}\geq s_t^L +  \tfrac{d_t^L}{2}$, we have  $|r_{t+1}- s_t^L| \geq |g_t^L - r_{t+1}|$, and since $r_{t+1}\leq s_t^L +  \tfrac{d_t^L+\delta_t}{2}$, we have  $|r_{t+1}- s_t^L| \leq |g_t^R - r_{t+1}|$, thus we get $s(r_{t+1}) = g_t^L$ and $s'(r_t) = s_t^L$. 
    
    Combining this with the  induction hypothesis, we get: $S_{t+1} = S_{t}\setminus\{g_t^L\} = (S_{t}'\cup \{g_t^L\}\setminus\{g_t^R\})\setminus\{g_t^L\} = S_{t}'\setminus\{g_t^R\} = S_{t+1}' \cup\{s_t^L\}\setminus\{g_t^R\}$, which implies that $g_{t+1}^L = s_t^L$, $g_{t+1}^R = g_t^R$, and $\delta_{t+1} = g_t^R - s_t^L = (g_t^R - g_t^L)+(g_t^L - s_t^L) = \delta_t + d_t^L$. In addition, $\Delta\text{cost}_{t+1} = ||r_{t+1}- g_t^L|-|r_{t+1}- s_t^L|| \leq  |g_t^L- s_t^L| = d_t^L$.
    
    \begin{figure}
    \centering
    \includegraphics[scale = 0.2]{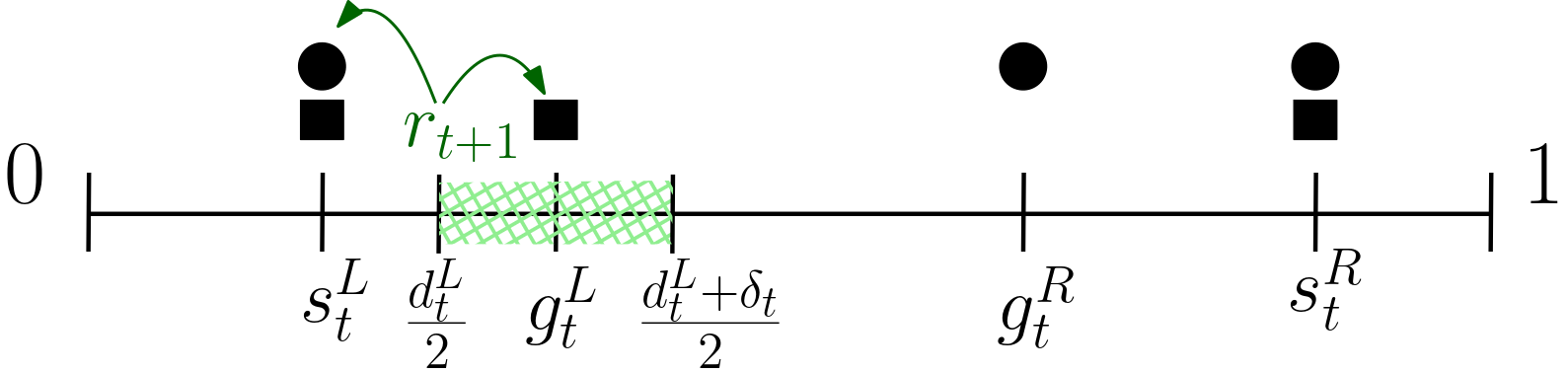}
    \caption{Illustration of Case 2 in the proof of Lemma \ref{lem:structural_hybrid}.}
    \label{fig:c2}
\end{figure}

    \item Case 3 (see Figure~\ref{fig:c3}): $r_{t+1}\in s_t^L + [\tfrac{d_t^L+\delta_t}{2}, d_t^L + \tfrac{d_t^R+\delta_t}{2}]$. In this case, we have $\mathcal{N}(r_{t+1}) \subseteq \{s_t^L, g_t^L, s_t^R\}$, $\mathcal{N}'(r_{t+1}) \subseteq \{s_t^L, g_t^R, s_t^R\}$, with $g_t^L\in \mathcal{N}(r_{t+1})$ and $g_t^R\in \mathcal{N}'(r_{t+1})$. Since $r_{t+1}\geq s_t^L+ \tfrac{d_t^L+\delta_t}{2}\geq s_t^L+\tfrac{d_t^L}{2} $, we have $|r_{t+1}- s_t^L| \geq |g_t^L - r_{t+1}|$, and since $r_{t+1}\leq s_t^L+d_t^L + \tfrac{d_t^R+\delta_t}{2}$, we have $|s_t^R - r_{t+1}| \geq |g_t^L - r_{t+1}|$, thus we get $s(r_{t+1}) = g_t^L$. Similarly, since $r_{t+1}\leq s_t^L+d_t^L + \tfrac{d_t^R+\delta_t}{2}\leq s_t^L+d_t^L + \delta_t+ \tfrac{d_t^R}{2} $, we have $ |s_t^R- r_{t+1}| \geq |g_t^R - r_{t+1}|$, and since $r_{t+1}\geq s_t^L+\tfrac{d_t^L+\delta_t}{2}$, we have $|s_t^L - r_{t+1}| \geq |g_t^R - r_{t+1}|$, thus we get $s'(r_{t+1}) = g_t^R$.   
    
    Combining this with the  induction hypothesis, we get: $S_{t+1} = S_{t}\setminus\{g_t^L\} = (S_{t}'\cup \{g_t^L\}\setminus\{g_t^R\})\setminus\{g_t^L\} = S_{t}'\setminus\{g_t^R\} = S_{t+1}'$,  which implies that $g_{t+1}^L = g_{t+1}^R = \emptyset$ and $\delta_{t+1} = 0$. In addition, $\Delta\text{cost}_{t+1}= ||r_{t+1}- g_t^L|-|r_{t+1}- g_t^R||\leq |g_t^R - g_t^L| = \delta_t$.  
    
     \begin{figure}
    \centering
    \includegraphics[scale = 0.25]{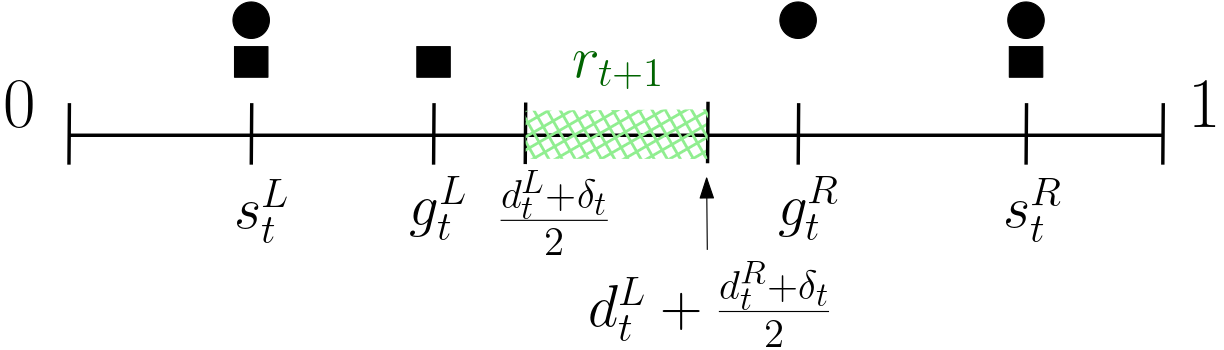}
    \caption{Illustration of Case 3 in the proof of Lemma \ref{lem:structural_hybrid}.}
    \label{fig:c3}
\end{figure}
    
    \item Case 4 (see Figure~\ref{fig:c4}): $r_{t+1}\in s_t^L + [d_t^L +\tfrac{d_t^R+\delta_t}{2}, d_t^L + \delta_t + \tfrac{d_t^R}{2}]$. This case is symmetric to Case 2 by noting the one to one correspondence between $0,d_t^L, s_t^L, g_t^L$  and $1,d_t^R, s_t^R, g_t^R$. We get that $s(r_{t+1}) =s_t^R$ $s'(r_{t+1}) = g_t^R$, which implies $g_{t+1}^L = g_t^L$,  $g_{t+1}^R = s_t^R$ and $\delta_{t+1}=s_t^R- g_t^L = (s_t^R - g_t^R) + (g_t^R - g_t^L) = d_t^R + \delta_t$. In addition, $\Delta\text{cost}_{t+1}\leq | s_t^R - g_t^R| = d_t^R$.
    
     \begin{figure}
    \centering
    \includegraphics[scale = 0.4]{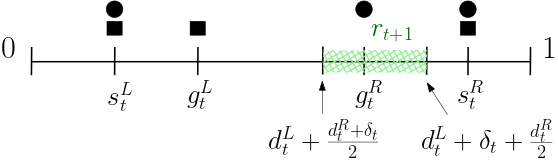}
    \caption{Illustration of Case 4 in the proof of Lemma \ref{lem:structural_hybrid}.}
    \label{fig:c4}
\end{figure}
    
    \item Case 5 (see Figure~\ref{fig:c5}): $r_{t+1}\in s_t^L + [d_t^L + \delta_t + \tfrac{d_t^R}{2}, d_t^L + \delta_t + d_t^R]$. This case is symmetric to Case 1 by noting the one to one correspondence between $0,d_t^L, s_t^L, g_t^L$  and $1,d_t^R, s_t^R, g_t^R$. We get that $s(r_{t+1}) = s'(r_{t+1}) = s_t^R$, which implies  $\Delta\text{cost}_{t+1}=0$, $g_{t+1}^L = g_t^L$, $g_{t+1}^R = g_t^R$ and $\delta_{t+1} = \delta_t$.
    
     \begin{figure}
    \centering
    \includegraphics[scale = 0.26]{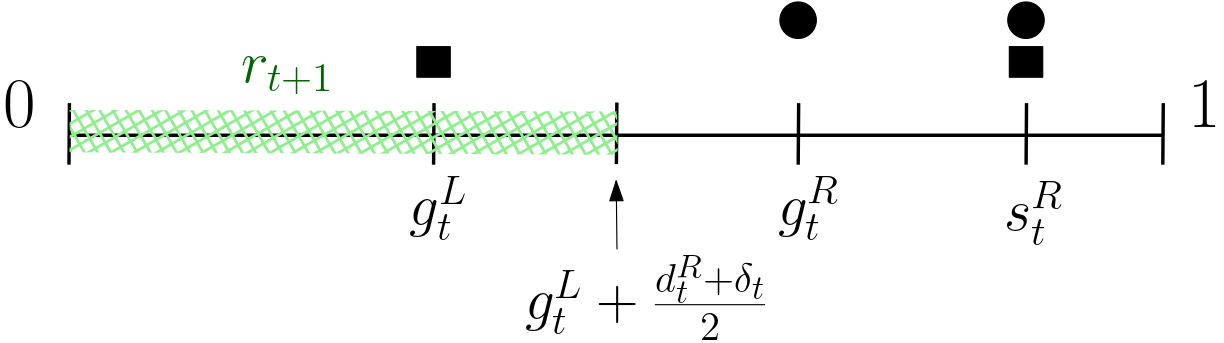}
    \caption{Illustration of Case 5 in the proof of Lemma \ref{lem:structural_hybrid}.}
    \label{fig:c5}
\end{figure}
    
    \item Case 6 (see Figure~\ref{fig:c6}): $r_{t+1}\in [0, s_t^L)\cup (s_t^R,1]$. In this case, the free servers neighboring $r_{t+1}$ are identical for $\H^m$ and $\H^{m-1}$, thus $s(r_{t+1}) = s'(r_{t+1})$. By using the assumption that $|S_t\setminus S_t'| = |S_t\setminus S_t'| =1$, we get  that $g_{t+1}^L = g_t^L$, $g_{t+1}^R = g_t^R$, $\delta_{t+1} = \delta_t$. In addition, since $s(r_{t+1}) = s'(r_{t+1})$, we have $\Delta\text{cost}_{t+1}=0$.
    
     \begin{figure}
    \centering
    \includegraphics[scale = 0.26]{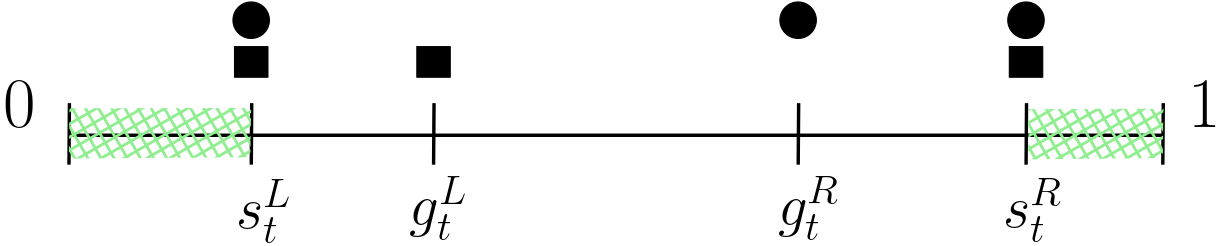}
    \caption{Illustration of Case 6 in the proof of Lemma \ref{lem:structural_hybrid}.}
    \label{fig:c6}
\end{figure}

\end{itemize}

Hence, in all cases, we have that bullet points $1,2$ hold at time $t+1$ and that the values of $g_{t+1}^L,g_{t+1}^R, \delta_{t+1}$ given in Table \ref{tab:delta_general} hold.

\vspace{0.2cm}
\noindent $\mathbf{s_t^L= \emptyset,s_t^R\neq \emptyset:}$
 We again consider all possibles cases depending on the location of request $r_{t+1}$. Note that the exact same argument as above shows that the value of $g_{t+1}^L, g_{t+1}^R, \delta_{t+1}$, and the upper bound on $\Delta \text{cost}_{t+1}$ given in the last three columns of Table \ref{tab:delta_special} are identical to those in the last three  columns of Table \ref{tab:delta_general}, and that bullet points $1,2$ hold at time $t+1$ in these cases. We thus only need to show the result in the case $r_{t+1}\in [0,g_t^L +\tfrac{d_t^R+\delta_t}{2}]$.

\begin{figure}
    \centering
    \includegraphics[scale = 0.26]{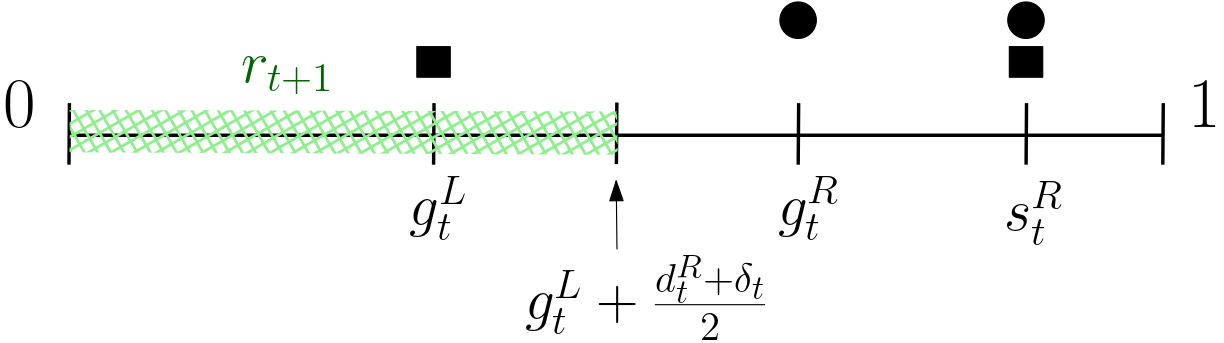}
    \caption{Illustration of the case $\mathbf{s_t^L= \emptyset,s_t^R\neq \emptyset:}$ and $r_{t+1}\in [0,g_t^L +\tfrac{d_t^R+\delta_t}{2}]$ in the proof of Lemma \ref{lem:structural_hybrid}.}
    \label{fig:cbis}
\end{figure}

In this case (see Figure~ \ref{fig:cbis}), we have $\mathcal{N}(r_{t+1}) \subseteq \{g_t^L, g_t^R\}$, $\mathcal{N}'(r_{t+1}) = \{g_t^R\}$, hence we immediately get $s(r_{t+1}) = g_t^L,s'(r_{t+1}) = g_t^R$. Combining this with the  induction hypothesis, we get: $S_{t+1} = S_{t}\setminus\{g_t^L\} = (S_{t}'\cup \{g_t^L\}\setminus\{g_t^R\})\setminus\{g_t^L\} = S_{t}'\setminus\{g_t^R\} = S_{t+1}'$,  which implies that $g_{t+1}^L = g_{t+1}^R = \emptyset$ and $\delta_{t+1} = 0$. In addition, $\Delta\text{cost}_{t+1}\leq |g_t^R - g_t^L| = \delta_t$. Hence, we have that bullet points $1,2$ hold at time $t+1$ and that the values of $g_{t+1}^L,g_{t+1}^R, \delta_{t+1}$ given in Table \ref{tab:delta_special} hold.

\vspace{0.2cm}
\noindent $\mathbf{s_t^R= \emptyset,s_t^L\neq \emptyset:}$
This case is symmetric to the  case $s_t^R= \emptyset,s_t^L\neq \emptyset$, by noting the one to one correspondence between $0,d_t^L, s_t^L, g_t^L$  and $1,d_t^R, s_t^R, g_t^R$.

\vspace{0.2cm}
\noindent $\mathbf{s_t^L=\emptyset, s_t^R= \emptyset}$: In this case, we have $S_t =\{g_t^L\}$ and $S_t' =\{g_t^R\}$. Hence, whatever the value of $r_{t+1}$, we get that $S_{t+1} = S_{t+1}' = \emptyset$ and $\delta_{t+1}=0$. In addition, we have $\Delta\text{cost}_{t+1}\leq |g_t^R- g_t^L|  = \delta_t$.

This concludes the proof that  bullet point $3$ is satisfied at time $t$ and that bullet points $1,2$ are satisfied at time $t+1$. Hence the three first bullet points of the lemma hold for all $t\in \{m,\ldots, n\}$. \\

Finally, we show bullet point $4$. Note that if $\delta_t = 0$ for some $t\in \{m,\ldots, n\}$, then by definition of $\delta_t$, we have $S_t =S_t'$. Since both $\H^m$ and $\H^{m-1}$ match greedily $r_{t+1}, \ldots, r_n$, and ties are broken consistently, we get $S_j= S_j'$ for all $j\in \{t,\ldots, n\}$, which shows that $\delta_j =0$ for $j\in\{t,\ldots, n\}$.
\end{proof}

In the remainder of this section, we assume that the structural properties proved in Lemma \ref{lem:structural_hybrid} hold.
By using the third and fourth bullet points of Lemma \ref{lem:structural_hybrid}, we now upper bound, for an arbitrary sequence of requests $R$, the total difference of cost incurred during the simultaneous execution of $\H^m$ and $\H^{m-1}$ by a function of the maximum gap $\delta_t$.

\begin{lemma}
\label{lem:diff_cost_delta} For any arbitrary set $S$ of $n$ servers and sequence $R$ of $n$ requests in $[0,1]$, we have \[
\text{cost}(\H^{m-1}) - \text{cost}(\H^{m}) \leq 2\max_{t\in \{m,\ldots, n-1\}}\delta_{t}.\]

\end{lemma}

\begin{proof}
Since $\H^m$ and $\H^{m-1}$ both match $r_1, \ldots, r_{m-1}$ to exactly the same servers as $\A$, we first have that $\text{cost}_t(\H^{m-1})- \text{cost}_t(\H^{m})=0$ for all $t\in [m-1]$. Then, since $\H^m$ matches $r_m$ to the same server as $\A$ while $\H^{m-1}$ matches $r_m$ greedily, we have that $|r_{m} - s'(r_m)| = \min\{|r_m-s|: s\in S_{\A,m-1}\}\leq |r_{m} - s(r_m)|$. Thus, $\text{cost}_m(\H^{m-1}) - \text{cost}_m(\H^{m})\leq 0$.

Next, we define $t_0 := \min \{t\geq m: \delta_t =0\}$, i.e., the first time step where the two sets of free servers become identical again. Note that by the fourth point of Lemma \ref{lem:structural_hybrid}, we have that $\delta_t=0$ for any $t\geq t_0$, which, by definition of $\delta$, implies that $S_t = S_t'$ for any $t\geq t_0$. We deduce that $\text{cost}_t(\H^{m-1})- cost_t(\H^{m})=0$ for any $t\in [t_0+1, \ldots, n]$. 

Finally, by a direct inspection of all possible cases enumerated in the third point of Lemma \ref{lem:structural_hybrid}, we get that for all $t\in \{m,\ldots, t_0-1\}$, $\text{cost}_t(\H^{m-1}) - \text{cost}_t(\H^{m})\leq \delta_{t} - \delta_{t-1}$, and we get that $cost_{t_0}(\H^{m-1}) - cost_{t_0}(\H^{m})\leq \delta_{t_0-1}$.

Putting everything together, we obtain
\begin{align*}
    \text{cost}(\H^{m-1}) - \text{cost} (\H^{m}) = \sum_{t=1}^{n}&(\text{cost}_t(\H^{m-1}) - \text{cost}_t(\H^{m}))\\
    &\leq 0 + \sum_{t=m+1}^{t_0-1}(\delta_t - \delta_{t-1}) + \delta_{t_0-1} + 0 = 2\cdot\delta_{t_0-1} - \delta_{m} \leq 2\max_{t\in \{m,\ldots, n-1\}}\delta_{t}.
\end{align*}

\end{proof}

In the remainder of the section, we consider the more specific case where the requests in $R$ are sampled uniformly at random in $[0,1]$. In the following lemma, we show that the probability that the distance $\delta_t$ between the (potential) extra server of the hybrid algorithm $\H^m$ and the (potential) extra server of  $\H^{m-1}$ will ever exceed $y\geq \delta_m$ at any time $t\in \{m,\ldots, n-1\}$ is upper bounded by $\tfrac{\delta_m}{y}$.
\begin{lemma}
\label{lem:worst_case}  Assume that requests $r_{m+1}, \ldots, r_n$ are sampled uniformly at random from $[0,1]$ and condition on $S_m,S'_m$. Then, we have that for any $y\in [\delta_m,1]$,
\[
\P\Big(\max_{t\in \{m,\ldots, n-1\}}\delta_t\geq y|(S_m, S'_m) \Big)\leq \frac{\delta_m}{y}.
\]
\end{lemma}

\begin{proof} We show by downward induction on $j$ that for any $j\in \{m,\ldots, n-1\}$, conditioning on $(S_m, S'_m), \ldots , (S_j, S'_j)$, we have that for any $y\in [\delta_j, 1]$,
\[
\P\Big(\max_{t\in \{j, \ldots, n-1 \}}\delta_t\geq y|(S_m, S'_m), \ldots , (S_j, S'_j)\Big)\leq \frac{\delta_j}{y}.
\]

We first show the base case, which is for $j=n-1$. It is immediate that
\[
\P\Big(\delta_{n-1}\geq y|(S_m, S'_m), \ldots , (S_{n-1}, S'_{n-1})\Big) = \begin{cases}
1 &\text{if } \delta_{n-1} = y\\
0 &\text{otherwise }
\end{cases} \leq\hspace{0.1cm} \frac{\delta_{n-1}}{y}.
\]

Next, let $j\in \{m,\ldots, n-2\}$, and assume that conditioning on $(S_m, S'_m), \ldots , (S_{j+1}, S'_{j+1})$, we have that for any $y\in [\delta_{j+1}, 1]$,
\[
\P\Big(\max_{t\in \{j+1, \ldots, n-1 \}}\delta_t\geq y|(S_m, S'_m), \ldots , (S_{j+1}, S'_{j+1})\Big)\leq \frac{\delta_{j+1}}{y}.
\]

Now, condition on $(S_m, S'_m), \ldots , (S_j, S'_j)$ and let $y\in [\delta_j, 1]$. Since the property is immediately true when $\delta_j = 0$ (by Lemma~\ref{lem:structural_hybrid}, bullet point 4) and $\delta_j=y$, we now assume that $0<\delta_j<y$.

Since $\delta_j \neq 0$, we have $S_j\neq S_j'$. Furthermore, by Lemma \ref{lem:structural_hybrid}, we have $|S_j\setminus S_j'| = |S_j'\setminus S_j| = 1$ with $S_j\Delta S_j' = \{g_j^L, g_j^R\}$. We assume without loss of generality that $S_j' = S_j \cup\{g_j^R\} \setminus\{g_j^L\}$. Recall that we defined $s_j^L = \max\{s\in S_j: s\leq g_j^L\}$,  $s_j^R = \min\{s\in S_j: s\geq g_j^R\}$, and $d_j^L = g_j^L - s_j^L$, $d_j^R = s_j^R - g_j^R$.

Next, we note the following fact: for any $r\in [0,1]$, letting $T(S_j, S'_j,r)$ and $T'(S_j, S'_j,r)$ be the value of $S'_{j+1}$ and $S_{j+1}$ conditioning on $(S_j, S'_j)$ and the event that $r_{j+1} = r$, we have that
\begin{align}
    &\P\Big(\max_{t\in \{j, \ldots, n-1 \}}\delta_t\geq y|(S_m, S'_m), \ldots , (S_j, S'_j), r_{j+1} = r\Big)\nonumber\\
    & = \P\Big(\max_{t\in \{j, \ldots, n-1 \}}\delta_t\geq y|(S_m, S'_m), \ldots , (S_j, S'_j), (S_{j+1}, S'_{j+1}) = (T(S_j, S'_j,r), T'(S_j, S'_j,r)), r_{j+1} = r\Big)\nonumber\\
    & = \P\Big(\max_{t\in \{j+1, \ldots, n-1 \}}\delta_t\geq y| (S_m, S'_m), \ldots , (S_j, S'_j), (S_{j+1}, S'_{j+1}) = (T(S_j, S'_j,r), T'(S_j, S'_j,r)), r_{j+1} = r\Big)\nonumber\\
     & = \P\Big(\max_{t\in \{j+1, \ldots, n -1\}}\delta_t\geq y| (S_m, S'_m), \ldots , (S_j, S'_j), (S_{j+1}, S'_{j+1}) = (T(S_j, S'_j,r), T'(S_j, S'_j,r))\Big)\nonumber\\
    &\leq \frac{\delta(T(S_j, S'_j,r), T'(S_j, S'_j,r))}{y},\label{eq:markov_1}
\end{align}

where the second equality is since $\delta_{j} <y$, the third equality is since conditioning on $S_{j+1}, S'_{j+1}$, we have that $\{\delta_t\}_{t\in \{j+1, \ldots, n\}}$ is independent on $r_{j+1}$, and the  inequality is by the induction hypothesis.

We now enumerate all possible cases depending on request $r_{j+1}$. We start by the case where $s_j^L,s_j^R\neq \emptyset$. By Lemma \ref{lem:structural_hybrid}, the values of $\delta(T(S_j, S'_j,r), T'(S_j, S'_j,r))$ are the ones given in Table \ref{tab:delta_general}.

\begin{itemize}
    \item Case 1: $r_{j+1}\in s_j^L +[0,\tfrac{d_j^L}{2}]$. We have, by Table \ref{tab:delta_general}, that $\delta(T(S_j, S'_j,r), T'(S_j, S'_j,r))= \delta_j $. Thus, by (\ref{eq:markov_1}), we get
    \[
    \P\Big(\max_{t\in \{j, \ldots, n \}}\delta_t\geq y|(S_m, S'_m), \ldots , (S_j, S'_j), r_{j+1}\in s_j^L +[0,\tfrac{d_j^L}{2}]\Big) \leq \frac{\delta_j}{y}.
    \]

    \item Case 2: $r_{j+1}\in s_j^L + [ \tfrac{d_j^L}{2}, \tfrac{d_j^L+\delta_j}{2}]$. We have $\delta(T(S_j, S'_j,r), T'(S_j, S'_j,r)) = \delta_j +d_j^L$. Thus, by (\ref{eq:markov_1}), we get
    \[
    \P\Big(\max_{t\in \{j, \ldots, n \}}\delta_t\geq y|(S_m, S'_m), \ldots , (S_j, S'_j), r_{j+1}\in s_j^L + [ \tfrac{d_j^L}{2}, \tfrac{d_j^L+\delta_j}{2}]\Big) \leq \frac{\delta_j + d_j^L}{y}.
    \]

    \item Case 3: $r_{j+1}\in s_j^L + [\tfrac{d_j^L+\delta_j}{2}, d_j^L + \tfrac{d_j^R+\delta_j}{2}]$. We have that $\delta(T(S_j, S'_j,r), T'(S_j, S'_j,r)) =0$. Thus, by (\ref{eq:markov_1}), we get 
    \[
    \P\Big(\max_{t\in \{j, \ldots, n \}}\delta_t\geq y|(S_m, S'_m), \ldots , (S_j, S'_j), r_{j+1}\in s_j^L + [\tfrac{d_j^L+\delta_j}{2}, d_j^L + \tfrac{d_j^R+\delta_j}{2}]\Big) \leq 0.
    \]
    
    \item Case 4: $r_{j+1}\in s_j^L + [d_j^L +\tfrac{d_j^R+\delta_j}{2}, d_j^L +\delta_j + \tfrac{d_j^R}{2}]$. We have $\delta(T(S_j, S'_j,r), T'(S_j, S'_j,r)) = \delta_j +d_j^R$. Thus, by (\ref{eq:markov_1}), we get
    \[
    \P\Big(\max_{t\in \{j, \ldots, n \}}\delta_t\geq y|(S_m, S'_m), \ldots , (S_j, S'_j), r_{j+1}\in s_j^L + [d_j^L +\tfrac{d_j^R+\delta_j}{2}, d_j^L +\delta_j + \tfrac{d_j^R}{2}]\Big) \leq \frac{\delta_j+d_j^R}{y}.
    \]
    
    \item Case 5: $r_{j+1}\in s_j^L + [d_j^L +\delta_j + \tfrac{d_j^R}{2}, d_j^L +\delta_j + d_j^R]$. We have $\delta(T(S_j, S'_j,r), T'(S_j, S'_j,r)) = \delta_j $. Thus, by (\ref{eq:markov_1}), we get
    \[
    \P\Big(\max_{t\in \{j, \ldots, n \}}\delta_t\geq y|(S_m, S'_m), \ldots , (S_j, S'_j), r_{j+1}\in s_j^L + [d_j^L +\delta_j + \tfrac{d_j^R}{2}, d_j^L +\delta_j + d_j^R]\Big) \leq \frac{\delta_j}{y}.
    \]
    
    \item Case 6: $r_{j+1}\in [0,s_j^L)\cup (s_j^R,1]$. We have $\delta(T(S_j, S'_j,r), T'(S_j, S'_j,r)) = \delta_j$. Thus, by (\ref{eq:markov_1}), we get
    \[
    \P\Big(\max_{t\in \{j, \ldots, n \}}\delta_t\geq y|(S_m, S'_m), \ldots , (S_j, S'_j), r_{j+1}\in [0,s_j^L)\cup (s_j^R,1]\Big) \leq \frac{\delta_j}{y}.
    \]
\end{itemize}

By combining the six cases above and using the fact that $r_{j+1}$ is drawn uniformly at random in $[0,1]$, we get
\begin{align*}
    &\P\Big(\max_{t\in \{j, \ldots, n-1 \}}\delta_t\geq y|(S_m, S'_m), \ldots , (S_j, S'_j)\Big) \\
    \hspace{0.1cm} &\leq \P(r_{j+1}\in s_j^L +[0,\tfrac{d_j^L}{2}])\cdot\frac{\delta_j}{y} 
    + \P(r_{j+1}\in s_j^L + [ \tfrac{d_j^L}{2}, \tfrac{d_j^L+\delta_j}{2}])\cdot\frac{\delta_j+d_j^L}{y}\\
     \hspace{0.1cm} &+  \P(r_{j+1}\in s_j^L + [\tfrac{d_j^L+\delta_j}{2}, d_j^L + \tfrac{d_j^R+\delta_j}{2}])\cdot 0 \\
    \hspace{0.1cm} &+ \P(r_{j+1}\in s_j^L + [d_j^L +\tfrac{d_j^R+\delta_j}{2}, d_j^L +\delta_j + \tfrac{d_j^R}{2}])\cdot\frac{\delta_j+d_j^R}{y}\\
    \hspace{0.1cm} &+ \P(r_{j+1}\in s_j^L + [d_j^L +\delta_j + \tfrac{d_j^R}{2}, d_j^L +\delta_j + d_j^R])\cdot\frac{\delta_j}{y}
    + \P(r_{j+1}\in [0,s_j^L)\cup (s_j^R,1])\cdot\frac{\delta_j}{y}\\
    \hspace{0.1cm} &= \frac{d_j^L}{2}\cdot\frac{\delta_j}{y} + \frac{\delta_j}{2}\cdot\frac{\delta_j+d_j^L}{y} + \frac{\delta_j}{2}\cdot\frac{\delta_j+d_j^R}{y} + \frac{d_j^R}{2}\cdot\frac{\delta_j}{y} + (1- (d_j^L+d_j^R+\delta_j))\cdot\frac{\delta_j}{y}\\
    \hspace{0.1cm} &=  \frac{\delta_j}{y}\cdot\Big(\frac{d_j^L}{2} + \frac{\delta_j+d_j^L}{2} + \frac{\delta_j+d_j^R}{2} + \frac{d_j^R}{2}+ (1- (d_j^L+d_j^R+\delta_j))\Big)\\
    \hspace{0.1cm} & = \frac{\delta_j}{y}.
\end{align*}

 We now consider the case where $s_j^L = \emptyset,s_j^R\neq \emptyset$. By Lemma \ref{lem:structural_hybrid}, the values of $\delta(T(S_j, S'_j,r), T'(S_j, S'_j,r))$ are the ones given in Table \ref{tab:delta_special}. We consider four different cases.
 \begin{itemize}
     \item Case 1: $r_{j+1}\in [0,g_j^L +\tfrac{d_j^R+\delta_j}{2}]$. We have, by Table \ref{tab:delta_special} that $\delta(T(S_j, S'_j,r), T'(S_j, S'_j,r)) =0$. Thus, by (\ref{eq:markov_1}), we get 
    \[
    \P\Big(\max_{t\in \{j, \ldots, n \}}\delta_j\geq y|(S_m, S'_m), \ldots , (S_j, S'_j), r_{j+1}\in [0,g_j^L +\tfrac{d_j^R+\delta_j}{2}]\Big) \leq 0.
    \]
     \item Cases 2,3,4: the upper bounds we get in the cases $r_{j+1}\in[g_j^L +\tfrac{d_j^R+\delta_j}{2}, g_j^L + \delta_j + \tfrac{d_j^R}{2}]$, $r_{j+1}\in[g_j^L + \delta_j + \tfrac{d_j^R}{2}, g_j^L + \delta_j + d_j^R]$ and $r_{j+1}\in(s_j^R,1]$ are identical as the ones in Cases 4,5,6 when $s_j^L,s_j^R\neq \emptyset$. 
 \end{itemize}
By combining the four cases above and using the fact that $r_{j+1}$ is drawn uniformly at random in $[0,1]$, we get
\allowdisplaybreaks{
\begin{align*}
    &\P\Big(\max_{t\in \{j, \ldots, n -1\}}\delta_j\geq y|(S_m, S'_m), \ldots , (S_j, S'_j)\Big) \\
    \hspace{0.1cm} & \leq \P(r_{j+1}\in [0,g_j^L +\tfrac{d_j^R+\delta_j}{2}])\cdot 0 \\
    \hspace{0.1cm}  &+ \P(r_{j+1}\in[g_j^L +\tfrac{d_j^R+\delta_j}{2}, g_j^L + \delta_j + \tfrac{d_j^R}{2}])\cdot\frac{\delta_j+d_j^R}{y}\\
    \hspace{0.1cm} &+ \P(r_{j+1}\in[g_j^L + \delta_j + \tfrac{d_j^R}{2}, g_j^L + \delta_j + d_j^R])\cdot\frac{\delta_j}{y}
    + \P(r_{j+1}\in(s_j^R,1])\cdot\frac{\delta_j}{y}\\
    \hspace{0.1cm} &=  \frac{\delta_j}{2}\cdot\frac{\delta_j+d_j^R}{y} + \frac{d_j^R}{2}\cdot\frac{\delta_j}{y} + (1- (d_j^L+d_j^R+\delta_j))\cdot\frac{\delta_j}{y}\\
    \hspace{0.1cm} &=  \frac{\delta_j}{y}\cdot\Big( \frac{\delta_j+d_j^R}{2} + \frac{d_j^R}{2}+ (1- (d_j^L+d_j^R+\delta_j))\Big)\\
    \hspace{0.1cm} & \leq \frac{\delta_j}{y}.
\end{align*}
}

The case $s_j^L\neq \emptyset,s_j^R= \emptyset$ is similar to the case above and we conclude in the same way.
Finally, in the case $s_j^L,s_j^R = \emptyset$, we have by Lemma \ref{lem:structural_hybrid} that for any $r_{j+1}$, $\delta(T(S_j, S'_j,r), T'(S_j, S'_j,r)) = 0$. Thus, by (\ref{eq:markov_1}), we get 
    \[
    \P\Big(\max_{t\in \{j, \ldots, n-1 \}}\delta_t\geq y|(S_m, S'_m), \ldots , (S_j, S'_j)\Big) \leq 0.
    \]
    
Hence, in all cases, we have shown that 
\[
\P\Big(\max_{t\in \{j, \ldots, n-1 \}}\delta_j\geq y|(S_m, S'_m), \ldots , (S_j, S'_j)\Big)\leq \frac{\delta_j}{y},
\]
which concludes the inductive case and the proof of the lemma.
\end{proof}

\begin{lemma}
\label{lem:upper_bound_delta_cost}
Assuming that the requests are sampled uniformly at random from $[0,1]$, we have, for some constant $C>0$:
\begin{align*}
    &\E\Big[\delta_m(1+\log(1/\delta_m))|S_{m-1}, r_m, \delta_m>0\Big]\cdot \P(\delta_m>0|S_{m-1}, r_m)\\
   & \leq C\cdot\E\Big[\big(1+\log\big(\tfrac{1}{\text{cost}_m(\A)}\big)\big)\text{cost}_m(\A)\big|S_{m-1}, r_m\Big].
\end{align*}
\end{lemma}

\begin{proof}
We first condition on the realization of variables $S_{m-1}, r_m$ and on the event $\{\delta_m>0\}$.

Recall that $\H^m$ and $\H^{m-1}$ both match $r_1, \ldots, r_{m-1}$ to the same servers as $\A$, and that $\H^m$ matches $r_m$ to the same server $s(r_m)$ as $\A$, while $\H^{m-1}$ matches $r_m$ greedily to $s'(r_m)$. Hence, we have that $|r_{m} - s(r_m)| = \text{cost}_m(\A)$ and that $|r_{m} - s'(r_m)| = \min\{|r_m-s|: s\in S_{\A,m-1}\}\leq |r_{m} - s(r_m)|$.
Thus, 
\[
\delta_m = |s(r_m) - s'(r_m)| \leq |r_{m} - s(r_m)| + |r_{m} - s'(r_m)|\leq 2|r_{m} - s(r_m)| =2 \text{cost}_m(\A).
\]

Now, note that $x\longmapsto x\log(\frac{1}{x})$ reaches its maximum value over $(0,1]$ at $x = 1/e$ (with $\frac{1}{e}\log(\frac{1}{1/e}) = \frac{1}{e}$) and is non-decreasing on $(0,1/e]$. Since $\delta_m\in (0,1]$, we thus have
\begin{align*}
    \delta_m(1+ \log(1/\delta_m))
    &\leq& 2 \text{cost}_m(\A) + 
\begin{cases}
2 \text{cost}_m(\A)\log(1/(2 \text{cost}_m(\A))) &\text{if } 2 \text{cost}_m(\A)\in (0,1/e]\\
1/e &\text{if } 2 \text{cost}_m(\A)\in [1/e,1].
\end{cases}\\
&\leq &\;2 \text{cost}_m(\A)\big(1+\log(1/(2 \text{cost}_m(\A))) + 1\big).
\end{align*}

As a result, we get that for some $C>0$,
\begin{align*}
    &\E\Big[\delta_m(1+\log(1/\delta_m))|S_{m-1}, r_m, \delta_m>0\Big]\cdot \P(\delta_m>0|S_{m-1}, r_m)\\
    &\leq C\cdot \E\Big[\big(1+\log\big(\tfrac{1}{\text{cost}_m(\A)}\big)\big)\text{cost}_m(\A)\big|S_{m-1}, r_m,\delta_m>0\Big]\cdot \P(\delta_m>0|S_{m-1}, r_m)\\ 
    &\leq C\cdot \E\Big[\big(1+\log\big(\tfrac{1}{\text{cost}_m(\A)}\big)\big)\text{cost}_m(\A)\big|S_{m-1}, r_m\Big].
\end{align*}
\end{proof}

We are now ready to present the proof of the hybrid lemma.
\begin{proof}[Proof of Lemma \ref{lem:hybrid}.]
Let $S$ be an arbitrary set of $n$ servers in $[0,1]$ and $R$ a sequence of $n$ requests drawn uniformly at random from $[0,1]$. In the remainder of the proof, we consider a simultaneous execution  of $\H^m$ and $\H^{m-1}$ with initial set of servers $S$ and requests $R$.

Conditioning on the variables $S'_m, S_m, S_{m-1}, r_m$, we have
\begin{align}
&\E[cost(\H^{m-1}) - cost(\H^{m})|S'_m, S_m, S_{m-1}, r_m]\nonumber\\
\leq\hspace{0.1cm} & 2\cdot \E\big[\max_{t\in \{m,\ldots, n-1\}}\delta_{t}|S'_m, S_m, S_{m-1}, r_m\big]&\text{Lemma \ref{lem:diff_cost_delta}}\nonumber\\
=\hspace{0.1cm} & 2\cdot \E\big[\max_{t\in \{m,\ldots, n-1\}}\delta_{t}|S'_m, S_m\big]&\text{$\max_{t\geq m}\delta_{t}\indep (S_{m-1},r_m)$ when $|(S'_m,S_m)$} \nonumber \\
=\hspace{0.1cm} & \mathbbm{1}_{\delta_m>0}\cdot 2\cdot\E\big[\max_{t\in \{m,\ldots, n-1\}}\delta_{t}|S'_m, S_m\big]&\text{Lemma \ref{lem:structural_hybrid} (bullet point 4)}\nonumber\\
\leq\hspace{0.1cm} &  \mathbbm{1}_{\delta_m>0}\cdot2\cdot\Big(\delta_m +  \int_{\delta_m}^{1} \mathbb{P}_{R}\big[\max_{t\in \{m,\ldots, n-1\}}\delta_{t}\geq y|S'_m, S_m\big]\text{dy}\Big)\nonumber\\
%
\leq\hspace{0.1cm} &\mathbbm{1}_{\delta_m>0}\cdot2\cdot\Big(\delta_m + \int_{\delta_m}^{1} \frac{\delta_m}{y}\text{dy}\Big) &\text{Lemma \ref{lem:worst_case}}\nonumber\\
 =\hspace{0.1cm} &\mathbbm{1}_{\delta_m>0}\cdot 2\cdot\delta_m(1+\log(1/\delta_m))\label{eq:lemmahybridconclusion}.
\end{align}
By the tower rule, we conclude that
\begin{align}
    & \E[cost(\H^{m-1}) - cost(\H^{m})|S_{m-1}, r_m] \nonumber\\
    =\hspace{0.1cm}&\E[\E[cost(\H^{m-1}) - cost(\H^{m})| S'_m, S_m, S_{m-1}, r_m]|S_{m-1}, r_m] \nonumber\\
    \leq\hspace{0.1cm}& \E[\mathbbm{1}_{\delta_m>0}\cdot 2\cdot\delta_m(1+\log(1/\delta_m))|S_{m-1}, r_m]&\text{by (\ref{eq:lemmahybridconclusion})}\nonumber\\
    =\hspace{0.1cm}& 2\cdot \E[\delta_m(1+\log(1/\delta_m))|S_{m-1}, r_m, \delta_m>0]\cdot \P(\delta_m>0| S_{m-1}, r_m)\nonumber\\
    \leq \hspace{0.1cm}& 2C\cdot\E\left[\big(1+\log\big(\tfrac{1}{\text{cost}_m(\A)}\big)\big)\text{cost}_m(\A)\big|S_{m-1}, r_m\right]. & \text{Lemma \ref{lem:upper_bound_delta_cost}}\nonumber \qquad \qedhere
\end{align}
\end{proof}

%% file: appendix_random_model.tex
\section{Missing Analysis from Section \ref{section:upper_bound_iid}}

\label{app_random_model}

\paragraph{The excess supply setting.}

We first recall some notations, that will be used throughout this section.
\begin{itemize}
    \item For any $\ell, m\in [0,1]$, we let $\xlr = |\{t\in [n-1]: r_t\in (\ell,m)\}|$ be the number of requests out of the $n-1$ first requests that arrived in the interval $(\ell,m)$.
    \item  For any $\ell, m\in [0,1]$, we let $\ylr = |\{t\in [n(1+\epsilon)]: s_t\in (\ell,m)\}|$ be the total number of servers that lie in the interval $(\ell,m)$.

\end{itemize}

We now prove a couple of  lemmas. The first one upper bounds the probability that for some given $z\in [ \tfrac{4(1+\epsilon/4)}{\epsilon n},1]$,  there exists an interval $I$ of length large enough w.r.t.~$z$ such that $r_n \in I$  and  the number of servers and of requests that arrived strictly before $r_n$ and lying in $I$ are equal.

\lemdiscretization*

\begin{proof}
Throughout the proof, we condition on the random variable $r_n$. We start by discretizing the interval $[0,1]$, and first let 
\[
j_0 = \begin{cases} \max \{j \in \mathbb{Z}_{\geq 0}: r_n - z - \tfrac{j}{n} \geq 0\} &\text{if } r_n \geq z\\
-1 &\text{otherwise.}
\end{cases} 
\]
and
\[k_0 = \begin{cases} \min\{k \in \mathbb{Z}_{\geq 0}, r_n +z + \tfrac{k}{n} \leq 1\} &\text{if } 1 - r_n \geq z \\
-1 &\text{otherwise.}
\end{cases}
\]

Consider the case $j_0, k_0\neq -1$. For all $j\in \{0,\ldots,j_0\}$,  we  let $\ell_j := r_n - z - \tfrac{j}{n}$, and for all $k\in \{0,\ldots,k_0\}$, we let $m_k := r_n + z +\tfrac{k}{n}$. We also let $\ell_{j_0+1} = 0$,  and $ m_{k_0+1} = 1$.

Now, consider any pair $(j,k)\in \{0,\ldots,j_0\}\times \{0,\ldots,k_0\}$. First, note that for any realization $R$ of the sequence of requests, and for any $\ell\in [\ell_{j+1},\ell_j], m\in [m_k, m_{k+1}]$, we have 
\begin{equation}
\label{eq:xlrylr}
    \xlr\leq x_{(\ell_{j+1},m_{k+1})} \quad \text{ and } \quad \ylr \geq y_{(\ell_{j},m_{k})}.
\end{equation}

Then, note that $x_{(\ell_{j+1},m_{k+1})}$ follows a binomial distribution $\B(n-1, m_{k+1}-\ell_{j+1})$ and that $y_{(\ell_j,m_k)}$ follows a binomial distribution $\B(n(1+\epsilon), m_k-\ell_j)$. Thus, by Chernoff bounds (Lemma~\ref{prop:chernoff_bounds}), we get that for some $C_\epsilon>0$, 
\begin{align}
\label{eq:cher1}
\P(x_{(\ell_{j+1},m_{k+1})}\geq  (n-1)(m_{k+1}- \ell_{j+1})(1+\epsilon/4))& \leq \ e^{- (n-1)(m_{k+1}- \ell_{j+1})\frac{(\epsilon/4)^2}{(2+\epsilon/4)}} \nonumber \\ & \leq \ e^{- (n-1)(m_{k}- \ell_{j})C_\epsilon},
\end{align}

and 
\begin{align}
\label{eq:cher2}
\P(y_{(\ell_j,m_k)}\leq  n(1+\epsilon)(m_{k}- \ell_{j})\big(1-\tfrac{\epsilon}{4(1+\epsilon)}\big))& \leq \  e^{- \frac{n(1+\epsilon)(m_{k}- \ell_{j}))(\frac{\epsilon}{4(1+\epsilon)})^2}{2}} \nonumber \\ & \leq \ e^{- (n-1)(m_{k}- \ell_{j})C_\epsilon}.
\end{align}

Next, since $z\geq \tfrac{4(1+\epsilon/4)}{\epsilon n}$, we have that
\begin{align}
\label{eq:diffcher}
    n(1+\epsilon)(m_{k}- \ell_{j})(1-\tfrac{\epsilon}{4(1+\epsilon)})) &- (n-1)(m_{k+1}- \ell_{j+1})(1+\epsilon/4)\nonumber\\ 
    &\geq n(1+\epsilon)(m_{k}- \ell_{j})(1-\tfrac{\epsilon}{4(1+\epsilon)})) - n(m_{k+1}- \ell_{j+1})(1+\epsilon/4)\nonumber\\
&\geq n(m_{k}- \ell_{j})(1+3\epsilon/4)) - n(m_{k}- \ell_{j} + \tfrac{2}{n})(1+\epsilon/4)\nonumber\\
 &= n(m_{k}- \ell_{j})\epsilon/2- 2(1+\epsilon/4)\nonumber\\
    &\geq nz\epsilon/2- 2(1+\epsilon/4)\nonumber\\
    &\geq 0.
\end{align}

Hence, if we have that $x_{(\ell_j,m_k)}<  (n-1)(m_{k+1}- \ell_{j+1})(1+\epsilon/4)$ and that $y_{(\ell_j,m_k)}> n(1+\epsilon)(m_k- \ell_{j})\big(1-\tfrac{\epsilon}{4(1+\epsilon)}\big)$, then, we obtain by (\ref{eq:xlrylr}) and (\ref{eq:diffcher}) that for all $\ell\in [\ell_{j+1},\ell_j], m\in [m_k, m_{k+1}]$:
\[\xlr\leq x_{(\ell_{j+1},m_{k+1})}< (n-1)(m_{k+1}- \ell_{j+1})(1+\epsilon/4) \leq n(1+\epsilon)(m_k- \ell_{j})\big(1-\tfrac{\epsilon}{4(1+\epsilon)}\big) < y_{(\ell_j,m_k)} \leq \ylr.
\]
Thus, we have:
\begin{align*}
&\quad\P(\exists \ell\in [\ell_{j+1},\ell_j], m\in [m_k, m_{k+1}] : x_{(\ell,m)} = y_{(\ell,m)}|\;r_n)\\
 &\leq \P(x_{(\ell_{j+1},m_{k+1})}\geq  (n-1)(m_{k+1}- \ell_{j+1})(1+\epsilon/4)|\;r_n) \\
 &\qquad\qquad\qquad\qquad\qquad\qquad\qquad\qquad\qquad+ \P(y_{(\ell_j,m_k)}\leq  n(1+\epsilon)(m_k- \ell_{j})(1-\tfrac{\epsilon}{4(1+\epsilon)})|\;r_n)\\
 &\leq 2e^{- (n-1)(m_k- \ell_{j})C_\epsilon},
\end{align*}
where the last inequality is by (\ref{eq:cher1}) and (\ref{eq:cher2}), and since $x_{(\ell_j,m_k)}$ and $y_{(\ell_j,m_k)}$ are independent of $r_n$.

By union bound over all $(j,k)\in \{0,\ldots,j_0\}\times \{0,\ldots,k_0\}$, we obtain that for some constant $C'_\epsilon>0$:
\begin{align}
\P(\exists \ell, m\in [0,1]& : x_{(\ell,m)} = y_{(\ell,m)}, ( r_n- \ell\geq z \text{ or  } \ell=0), (m - r_n\geq z\text{ or } m=1 )\;|\;r_n)\nonumber\\
&=
\P\Big(\bigcup_{j,k\in \{0,\ldots,j_0\}\times \{0,\ldots,k_0\}}\{\exists \ell\in [\ell_{j+1},\ell_j], m\in [m_k, m_{k+1}] : x_{(\ell,m)} = y_{(\ell,m)}\}|\;r_n\Big)\nonumber\\
&\leq \sum_{j,k\in \{0,\ldots,j_0\}\times \{0,\ldots,k_0\}}
\P(\exists \ell\in [\ell_{j+1},\ell_j], m\in [m_k, m_{k+1}] : x_{(\ell,m)} = y_{(\ell,m)}|\;r_n)\nonumber\\
& \leq \sum_{j,k\in \{0,\ldots,j_0\}\times \{0,\ldots,k_0\}} 2e^{- (n-1)(m_k- \ell_{j})C_\epsilon}\nonumber\\
 & = 2e^{- 2(n-1)zC_\epsilon}\sum_{j,k\in \{0,\ldots,j_0\}\times \{0,\ldots,k_0\}} e^{- (n-1)(\frac{j+k}{n})C_\epsilon}\nonumber\\
 &= C_{\epsilon}' e^{-2 nzC_\epsilon}\label{eq:discretization},
\end{align}
where the last inequality holds since for any $C_\epsilon>0$, $\sum_{j=1,\dots,+\infty,k=1,\dots,+\infty} e^{- (j+k)C_\epsilon}$ converges.

Next, we consider the case where $k_0= -1$ or $j_0= -1$. Now, by a similar argument as above, we have that if $j_0= -1$ and $k_0\neq -1$, then for all $k\in \{0,\ldots, k_0\}$, \[\P(\exists m\in [m_k, m_{k+1}] : x_{(0,m)} = y_{(0,m)}|\;r_n)\leq 2e^{- n(m_k- 0)C_\epsilon},\]
and if $k_0= -1$ and $j_0\neq -1$, then for all $j\in \{0,\ldots, j_0\}$,
\[\P(\exists \ell\in [\ell_{j+1}, \ell_j] : x_{(\ell,1)} = y_{(\ell,1)}|\;r_n)\leq 2e^{- n(1- \ell_j)C_\epsilon}.\]
We conclude in a similar way as in (\ref{eq:discretization}) that there are $C_{\epsilon}, C_{\epsilon}'>0$ such that
$$
\P(\exists \ell, m\in [0,1] : x_{(\ell,m)} = y_{(\ell,m)}, ( r_n- \ell\geq z \text{ or  } \ell=0), (m - r_n\geq z\text{ or } m=1 )\;|\;r_n)\leq C_{\epsilon}' e^{-nzC_\epsilon}.
$$
Finally, if $j_0= -1$ and $k_0= -1$, note that $1-z<r_n<z$. In this case, since $x_{(0,1)} = n-1$ and $y_{(0,1)} = n(1+\epsilon)$, we simply have 
\begin{align*}
   &\P(\exists \ell, m\in [0,1] : x_{(\ell,m)} = y_{(\ell,m)}, ( r_n- \ell\geq z \text{ or  } \ell=0), (m - r_n\geq z\text{ or } m=1 )\;|\;r_n) \\
   &=\P(x_{(0,1)} = y_{(0,1)}|\;r_n) = 0 \leq C_{\epsilon}'e^{-nzC_\epsilon}.
\end{align*}

\end{proof}

\lemgreedycostincreases*
\begin{proof} Recall that $S_0\supseteq\ldots\supseteq\S_n$ denote the sequence of set of free servers obtained during the execution of $\G$. 
Since for all $i\in [n-1]$, $\mathcal{S}_{i}\subseteq\mathcal{S}_{i-1}$, we have that for all $S\subseteq[0,1]$ with $|S| = n-(i-1)$ and $S'\subseteq S$ such that $|S'| = |S|-1$:
\begin{align*}
    \E[cost_{i}(\G)|S_{i-1} = S] &= \E_{r_{i}\sim \mathcal{U}[0,1]}[\min_{s\in S}|r_{i}-s|]\\
    &\leq \E_{r_{i+1}\sim \mathcal{U}[0,1]}[\min_{s\in S'}|r_{i+1}-s|] \\&= \E[cost_{i+1}(\G)|S_{i} = S']\\
    &= \E[cost_{i+1}(\G)|S_{i} = S', S_{i-1} = S],
\end{align*}

where the last equation holds since conditioning on $S_i$, the matching decision for $r_{i+1}$ is independent of $S_{i-1}$.

Hence, by applying a first time the tower rule over $S_{i-1}$, we get
$\E[cost_{i}(\G)]\leq \E[cost_{i+1}(\G)|S_{i} = S']
$, and by applying it a second time over $S_{i}$, we get
\[\E[cost_{i}(\G)]\leq \E[cost_{i+1}(\G)].\] 
\end{proof}

%% file: appendix_UB_semi_random.tex
\section{Missing Analysis from Section \ref{section:sim_random}}

\label{app_semi_random_UB}

\lemunifcost*

\begin{proof}
We first lower bound the  probability that $\text{cost}_t(\A)\geq 1/n^4$. Conditioning on $S_{t-1}$, we have
\begin{align*}
   \P(\text{cost}_t(\A)< 1/n^4\;|\;S_{t-1} ) & \leq \P(\exists s\in S_{t-1}: |r_t-s|< 1/n^4\;|\;S_{t-1} )\\
   &= \P\Big(\bigcup_{s\in S_{t-1}}\{ r_t\in [\max(0,s-1/n^4), \min(1, s+1/n^4)]\}\;|\;S_{t-1}\Big)\\
& \leq \sum_{s\in S_{t-1}} \P(r_t\in [\max(0,s-1/n^4), \min(1, s+1/n^4)]\;|\;S_{t-1})\\
&\leq \sum_{s\in S_{t-1}} 2/n^4\\
&\leq 2/n^3.
\end{align*}
Thus,  $\P(\text{cost}_t(\A)< 1/n^4)\leq 2/n^3$.

Next, we lower bound $\E[\text{cost}_t(\A)]$. We condition on $S_{t-1}$ and let $0 \leq s_{t,1}\leq\ldots \leq s_{t,n}\leq 1 $ denote the ordered servers of $S_{t-1}$. By convention, we also write $s_{t,0}=0$, $s_{t,n+1}=1$. Then,
\begin{align*}
    \mathbb{E}[\text{cost}_t(\A)|S_{t-1}]&= \sum_{i=0}^n \mathbb{P}(r_t \in [s_{t,i}, s_{t,i+1}]) \mathbb{E}[\text{cost}_t(\A)|r_t \in [s_{t,i}, s_{t,i+1}]]\\
    &= \sum_{i=0}^n \frac{(s_{t,i+1}-s_{t,i})^2}{2}.
\end{align*} 
Since $\sum_{i=0}^n (s_{t,i+1}-s_{t,i}) = s_{t,n+1} - s_{t,0} = 1$, the above sum is minimized when $s_{t,i+1}-s_{t,i} = 1/(n+1)$ for all $i$, and the minimum  value is $\tfrac{1}{2(n+1)}$. By the tower law, we deduce that $\mathbb{E}[cost_t(\A])\geq \frac{1}{2(n+1)}$.
\end{proof}

In the remainder of this section, we demonstrate the existence of a constant competitive algorithm for the random requests model that makes neighboring matches. We first show that we can always transform any algorithm into an algorithm which satisfies this last property without increasing the total cost.

\begin{restatable}{rLem}{lemnearestneighbours}
\label{lem:nearest_neighbours}
For any online algorithm $\mathcal{A}$, there exists an algorithm $\mathcal{A}'$ that makes neighboring matches such that $\E[cost(\mathcal{A})']\leq \E[cost(\mathcal{A})]$.
\end{restatable}

\begin{proof}
We show the result by induction on $t$. Let $t_0 \geq 0$ and suppose that $\mathcal{A}$ is an  online algorithm that makes neighboring matches for all $t\leq t_0$, but does not necessarily make neighboring matches when $t>t_0$. Without loss of generality, assume that $r_{t_0+1}$ is matched by $\mathcal{A}$ to an available server $s\in S_{\A,t_0}$ such that $r_{t_0+1}\leq s$. Now, let $s' = \min\{z\in S_{\A,t_0}: z\geq r_{t_0+1}\} $ denote the closest available server on the right of $r_{t_0+1}$, and let $j\in \{t_0+1,\ldots n\}$ be such that $\A$ matches request $r_j$ to $s'$. We define the algorithm $\mathcal{A}'$ that matches all requests $r_t$ to exactly the same servers $\mathcal{A}$ matches them to for all $t\neq j, t_0+1$, matches $r_{t_0+1}$ to $s'$ and matches $r_j$ to $s$ (this is a valid construction since $s'$ is available when $r_{t_0+1}$ arrives and $s$ is available when $r_j$ arrives). Then by construction, $\mathcal{A}'$ makes neighboring matches for all  $t\leq t_0 + 1$.

We now analyse the cost of $\A'$.  Since $\mathcal{A}$ and $\mathcal{A}'$ match all requests other than $r_{t_0+1}$ and $r_j$ to the same servers, they incur the same cost for these requests. Now, we consider three cases: if $r_j\leq s'$, then
\begin{align*}
   \text{cost}_{t_0+1}(\A) + \text{cost}_{j}(\A) &= |r_{t_0+1}-s| + |r_j- s'| = (s - r_{t_0+1}) + (s' - r_j) \\
   & = (s' - r_{t_0+1}) + (s - r_j) =   |r_{t_0+1}- s'| + |r_j- s| = \text{cost}_{t_0+1}(\A') + \text{cost}_{j}(\A'), 
\end{align*}
if $s'\leq r_j\leq s$, then
\begin{align*}
   \text{cost}_{t_0+1}(\A) + \text{cost}_{j}(\A) &= |r_{t_0+1}- s| + |r_j- s'|\\
   &= ( s - r_{t_0+1}) + (r_j- s')| = ( s' - r_{t_0+1}) + (s-r_j) + 2(r_j - s')\\
   &\geq ( s' - r_{t_0+1}) + (s-r_j) =  |r_{t_0+1}- s'| + |r_j-s| = \text{cost}_{t_0+1}(\A') + \text{cost}_{j}(\A'),
\end{align*}
and if $r_j\geq s$, then
\begin{align*}
   \text{cost}_{t_0+1}(\A) + \text{cost}_{j}(\A) &= |r_{t_0+1}- s| + |r_j- s'|\\
   &= ( s - r_{t_0+1}) + (r_j-s')| = ( s' - r_{t_0+1}) + (r_j-s) + 2(s - s')\\
   &\geq ( s' - r_{t_0+1}) + (r_j-s) =  |r_{t_0+1}- s'| + |r_j-s| = \text{cost}_{t_0+1}(\A') + \text{cost}_{j}(\A').
\end{align*}
Hence, in all cases,  $\mathcal{A}'$ achieves a lower cost than $\mathcal{A}$ for requests $\{r_{t_0+1},r_j\}$.

Therefore, $\mathcal{A}'$ makes neighboring matches and is such that $\E[cost(\mathcal{A})']\leq \E[cost(\mathcal{A})]$.
\end{proof}

Next, we show that a simple adaptation of the algorithm Fair-Bias from \cite{GuptaGPW19} is a constant competitive algorithm for the random requests model. We first recall a result from \cite{GuptaGPW19}.

\begin{lemma}[Theorem 4.6. in \cite{GuptaGPW19}]\label{lem:fair_bias} Let $(S, d)$ be a tree metric with a server at each of the $n = |S| \geq 2$ points. Algorithm Fair-Bias is $9$-competitive if the requests are drawn from a known distribution $\mathcal{D}$ over the servers' locations.
\end{lemma}

Note that in the above lemma, the requests have support in the servers' locations, whereas, in the random requests model, we consider uniform requests in $[0,1]$. We show in the following lemma that we can nevertheless derive from Algorithm Fair-Bias a constant competitive algorithm $\A$ in the random requests model, and such that $\A$ makes neighboring matches.

\corggp*

\begin{proof} To ease the notations, we write $\mathcal{F}$ to denote the algorithm Fair-bias in the remainder of the proof.

Given an instance of the random requests model with set of servers $S$ and a realization of the requests sequence $R=\{r_1,\ldots, r_n\}$, we let $\Tilde{R} =\{\tilde{r_1},\ldots, \tilde{r_n}\}$ be such that for all $i\in [n]$, $\tilde{r_i}= \argmin_{s\in S}|r_i-s|$ is the closest server location to $r_i$. We consider the algorithm $\A$ that matches the requests in $R$ exactly to the same servers $\mathcal{F}$ matches the requests in $\tilde{R}$.  In order to analyze $\A$, we now define a distribution $\mathcal{D}$ over the servers locations, such that for all $s\in S$, \[\P_{r\sim\mathcal{D}}(r = s) = \P_{r\sim \mathcal{U}([0,1]) }\big(\argmin_{s'\in S}|r-s'| = s\big).\]
 Note that since for all $i \in [n]$, $r_i \sim \mathcal{U}([0,1])$, we have, by construction, that $\tilde{R}\sim \mathcal{D}$ when $R\sim \mathcal{U}([0,1])$.

We now show that $\A$ has a constant competitive ratio. Let $R$ be the realization of the requests and $\tilde{R}$ the corresponding transformed requests. We let $M_R$ be an optimal offline matching for $R$ and  $\OPT_R$ be the cost of this matching. We also let $s_{M_R}(i)$ be the server $M$ matches $i$ to. In addition, let $\OPT_{\tilde{R}}$ denote the cost of an optimal offline matching for $\tilde{R}$. Now, the cost of the matching returned by $\A$ satisfies: 
\begin{equation}
    \label{eq:cost_r_r'}
    \text{cost}(\A, (R,S)) =  \sum_{i=1}^n |r_i - s_{\A}(r_i)| \leq \sum_{i=1}^n (|r_i - \tilde{r_i}| + |\tilde{r_i} - s_{\A}(r_i)|) = \sum_{i=1}^n |r_i - \tilde{r_i}| + \text{cost}(\mathcal{F}, (\tilde{R}, S)).
\end{equation}
Since for all $i\in [n]$, we have that $|r_i - \tilde{r_i}| = \min_{s\in S} |r_i -  s|\leq |r_i - s_{M_R}(i)|$, we immediately get 
\begin{equation}
\label{eq:FB1}
    \sum_{i=1}^n |r_i - \tilde{r_i}|\leq \OPT_R.
\end{equation}

In addition, by considering the matching $\{(\tilde{r_i}, s_{M_R}(i))\}_{i\in [n]}$ for $\tilde{R}$, we get 
\[
\OPT_{\tilde{R}} \leq \sum_{i=1}^n |\tilde{r_i} - s_{M_R}(i)| \leq \sum_{i=1}^n (|\tilde{r_i} - r_i| + |r_i - s_{M_R}(i)|) \leq  \sum_{i=1}^n 2|r_i - s_{M_R}(i)| = 2\OPT_R.
\]
Hence, $\mathbb{E}_{R\sim \mathcal{U}([0,1])}[\OPT_{\tilde{R}}]\leq 2\cdot  \mathbb{E}_{R\sim \mathcal{U}([0,1])}[\OPT_{R}]$. Combining this with the fact that fair-bias is $9$-competitive by Lemma \ref{lem:fair_bias}, and using the definition of the distribution $\mathcal{D}$, we obtain
\begin{align}
    \E_{R\sim \mathcal{U}([0,1])}[\text{cost}_{\tilde{R}}(\mathcal{F})] &=  \E_{\tilde{R}\sim \mathcal{D}}[\text{cost}_{\tilde{R}}(\mathcal{F})]  \leq 9\cdot \mathbb{E}_{\tilde{R}\sim \mathcal{D}}[\OPT_{\tilde{R}}] \nonumber\\
    &= 9\cdot \mathbb{E}_{R\sim \mathcal{U}([0,1])}[\OPT_{\tilde{R}}]\leq 18\cdot  \mathbb{E}_{R\sim \mathcal{U}([0,1])}[\OPT_{R}]\label{eq:FB2}.
\end{align}

Hence, taking the expectation over $R\sim \mathcal{U}([0,1])$ on both sides of (\ref{eq:cost_r_r'}) and combining it with (\ref{eq:FB1}) and (\ref{eq:FB2}), we finally obtain
\[
\E_{R\sim \mathcal{U}([0,1])}[\text{cost}_R(\A)]\leq (1+ 18) \cdot \mathbb{E}_{R\sim \mathcal{U}([0,1])}[\OPT_{R}],
\]
which shows that $\A$ is constant competitive in the random requests model. Using Lemma \ref{lem:nearest_neighbours}, we can then transform $\A$ into a constant competitive algorithm $\A'$ that makes neighboring matches, which completes the proof of the lemma.
\end{proof}

%% file: appendix_HG.tex
\section{Hierarchical Greedy is $\Omega(n^{1/4})$ in the Random Requests Model}
\label{app_HG}

In this section, we show that in the random requests model, the Hierarchical Greedy algorithm proposed in \cite{Kanoria21} is $\Omega(n^{1/4})$ competitive on the line. We first introduce the instance on which this lower bound is achieved. To ease the presentation, we define an instance $\mathcal{J}_{2n}$ with $2n$ servers and $2n$ requests and where the servers and requests are in $[0,2]$. Note that by as simple scaling argument, this instance can be cast as an instance of the random requests model with $n$ servers and requests in $[0,1]$.

\paragraph{Description of the instance $\mathcal{J}_{2n}$.}
 
We define the set of servers $S_0$ as follows: there are $n-n^{3/4}$ servers uniformly spread in the interval $[0,1 - n^{-1/4}]$, there are no servers in the interval $(1 - n^{-1/4}, 1)$, there are $n^{3/4}$ servers at position $1+n^{-1/4}$, and the remaining $n$ servers are uniformly spread in the interval $[1, 2]$. More precisely, we let $s_{j} = \frac{j}{n}$ for all $j\in [n-n^{3/4}]$, $s_{j} = 1+n^{-1/4}$ for all $j\in \{n-n^{3/4}+1,\ldots, n\}$ and $s_{j} = \frac{j}{n}$ for all $j\in \{n+1,\ldots, 2n\}$. The sequence of requests $R$ contains $2n$ requests sampled uniformly at random in $[0,2]$.  We note that, interestingly, the  servers are almost uniform since a $1-o(1)$ fraction of the servers are uniformly spread in the interval $[0,2]$. In other words, the Hierarchical Greedy algorithm is not robust to a small perturbation of the servers.

\begin{lemma}
\label{lem:offline_hg}
The expected value of the optimal offline matching for the instance $\mathcal{J}_{2n}$ satisfies: $\E[\OPT] = O(\sqrt{n})$.
\end{lemma}

\begin{proof}
For a given realization $R$ of the requests sequence, we partition the requests into $R_1 = \{r\in R: r \in [0,1-n^{-1/4}]\}$, $R_2 = \{r\in R: r \in (1-n^{-1/4},1)\}$ and $R_3 = \{r\in R: r \in [1,2]\}$. We also let $\overline{R}_1$ be the first $n-n^{3/4}$ elements of $R_1$, or $\overline{R}_1 = R_1$ if $|R_1|< n-n^{3/4}$; we let $\overline{R}_2$ be the first $n^{3/4}$ elements of $R_2$, or $\overline{R}_2 = R_2$ if $|R_2|< n^{3/4}$, and we let $\overline{R}_3$ be the first $n$ elements of $R_3$, or $\overline{R}_3 = R_3$ if $|R_3|< n$.

We now define the following matching $M$, where for all $r\in R$, $s_M(r)$ denotes the server to which $r$ is matched and for all $\Tilde{R}\subseteq R$, $M|_{\Tilde{R}}$ denotes the restriction of $M$ to requests in $\Tilde{R}$:
\begin{itemize}
    \item $M|_{\overline{R}_1}$ is an optimal matching between $\overline{R}_1$ and $S_0\cap [0,1-n^{-1/4}]$.
    \item For all $r\in\overline{R}_2$, $s_M(r)=1+n^{-1/4}$.
    \item $M|_{\overline{R}_3}$ is an optimal matching between $\overline{R}_3$ and $S_0\cap [1,2]$. 
    \item The remaining requests are matched arbitrarily to the remaining free servers.
\end{itemize}
Note that $M$ is well defined since $|\overline{R}_1|\leq n-n^{3/4}= |S_0\cap[0,1-n^{-1/4}]|$, $|\overline{R}_2|\leq n^{3/4}= |S_0\cap\{1+n^{-1/4}\}|$, and $|\overline{R}_3|\leq n= |S_0\cap[1,2]|$.

Now, for all $r\in \overline{R}_2$, since $r\in (1-n^{-1/4},1)$, we have $|r-s_M(r)| = |1+n^{-1/4}-r|\leq 2n^{-1/4}$, hence, letting  $\text{cost}(M)$ denote the cost of the matching $M$, we have
\begin{equation}
    \E[\text{cost}(M|_{\overline{R}_2})] =\E[ \sum_{r\in \overline{R}_2} |s_M(r)-r|] \leq \E[|\overline{R}_2|] \cdot 2n^{-1/4} \leq   n^{3/4}\cdot 2n^{-1/4} = 2n^{1/2}.\label{eq:HGopt1}
\end{equation}
Next, note that the requests in  $\overline{R}_1$ are uniform i.i.d. in $ [0,1-n^{-1/4}]$ and the servers in $S_0\cap [0,1-n^{-1/4}]$ are uniformly spread in $ [0,1-n^{-1/4}]$. Similarly, the requests in  $\overline{R}_3$ are uniform i.i.d. in $ [1,2]$ and the servers in $S_0\cap [1,2]$ are uniformly spread in $ [1,2]$. Hence by Lemma \ref{lem:offline}, we have that
\begin{equation}
    \E[\text{cost}(M|_{\overline{R}_1})] + \E[\text{cost}(M|_{\overline{R}_3})]= O(\sqrt{n}).\label{eq:HGopt2}
\end{equation}

Now, note that $|R_1| = |\{r\in R: r \in [0,1-n^{-1/4}]\}|$ follows a binomial distribution $\mathcal{B}(2n, (1-n^{-1/4})/2)$,  $|R_2| = |\{r\in R: r \in (1-n^{-1/4},1)\}|$ follows a binomial distribution $\mathcal{B}(2n, n^{-1/4}/2)$ and $|R_3| = |\{r\in R: r \in [1,2]\}|$ follows a binomial distribution $\mathcal{B}(2n, 1/2)$. Hence, by Lemma \ref{lem:jensen}, we get
\begin{align*}
    &\E[|R_1\setminus\overline{R}_1|] = \E[\max(0, |R_1| -  n-n^{3/4})] \leq \E[||R_1| - n-n^{3/4}|]\\
    &\qquad\qquad\qquad\qquad\qquad\qquad\qquad\qquad\qquad\qquad\leq \sqrt{(n-n^{3/4})\cdot(1- (1-n^{-1/4})/2)} = O(\sqrt{n}),\\
    &\E[|R_2\setminus\overline{R}_2|] = \E[\max(0, |R_2| -  n^{3/4})] \leq \E[||R_2| -   n^{3/4}|] \leq \sqrt{n^{3/4}\cdot(1-n^{-1/4}/2)} = O(\sqrt{n}),\\
    &\E[|R_3\setminus\overline{R}_3|] = \E[\max(0, |R_3| -  n] \leq \E[||R_3| -   n|]\leq \sqrt{n(1-1/2)} = O(\sqrt{n}).\\
\end{align*}
Since for all $r\in R$, we have $|s_M(r)-r|\leq 1$, we get 
\begin{equation}
    \E[\text{cost}(M|_{(R_1\setminus\overline{R}_1)\cup (R_2\setminus\overline{R}_2))\cup (R_3\setminus\overline{R}_3)})]\leq \E[|R_1\setminus\overline{R}_1|] + \E[|R_2\setminus\overline{R}_2|] + \E[|R_3\setminus\overline{R}_3|] = O(\sqrt{n}).\label{eq:HGopt3}
\end{equation}
Combining (\ref{eq:HGopt1}), (\ref{eq:HGopt2}) and (\ref{eq:HGopt3}), we get
\begin{align*}
    \E[\OPT]\leq \E[\text{cost}(M)] &= \E[\text{cost}(M|_{(R_1\setminus\overline{R}_1)\cup (R_2\setminus\overline{R}_2)\cup (R_3\setminus\overline{R}_3)}] + \E[\text{cost}(M|_{\overline{R}_1})] \\
    &\qquad\qquad\qquad\qquad\qquad+ \E[\text{cost}(M|_{\overline{R}_2})] + \E[\text{cost}(M|_{\overline{R}_3})]= O(\sqrt{n}).
\end{align*}
\end{proof}

\begin{lemma}
\label{lem:lower_hg}
The expected cost of the matching returned by algorithm $\A^H$ on instance $\mathcal{J}_{2n}$ satisfies: $\E[cost(\A^H, \mathcal{J}_{2n})]  = \Omega(n^{3/4})$.
\end{lemma}

\begin{proof}
For a given realization $R$ of the requests sequence, we let $R_1 =\{r\in R: r\in [0,1]\}$. We also let $\overline{R}_1$ be the first $n-n^{3/4}$ elements of $R_1$, or $\overline{R}_1 = R_1$ if $|R_1|< n-n^{3/4}$.
Now, note that $|R_1| = |\{r\in R: r \in [0,1]\}|$ follows a binomial distribution $\mathcal{B}(2n, 1/2)$ with mean $n$. Hence,  
\begin{equation}
\label{eq:r1r1}
    \E[|R_1\setminus\overline{R}_1|] \geq n^{3/4}\P(|R_1\setminus\overline{R}_1|\geq n^{3/4}) = n^{3/4}\P(|R_1|\geq n) = \frac{1}{2}n^{3/4}.
\end{equation}

Next, note that the Hierarchical Greedy algorithm matches a request $r\in [0,1]$ to a server in $(1,2]$ only if $[0,1]$ has no more available servers. Hence, since $\overline{R}_1$ contains at most $n-n^{3/4}$ requests and there are initially $n-n^{3/4}$ servers in $[0,1]$, all requests in $\overline{R}_1$ will be matched to servers in $[0,1]$. Now, if $|R_1\setminus\overline{R}_1|>0$, then $|\overline{R}_1|= n-n^{3/4} = S_0\cap[0,1]$; hence, when any request $r\in |R_1\setminus\overline{R}_1|$ arrives, all the servers in $[0,1]$ have already been matched to a request in $\overline{R}_1$. Therefore, $r$ is matched by $\A^H$ to a server $s_{\A^H}(r)$ in $(1,2]$. Noting that the requests in $R_1\setminus\overline{R}_1$ are uniform in $[0,1]$, we thus have, for any $t$ such that $r_t\in |R_1\setminus\overline{R}_1|$,
\begin{align}
\label{eq:r1lb}
   \E[cost_t(\A^H, r_t)|r_t\in R_1\setminus\overline{R}_1]
   = \E[s_{\A^H}(r_t)-r_t|r_t\in R_1\setminus\overline{R}_1] \geq 1- \E[r_t|r_t\in R_1\setminus\overline{R}_1] = \frac{1}{2}.  
\end{align}

Thus, 
\begin{align*}
    \E[cost(\A^H, \mathcal{J}_{2n})]&= \sum_{t\in [n]} \E[cost_t(\A^H, r_t)]\\
    &\geq \sum_{t\in [n]} \E[cost_t(\A^H, r_t)|r_t\in R_1\setminus\overline{R}_1]\P(r_t\in R_1\setminus\overline{R}_1)\\
    &\geq \sum_{t\in [n]} \frac{1}{2}\P(r_t\in R_1\setminus\overline{R}_1)\tag{by (\ref{eq:r1lb})}\\
    &= \frac{1}{2}\E[\sum_{t\in [n]} \mathbbm{1}_{r_t\in R_1\setminus\overline{R}_1}]\\
    &=\frac{1}{2} \E[|R_1\setminus\overline{R}_1|]\\& = \frac{1}{4}n^{3/4}\tag{by (\ref{eq:r1r1})}.
\end{align*}
\end{proof}

By combining Lemmas \ref{lem:lower_hg} and \ref{lem:offline_hg}, we get the following result.

\begin{lemma}
For online matching on the line in the random requests model, the Hierarchical Greedy algorithm $\A^H$  achieves an $\Omega(n^{1/4})$-competitive ratio.
\end{lemma}

%% file: appendix_overview_LB.tex
\section{Greedy is $\Omega(\log{n})$-competitive}
\label{sec:full_proof}

\input{full_proof_LB}

%% file: full_proof_LB.tex
\label{app_lower}

In this section, we give a more detailed proof of our lower bound result. All omitted proofs can be found in Appendix \ref{appendix:full_LB}.

\subsection{Preliminaries}
\label{subsec:preliminaries}

\textbf{Description of the instance.} There are $n^{4/5}+\excess$ servers located at point $0$, there are no  servers in the interval $(0, n^{-1/5}]$ and the remaining $n-(n^{4/5}+\excess)$ servers are uniformly spread in the interval $(n^{-1/5}, 1]$. More precisely, for all $j\in [n^{4/5} + \excess]$, we set $s_j= 0$. Then, we let $\tilde{n}:= n - \excess/(1-n^{-1/5})$, and for all $j\in \{1,\ldots, n-n^{4/5}-\excess)\}$, we set  $s_{(n^{4/5} + \excess)+j} = n^{-1/5} + \frac{j}{\tilde{n}}$ (see Figure \ref{fig:lb_instance} for an illustration of the instance).  We note that, interestingly, the  servers are almost uniform since a $1-o(1)$ fraction of the servers are uniformly spread in an interval $(o(1), 1]$.

We now give bounds on the number of servers contained in each subinterval of $[n^{-1/5},1]$.

\begin{restatable}{rLem}{factnumber}
\label{fact:instance}
Let $\tilde{n}:= n - \excess/(1-n^{-1/5})$. For any $I\subseteq[n^{-1/5},1]$, we have $|S_0\cap I|\in [\tilde{n}|I|-1, \tilde{n}|I|+3]$.
\end{restatable}

\paragraph{Basic definitions and notations.}

We first introduce some notation and terminology.

\begin{definition}We consider a partition $I_0, I_1, \ldots$ of $(0,1]$ into intervals of geometrically increasing size, where $I_i = (y_{i-1}, y_i]$ and $y_i = (3/2)^{i}n^{-1/5}$ (with the convention $y_{-1}=0$).
\end{definition}

We now define an algorithm $\A$ to which we will compare the greedy algorithm.

\begin{definition}
 We let $\mathcal{A}$ be the algorithm that, for all $t\in [n]$, matches $r_t$ to a free server at location $0$ if $r_t\in [0,n^{-1/5}]$ and $S_{\mathcal{A},t-1}\cap\{0\}\neq \emptyset$, and, otherwise, matches $r_t$ greedily. For all $m\geq 0$, we recall that $\H^m$ denotes the hybrid algorithm that matches the first $m$ requests according to $\A$, then, matches greedily the remaining requests to the remaining free servers.
\end{definition}

\paragraph{A useful tool: regularity of the requests sequence.}

Informally, we define a sequence of requests $R$ regular if in every time interval, its realized density is not much different from its expected density. We now give some intuition about why we define such a notion. Throughout the proof, many random events  can be shown to occur with high probability by successive applications of simple Chernoff bounds. Once the sequence of requests is assumed to be regular,  these events become deterministic events, which greatly simplifies the analysis.

More formally, we start by discretizing the interval $[0,1]$ as $\mathcal{D} = \{\tfrac{i}{n}: i\in \{0,\ldots, n\}\}$. For any interval $I = [i_L,i_R]\subseteq [0,1]$, we also consider $d^+(I)$, the smallest interval with end points in $\mathcal{D}$ that contains $I$, and $d^-(I)$, the largest interval with end points in $\mathcal{D}$ contained in $I$.
 \begin{enumerate}
     \item $d^+(I):= [d_L^+,d_R^+]$, with $d_L^+ = \max\{x\in \mathcal{D}| x\leq i_L\}$ and $d_R^+ = \min\{x\in \mathcal{D}| x\geq i_R\}$ 
     \item $d^-(I):= [d_L^-,d_R^-]$, with $d_L^- = \min\{x\in \mathcal{D}| x\geq i_L\}$ and $d_R^- = \max\{x\in \mathcal{D}| x\leq i_R\}$.
 \end{enumerate}
 
 \begin{definition}
 
\label{def:reg}
We say that a realization $R$ of the sequence of requests is regular if for all $d,d'\in \mathcal{D}$ such that $d<d'$, and for all $t,t' \in [n]$ such that $t<t'$,
\begin{enumerate}
    \item $|\{j\in \{t,\ldots, t'\}|\;r_j \in [d,d']\}|\geq (d'-d)(t'-t) - \log(n)^2\sqrt{(d'-d)(t'-t)}$,
    \item and if $(d'-d)(t'-t) = \Omega(1)$, then \[|\{j\in \{t,\ldots, t'\}|\;r_j \in [d,d']\}|\leq (d'-d)(t'-t) + \log(n)^2\sqrt{(d'-d)(t'-t)}.\]
\end{enumerate}
\end{definition}

We now show that $R$ is regular with high probability.

\lemreg*
\begin{proof} Note that for all $d,d'\in D$ such that $d<d'$ and $t,t' \in [n]$ such that $t<t'$, $|\{j\in \{t,\ldots, t'\}|\;r_j \in [d,d']\}|$ follows a binomial distribution $\mathcal{B}(t'-t, d'-d)$. Hence the lemma results from a direct application of Chernoff Bounds (Lemma \ref{cor:CB}) and a union bound over all $d,d'\in D$ and $t,t'\in [n]$.
\end{proof}

 We now show a property that is implied by the regularity of a sequence R of requests. We define $m_t$ = $|S_{t}\cap(0,1]|$, and we denote by $0<s_{t,1}<\ldots< s_{t,m_t}\leq 1$ the locations of the $m_t$ free servers with positive location in $S_t$. For some small $\epsilon>0$, we define $c_3 =\frac{4}{5}+\epsilon$. The following lemma upper bounds the distance between two consecutive free servers with positive location in $S_t$ at time $t\in [n-o(n)]$ for algorithm $\H^m$ assuming that $R$ is regular.

\begin{restatable}{rLem}{lemlnbeta}
\label{prop:l_nbeta}
Assume that the  sequence of requests is regular. Then, for $n$ large enough and for all $t \in [n-n^{c_3}]$ and $j \in [m_t-1]$, we have  $s_{t,j+1} - s_{t,j}
\leq 2\log(n)^4n^{1-2c_3}$.
\end{restatable}

\subsection{Upper bound on the cost of the optimal offline matching}

The goal of this section is to prove Lemma~\ref{lem:cost_opt}, which gives an upper bound to the cost of the offline optimum. We first introduce a useful lemma.
\begin{lemma}
\label{lem:unif_iid}
   Let $m\geq 0$ and $R = \{r_1,\ldots, r_{|R|}\}$ be a set of at most $m$ requests uniformly drawn from the interval $(0,1]$ and $Z = \{z_1,\ldots, z_m\}$ be a set of $m$ servers such that for all $i \in \{1,\ldots, m\}$, $z_i = \frac{i}{m}$. Then, the optimal matching $M^*$ between $Z$ and $R$ satisfies $\E(\text{cost}(M^*)) = O(\sqrt{m})$.
\end{lemma}

\begin{proof} We assume without loss of generality that $R$ contains exactly $m$ requests and we let $r_{(1)}<\ldots <r_{(m)}$ denote the ordered statistics of $R$. In this case, we claim that an optimal matching $M^*$ between $R$ and $Z$ is to match each  $r_{(i)}$ to $z_i$ for all $i\in \{1,\ldots, m\}$ (see the proof of Theorem 2.5 in \cite{Akbarpour21} for a proof of this fact).

Now, it is a known fact that for all $i\in \{1,\ldots, m\}$, $r_{(i)}$ follows a Beta distribution $B(i, m+1-i)$ (see \cite{Beta_distrib}). In particular, we have that $\mathbb{E}[r_{(i)}] = \tfrac{i}{m+1}$ and $\text{std}(r_{(i)}) = \sqrt{\frac{i(m-i+1)}{(m+1)^2(m+2)}}\leq \frac{1}{\sqrt{m}}$. We thus obtain
\begin{align*}
\E(\text{cost}(M^*)) &= \sum_{i=1}^m \mathbb{E}(|r_{(i)} - z_i|)\\
&\leq \sum_{i=1}^m \mathbb{E}(|r_{(i)} - \mathbb{E}(r_{(i)})|) + \mathbb{E}(|\mathbb{E}(r_{(i)}) - z_i|)\\
    &\leq \sum_{i=1}^m \text{std}(r_{(i)}) +\sum_{i=1}^m \mathbb{E}(|\mathbb{E}(r_{(i)}) - z_i|)\tag{Lemma \ref{lem:jensen}}\\
    &\leq \sum_{i=1}^m \frac{1}{\sqrt{m}} + \sum_{i=1}^m |\tfrac{i}{m+1}- \tfrac{i}{m}|\\
    &= O(\sqrt{m}).
\end{align*}
\end{proof}

We now give an upper bound on the cost of the optimal offline matching for our lower bound instance. 

\optoffline*

\begin{proof} For a given realization $R$ of the requests sequence, we partition the requests into $R_1 = \{r\in R: r \in [0,n^{-1/5}]\}$ and $R_2 = \{r\in R: r \in (n^{-1/5},1]\}$. We also let $\overline{R}_1$ be the $n^{4/5}$  requests of $R_1$ that arrived first, or $\overline{R}_1 = R_1$ if $|R_1|< n^{4/5}$, and let $\overline{R}_2$ be the  $n-(n^{4/5}+\excess)$  requests of $R_2$ that arrived first, or $\overline{R}_2 = R_2$ if $|R_2|< n-(n^{4/5}+\excess)$.

We now define the following matching $M$, where for all $r\in R$, $s_M(r)$ denotes the server to which $r$ is matched, and for any subset $\tilde{R}$ of the requests, $M|_{\tilde{R}}$ denote the restriction of  $M$ to $\tilde{R}$:
\begin{itemize}
    \item For all $r\in\overline{R}_1$, $s_M(r)=0$.
    \item $M|_{\overline{R}_2}$ is an optimal matching between $\overline{R}_2$ and $S_0\cap (n^{-1/5},1]$. 
    \item The remaining requests are matched arbitrarily to the remaining free servers.
\end{itemize}
Note that $M$ is well defined since $|\overline{R}_1|\leq n^{4/5}\leq |S_0\cap\{0\}|$ and $|\overline{R}_2|\leq n-(n^{4/5}+\excess)\leq |S_0\cap(n^{-1/5},1]|$. 

Now, for all $r\in R_1$, since $r\in [0,n^{-1/5}]$, we have $|r-s_M(r)| = |r-0|\leq n^{-1/5}$, hence 
\begin{equation}
    \label{eq:opt1}
    \E[\text{cost}(M|_{\overline{R}_1})] =\E[ \sum_{r\in \overline{R}_1} |s_M(r)-r|] \leq \E[|\overline{R}_1|] \cdot n^{-1/5} \leq   n^{4/5}\cdot n^{-1/5} = n^{3/5}.
\end{equation}

Next, note that the requests in  $\overline{R}_2$ are uniform i.i.d. in $ (n^{-1/5},1]$ and the servers in $S_0\cap (n^{-1/5},1]$ are uniformly spread in $ (y_0,1]$, hence by using Lemma \ref{lem:unif_iid} and a simple scaling argument, we get
\begin{equation}
    \label{eq:opt2}
    \E[\text{cost}(M|_{\overline{R}_2})] = O(\sqrt{|\overline{R}_2|}) = O(\sqrt{n}).
\end{equation}

Now, note that $|R_1| = |\{r\in R: r \in [0,n^{-1/5}]\}|$ follows a binomial distribution $\mathcal{B}(n, n^{-1/5})$ with mean $n^{4/5}$ and standard deviation $\sqrt{n^{4/5}(1-n^{-1/5})}$, thus by Lemma \ref{lem:jensen}, we have $
\E[\max(0, |R_1| - n^{4/5})] \leq \E[||R_1| -   n^{4/5}|]\leq \sqrt{n^{4/5}(1-n^{-1/5})} \leq \sqrt{n}$. Since by definition, $R_1\setminus\overline{R}_1$ contains $\max(0, |R_1| - n^{4/5})$ elements, we thus have
\[
\E[|R_1\setminus\overline{R}_1|] = \E[\max(0, |R_1| -  n^{4/5})] \leq \sqrt{n}.
\]

We also have that $|R_2| = |\{r\in R: r \in (n^{-1/5},1]\}|$ follows a binomial distribution $\mathcal{B}(n, 1-n^{-1/5})$ with mean $n-n^{4/5}$ and standard deviation $\sqrt{(n-n^{4/5})n^{-1/5})}$. Hence, by Lemma \ref{lem:jensen}, we have $
\E[\max(0, |R_2| - (n - n^{4/5}))] \leq \E[|R_2| - (n - n^{4/5})|]\leq \sqrt{(n-n^{4/5})n^{-1/5}} \leq \sqrt{n}.$
Since by definition, $R_2\setminus\overline{R}_2$ contains $\max(0, |R_2| - (n - n^{4/5}-\excess))$ elements, we thus have
\begin{align*}
    \E[|R_2\setminus\overline{R}_2|] &= \E[\max(0, |R_2| - (n - n^{4/5}-\excess))] \\
    &\leq \E[\max(0, |R_2| - (n - n^{4/5}))] + \excess \\
    &=  \tilde{O}(\sqrt{n}).
\end{align*}

Since for all $r\in R$, we have $|s_M(r)-r|\leq 1$, we get
\begin{equation}
    \label{eq:opt3}
    \E[\text{cost}(M|_{(R_1\setminus\overline{R}_1)\cup (R_2\setminus\overline{R}_2)})]\leq \E[|R_1\setminus\overline{R}_1|] + \E[|R_2\setminus\overline{R}_2|] = \tilde{O}(\sqrt{n}).
\end{equation}

Combining (\ref{eq:opt1}), (\ref{eq:opt2}) and (\ref{eq:opt3}), we finally get
\[
\E[\OPT]\leq \E[\text{cost}(M)] = \E[\text{cost}(M|_{(R_1\setminus\overline{R}_1)\cup (R_2\setminus\overline{R}_2)})] + \E[\text{cost}(M|_{\overline{R}_1})] + \E[\text{cost}(M|_{\overline{R}_2})] = O(n^{3/5}).\]
\end{proof}

\subsection{Analysis of $(S_1,\ldots, S_n)$.}
\label{subsec:analysis_of_S}

We first introduce a few constants that will be used thoughout the proof. We recall that $c_3 =\frac{4}{5}+\epsilon$ for some small $\epsilon>0$. We also define the following constants: $c_1 = \frac{2}{9}(1-\epsilon), c_2 = \frac{2}{3}\left((1+\epsilon) + \frac{1}{9}(1-\epsilon)\right)$. Note that in particular, we have $1>c_2>1/2>c_1>0$. In addition, we define $d_1 := (1-c_3)/\log(1/(1-c_2))$.

In this section, we consider a fixed value $m \leq c_1n$ and we give some global property of the sequence $(S_1,\ldots, S_n)$ of sets of free servers for $\H^m$. More precisely, we first define for all interval $I$ the time $t_I :=\min \{t\geq 0| S_t\cap I = \emptyset\}$ at which the last free server of $I$ is matched to some request (we say that $I$ is \textbf{depleted} at time $t_I$). The objective is to show that during the execution of $\H^m$, and for all $i < j$, the interval $I_i$ is depleted at an earlier time step than $I_j$, and that all intervals $\{I_i\}_{i\in [\io]}$ are depleted between times $m$  and $n-n^{c_3}$ (which is formally stated in the next lemma).

\begin{restatable}{rLem}{lemlb}
\label{lem:lb1}
Let $m\leq c_1n$ and consider algorithm $\H^m$. Then, assuming that the sequence of requests $R$ is regular, we have that   $c_1n$ $< t_{1}<\ldots< t_{\io}\leq n- n^{c_3}.$ In addition, we have   $c_1n\leq t_{\{0\}}$.
\end{restatable}

Before presenting the proof of Lemma \ref{lem:lb1}, we introduce a few technical properties. We first show a simple but useful lemma.
\begin{lemma}
\label{claim:dlogn}
For all $i\in \{0,\ldots, \io\}$, we have $(1-c_2)^{i}n \geq n^{c_3}$.
\end{lemma}

\begin{proof} Let $i\in \{0,\ldots, \io\}$. Then,
\begin{equation*}
    (1-c_2)^{i}n = e^{-i\log(1/(1-c_2)) + \log(n)}\geq  e^{-\frac{(1-c_3)\log(n)}{log(1/(1-c_2))}\log(1/(1-c_2)) + \log(n)} =  e^{-(1-c_3)\log(n) + log(n)} = n^{c_3},
\end{equation*}
where the inequality is by definition of $d_1$ and since $i\leq \io$.
\end{proof}

 We now show a number of properties that are satisfied for all $i\in [\io]$ under the assumption that the sequence  of requests $R$ is regular.

 We first show that if the depletion time $t_{i-1}$ of interval $I_{i-1}$ is small enough, then $I_{i-1}$ is depleted before $I_i$.

\begin{restatable}{rLem}{IiIi}
\label{lem:Ii_Ii} 
Let $m\leq c_1n$ and $i\in [\io]$. Assume that $R$ is regular and that $t_{i-1}\leq n - (1-c_2)^{i-1}n$. Then, $t_{i-1}< t_{i}$.

\end{restatable}

Next, we show that if the intervals $I_1,\ldots, I_{i-1}$ are depleted in increasing order of $i$ and that $t_{i-1}$ is small enough, then $t_i$ is also small enough. To this end, we first introduce a couple of lemmas. The first one upper bounds the number of requests that arrived in $I_i$ and were matched outside of $I_i$ until time $\min (t_i,t_{i-1} + c_2(n-t_{i-1}))$.

\begin{restatable}{rLem}{propfi}\label{prop_fi}
 Let $m\leq c_1n$ and $i\in [\io]$. Assume that $R$ is regular, that  $t_0<\ldots<t_{i-1}\leq n - (1-c_2)^{i-1}n$ and that $t_{i-1}<t_{i}$. Let $\overline{t_i} := \min (t_i,t_{i-1} + c_2(n-t_{i-1}))$. Then,
 \[|\{j\in [\overline{t_i}]:r_j\in I_i, s_{\H^m}(r_j)\notin I_i\}| = \tilde{O}(\sqrt{n}).\]

\end{restatable}

The next lemma lower bounds the number of requests that arrived in the interval $[\tfrac{3}{4} y_{i-1}, y_{i-1}]$  and were matched inside $I_i$ from time $t_{i-1} + 1 + c_1(n - t_{i-1})$ to time $\min (t_i,t_{i-1} + c_2(n-t_{i-1}))$.

\begin{restatable}{rLem}{propgi}\label{prop_gi}
Let $m\leq c_1n$ and $i\in [\io]$. Assume that $R$ is regular, that $t_0<\ldots<t_{i-1}$ and that $t_{i-1}<t_{i}$. Let $\overline{t_i} := \min (t_i,t_{i-1} + c_2(n-t_{i-1}))$. Then,
\begin{align*}
     &|\{j\in \{t_{i-1} + 1 + c_1(n - t_{i-1}),\ldots ,\overline{t_i}\} :r_j\in [\tfrac{3}{4} y_{i-1}, y_{i-1}], s_{\H^m}(r_j)\in I_i\}|\\ 
     &\geq \frac{1}{2}(\overline{t_i} - t_{i-1} - c_1 (n - t_{i-1}))|I_i|- \tilde{O}(\sqrt{n}).
\end{align*}
 
\end{restatable}

Using the two above lemmas, we show that if the intervals $I_1,\ldots, I_{i-1}$ are depleted in increasing order of $i$ and $t_{i-1}$ is small enough, then $t_i\leq n- (1-c_2)^in$. 

\begin{figure}
    \centering
    \includegraphics[scale=0.25]{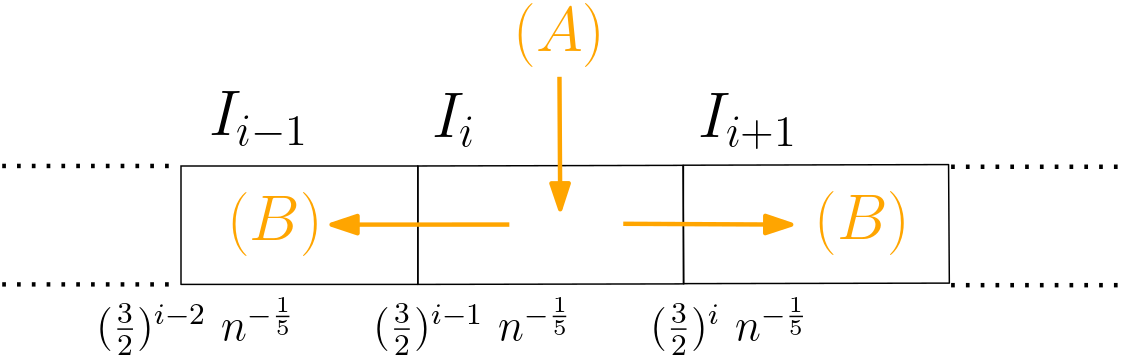}\\
    \includegraphics[scale=0.3]{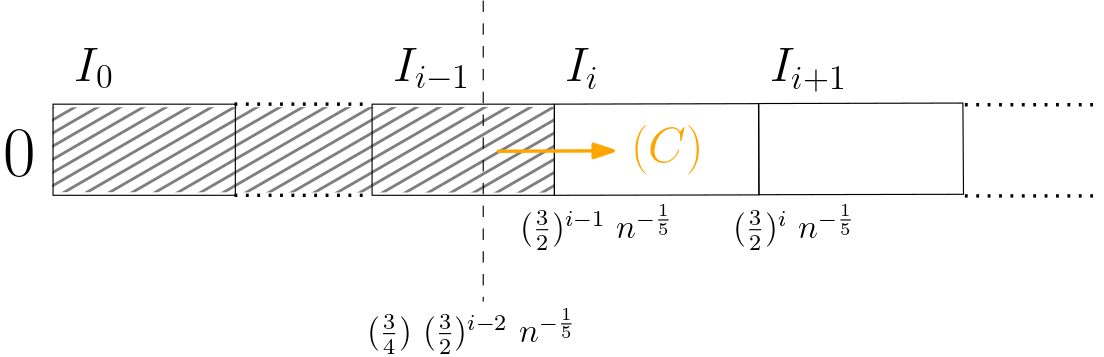}
    \caption{Requests in and out of $I_i$ up to time $\overline{t}_i= \min(t_i, t_{i-1} + c_2(n-t_{i-1}))$, with (A) the total number of requests that arrived in $I_i$ from time 0 to $\overline{t}_i$, (B) the total number of requests that arrived in $I_i$ and were matched outside $I_i$ from time $0$ to $\overline{t}_i$, and $(C)$ the total number of requests that arrived in $[\tfrac{3}{4}y_{i-1}, y_{i-1}]$ and were matched inside $I_i$ from time $t_{i-1} + 1 +c_1(n - t_{i-1})$ to time $\overline{t}_i$ (note that there are no free servers in the dashed area for times $t\geq t_{i-1}$).}
    \label{fig:flow}
\end{figure}

\begin{lemma}\label{lem:I_i_depleted_t} 
   Let $m\leq c_1n$ and $i\in [\io]$. Assume that $R$ is regular, that $t_0<\ldots< t_{i-1}\leq n- (1-c_2)^{i-1}n$, and  that $t_{i-1}<t_{i}$. Then, $t_i\leq n- (1-c_2)^in$.
 \end{lemma}
 
\begin{proof}
 Fix $i\in [\io]$. We start by lower bounding the number of requests that were matched to servers inside $I_i$ until time $\overline{t_i} := \min(t_i, t_{i-1} + c_2(n-t_{i-1}))$ included. First, we have (see Figure \ref{fig:flow}):
\begin{align*}
&\quad|\{j\in [\overline{t_i}]\;| s_{\H^m}(r_j) \in I_i\}| \\
&=  |\{j\in [\overline{t_i}]\;| r_j\in I_i, s_{\H^m}(r_j) \in I_i\}| +  |\{j\in [\overline{t_i}]\;|  r_j\notin I_i, s_{\H^m}(r_j) \in I_i\}| \\
    &\geq \Big[|\{j\in [\overline{t_i}]:r_j\in I_i\}|&\text{(A)}\\
    &- |\{j\in [\overline{t_i}]:r_j\in I_i, s_{\H^m}(r_j)\notin I_i\}|\Big]&\text{(B)}\\
    &+|\{j\in \{t_{i-1} + 1 +c_1(n - t_{i-1}),\ldots ,\overline{t_i}\} :r_j\in [\tfrac{3}{4}y_{i-1}, y_{i-1}], s_{\H^m}(r_j)\in I_i\}|&\text{(C)}
\end{align*}
where the lower bound in (C)  holds since $I_i=(y_{i-1},y_i]$; hence $[\tfrac{3}{4}y_{i-1}, y_{i-1}]\subseteq [0,1]\setminus I_i$.

We now bound each of these three terms separately. Since we assumed that the sequence of requests is regular, by applying the first regularity condition with $t=0$, $t'= \overline{t_i}$, $[d,d'] = d^-(I_i)$),  we have that
\[
|\{j\in [\overline{t_i}]:r_j\in I_i\}| \geq |\{j\in [\overline{t_i}]:r_j\in d^-(I_i)\}| \geq d^-(I_i)\overline{t_i} - \log(n)^2\sqrt{d^-(I_i)\overline{t_i}}=  |I_i|\overline{t_i} -\tilde{O}(\sqrt{n}).
\]

By Lemma \ref{prop_fi}, we have that
\begin{align*}
    |\{j\in [\overline{t_i}]:r_j\in I_i, s_{\H^m}(r_j)\notin I_i\}| = \tilde{O}(\sqrt{n}),
\end{align*}

and by Lemma \ref{prop_gi}, we have that
\begin{align*}
    &|\{j\in \{t_{i-1} + 1 +c_1(n - t_{i-1}),\ldots ,\overline{t_i}\}:r_j\in [\tfrac{3}{4}y_{i-1}, y_{i-1}], s_{\H^m}(r_j)\in I_i\}|\\
    &\geq \frac{1}{2}(\overline{t_i} - t_{i-1} - c_1 (n - t_{i-1})|I_i|- \tilde{O}(\sqrt{n}).
\end{align*}

Combining the four previous inequalities gives
\begin{align}
\label{eq:A}
  |\{j\in [\overline{t_i}] : s_{\H^m}(r_j) \in I_i\}|\geq |I_i| \Big[\overline{t_i} + \frac{1}{2}(\overline{t_i} - t_{i-1} - c_1 (n - t_{i-1}))\Big] - \tilde{O}(\sqrt{n})\nonumber.
\end{align}

Now, $|\{j\in [\overline{t_i}] : s_{\H^m}(r_j) \in I_i\}|$ is trivially upper bounded by the initial number of servers available in $I_i$, which, by  Lemma~\ref{fact:instance}, is at most $|I_i|\tilde{n} + 1< |I_i|(n + 1)$. By combining this upper bound with the above lower bound and by simplifying the $|I_i|$ on both sides, we obtain
\begin{equation}
\label{eq:lb_ub_lem3}
    n + 1>  \overline{t_i} + \frac{1}{2}(\overline{t_i} - t_{i-1} - c_1 (n - t_{i-1}))- \tilde{O}(\sqrt{n}/|I_i|) = \overline{t_i} + \frac{1}{2}(\overline{t_i} - t_{i-1} - c_1 (n - t_{i-1})- \tilde{O}(n^{7/10}),
\end{equation}
where the equality is since $|I_i| = \Omega(n^{-1/5})$ for all $i\geq 0$.

Next, we show that the previous inequality implies that $\overline{t_i} = t_i$. Assume by contradiction that $\overline{t_i} =  t_{i-1} + c_2(n-t_{i-1})$. We get
\begin{align*}
    &t_{i-1} + c_2(n-t_{i-1}) + \frac{1}{2}(t_{i-1} + c_2(n-t_{i-1}) - t_{i-1} - c_1 (n - t_{i-1})) - \tilde{O}(n^{7/10})\\
    & = t_{i-1} + (n-t_{i-1})(c_2(1+1/2) - c_1/2) - \tilde{O}(n^{7/10})\\
    & = t_{i-1} + (n-t_{i-1})(1+\epsilon) - \tilde{O}(n^{7/10})\\
    &= n + \epsilon(n- t_{i-1}) - \tilde{O}(n^{7/10})\\
    &\geq n + \epsilon n^{c_3} - \tilde{O}(n^{7/10})\\
    &> n +1,
\end{align*}

where the second equality  holds since $c_2(1+1/2) - c_1/2 = \frac{3}{2}\cdot\frac{2}{3}\left((1+\epsilon) + \frac{1}{9}(1-\epsilon)\right) -\frac{1}{2}\cdot\frac{2}{9}(1-\epsilon)=  (1+\epsilon)$. The first inequality  holds since $t_{i-1}\leq n- (1-c_2)^{i-1}n\leq n - n^{c_3}$ (by using the assumption of the lemma and from Lemma \ref{prop:l_nbeta}), and the last inequality  holds since we set $c_3>3/4$ and assumed $n$ large enough.

Hence, by (\ref{eq:lb_ub_lem3}), we cannot have $\overline{t_i} =  t_{i-1} + c_2(n-t_{i-1})$, thus $\overline{t_i} =  \min(t_{i-1} + c_2(n-t_{i-1}),t_i) = t_i$. 

Using the assumption that $t_{i-1}\leq n-(1-c_2)^{i-1}n$, we conclude that
\[
t_i\leq t_{i-1} + c_2(n-t_{i-1}) = c_2 n + (1-c_2)t_{i-1} \leq c_2 n + (1-c_2)(n-(1-c_2)^{i-1}n) = n-(1-c_2)^in.
\]
\end{proof}
Finally, we show in the two following lemmas that $I_1$ is not yet depleted at time $c_1n$, and that if it is the case, we also have that $\{0\}$ is not yet depleted at time  $c_1n$.
\begin{restatable}{rLem}{propmt}\label{prop:mt_1}
Let $m\leq c_1n$ and $i\in [\io]$. Assume that $R$ is regular and that $t_1<t_2$. Then, $c_1n<t_1$.

\end{restatable}

\begin{restatable}{rLem}{lemmto}
\label{lem:mt0}
Let $m\leq c_1n$ and $i\in [\io]$. Assume that $R$ is regular and that  $c_1n<t_1$. Then, $c_1n<t_{\{0\}}$.

\end{restatable}

We are now ready to present the proof of Lemma \ref{lem:lb1}, that we restate below for convenience. 
\lemlb*

\begin{proof}
Fix $m\in [c_1n]$ and assume that the  sequence of requests is regular. We first show by induction on $i$ that, for $n$ sufficiently large, we have 
    $t_0<\ldots< t_i\leq n - (1-c_2)^in$ for all $i\in \{0,\ldots, \io\}$.
    
    The base case is immediate since by construction of the instance, $I_0\cap S_0 = (0,n^{-1/5}]\cap S_0 =\emptyset$, which implies that $t_0=0 = n-(1-c_2)^0n$. 
    
    Now, for $n$ large enough, let $i\in [\io]$ and assume that $t_0<\ldots< t_{i-1}\leq n - (1-c_2)^{i-1}n$. Then, in particular, we have that $t_{i-1}\leq n - (1-c_2)^{i-1}n$, hence $t_{i-1}< t_i$ by Lemma \ref{lem:Ii_Ii}. By combining this with the assumption that $t_0<\ldots< t_{i-1}\leq n - (1-c_2)^{i-1}n$, we obtain that $t_i\leq n - (1-c_2)^in$ by Lemma \ref{lem:I_i_depleted_t}. Hence, we get $t_0<\ldots< t_i\leq n - (1-c_2)^in$, which concludes the inductive case.
    
    By applying the previous inequalities with $i = \io$, and by Lemma \ref{claim:dlogn}, we thus have $t_{1}<\ldots< t_{\io}\leq n - (1-c_2)^{\io}n\leq n- n^{c_3}$. In addition, since $t_1<t_2$ and since we assumed $m\leq c_1n$, we have that $m\leq c_1n<t_1$ by Lemma \ref{prop:mt_1}, which also implies that $m\leq c_1n<t_{\{0\}}$ by Lemma \ref{lem:mt0}. 
    We conclude that
    $m< t_{1}<\ldots< t_{\io}\leq n-n^{c_3}$ and that $m<t_{\{0\}}$. \end{proof}

\subsection{Lower bound on $\mathbb{E}[cost(\H^{m-1}) - cost(\H^m)]$.}
\label{subsec:lower_bounc_cost}

\label{sec:proof_lemma_mcn2}

The objective of this section is to use the characterization of the remaining servers $(S_1, \ldots, S_n)$ from Lemma~\ref{lem:lb1} in Section~\ref{subsec:analysis_of_S} to prove the following lemma, in which we lower bound the total difference of cost between algorithms $\H^{m-1}$ and $\H^m$ conditioned on the location of request $r_m$. The proof is given at the end of the section.

\lemmnc*

\paragraph{Structural properties.} In order to prove Lemma \ref{lem:m_cn2}, we first introduce a few structural properties about the sets $(S_0,\ldots, S_t)$ and $(S_0', \ldots, S_t')$ of free servers for $\H^{m-1}$ and $\H^m$, respectively. We first show that at every time step $t$, there are at most two servers in the  symmetric difference between $S_t$ and  $S_t'$, and that the potential extra free server in $S_t'$ is always located at $0$ whereas the potential extra free server in $S_t$ is the leftmost free server in $S_t$ that is not at location $0$ (see Figure~\ref{fig:config_LB}).

\lemconfigSS*

Armed with the previous lemma, we define the gap $\delta_t := \min\{s \in S_t : s > 0\}$ between the unique available server in $S'_t \setminus S_t = \{0\} $ and the unique available server in $S_t \setminus S_t' = \{\min\{s \in S_t : s > 0\}\}$. In the following, we let $s_{t,1} = \min\{s>0: s\in S_t\}$ and $s_{t,2}= \min\{s>s_{t,1}: s\in S_t\}$ denote the first two servers with positive location for $\H^m$ just after matching $r_t$.

\begin{definition}
\label{def:delta}
For all $t\in [n]$, we let  $\delta_t:= \begin{cases} 0 &\text{ if } S_t = S_t'. \\
    s_{t,1}&\text{ otherwise. }
    \end{cases}$
\end{definition}

We now present a few properties satisfied by $\{(\delta_t,S_t)\}_{t\geq m}$. We start by a partial characterization of the value of $(\delta_t, S_t)$ and of the difference of cost $\Delta \text{cost}_{t+1}:= \text{cost}_{t+1}(\H^{m-1}) - \text{cost}_{t+1}(\H^m)$  between the costs incurred by $\H^{m-1}$ and $\H^{m}$ at time step $t$ as a function of $\delta_t$ and  $S_t$.

\lemMCdef*


\begin{center}
\begin{table}[]
    \centering
    \begin{tabular}{c||c| c| c| c| c} 

$r_{t+1}\in \ldots$ & $[0,\tfrac{\delta_t}{2}]$ &  $[\tfrac{\delta_t}{2}, \tfrac{\delta_t+w_t}{2}]$ &   $[\tfrac{\delta_t+w_t}{2}, \delta_t +\tfrac{w_t}{2}]$ &$[\delta_t +\tfrac{w_t}{2}, \delta_t+w_t]$ & $[\delta_t+w_t,1]$\\
 \hline
 $S_{t+1}$ & $S_t\setminus \{0\}$  &   $S_t\setminus \{\delta_t\}$ &$S_t\setminus \{\delta_t\}$ & $S_t\setminus \{\delta_t+w_t\}$ & $\exists s\in [\delta_t+w_t,$ \\
  &   &  & & &$1]\cap S_t: S_t\setminus \{s\}$\\ 
 \hline
 $\delta_{t+1}$ &  $\delta_{t}$ & $0$  &  $\delta_t+w_t$ & $\delta_t$ & $\delta_t$\\
 \hline
 $\mathbb{E}[\Delta\text{cost}_{t+1}|\ldots]$ &  $\geq 0$ & $\geq 0$ & $\geq
\begin{cases}
\frac{w_t}{2} &\text{ if }  w_t\leq \delta_t\\
 0 &\text{otherwise.}
\end{cases}
$ & $\geq 0$ & $\geq 0$\\
\end{tabular}
    \caption{Values of $(\delta_{t+1},S_{t+1})$ and expected value of  $\Delta\text{cost}_{t+1}$ conditioning on $(\delta_t, S_t)$ and on $r_{t+1}$, assuming that $S_t\cap \{0\}\neq \emptyset$, $\delta_t\neq 0$ and $|S_t\cap (\delta_t,1]|\geq 1$, and where $w_t:=s_{t,2}-s_{t,1}$.}
    \label{tab:MC_gamma_S}
\end{table}
\end{center}

In Lemma \ref{lem:simple_version_process}, we use the properties given in Lemma \ref{lem:MC_def} to  lower bound the probability that the gap $\delta$ has not yet disappeared at the time all servers in $(0,y]$ have been depleted, or that all the servers at location $0$ are depleted before either of these events occurs. We first recall that for any interval $I\subseteq [0,1]$, $t_{I}:= \min \{t\geq m|\;S_{t}\cap I = \emptyset\}$ is the time at which $I$ is depleted. We also define a couple additional stopping times for $\{(\delta_t, S_t)\}$.

\begin{definition}
\label{def:stopping_times_delta}
\hfill
\begin{itemize} 
    \item \textbf{Distance between $s_{t,2}$ and $s_{t,1}$ becomes large or $s_{t,2} = \emptyset$.} Let $t^w := \min\{t\geq m: s_{t,2} - s_{t,1}> s_{t,1}, \text{ or } s_{t,2}= \emptyset\}$.
    \item \textbf{$\delta$ disappears.} Let $t^d := \min\{t\geq m: \delta_t = 0\}$.
\end{itemize}
\end{definition}

\lemsimpleversionprocess*


We conclude this part by two simple properties. The first is about the initial gap $\delta_{m}$ just after matching request $r_m$.

\begin{restatable}{rLem}{claimEdeltam}
\label{claim:E_delta_m} 
The following properties hold:
\begin{enumerate}
    \item If $\delta_m>0$, then $r_m\in [0,y_0]$.
    \item For all $m\in [n]$, $\delta_m\in [0, 2n^{-1/5}]$.
    \item For all $m\in [c_1 n]$, $ \mathbb{E}[\delta_m|r_m \in [0,y_0]]\geq \frac{n^{-1/5}}{4}- \hp$.
\end{enumerate}
\end{restatable}

Finally, we show that if $R$ is regular, then for all $i\in [\io]$, the interval $(0,y_i]$ is depleted before all servers at location $0$ are depleted, and we upper bound the probability that all servers at location $0$ are depleted before $\delta$ disappears.

\begin{restatable}{rLem}{lemdepletedbefore}
\label{lem:depleted_before}
For all $m\in [n]$ and $i \in [d_1\log{n}]$, 
\begin{enumerate}
    \item if $R$ is regular, then $t_{(0,y_i]}\leq t_{\{0\}}$.
    \item $\P(t^d> t_{\{0\}}|r_m\in [0,y_0])= O(n^{-1/5})$.
\end{enumerate}

\end{restatable}

\paragraph{Lower bound on $\mathbb{E}[cost(\H^{m-1}) - cost(\H^m)]$ as a function of the gap $\delta$.} Using the structural properties stated above, we lower bound the expected difference of cost for matching requests $r_{m+1},\ldots, r_n$. 

\lemcostdeltagamma*

The full proof is in Appendix \ref{app_lower} and we only present here the main steps: by the second  property of Lemma \ref{lem:MC_def}, we have that while there still are some free servers at location $0$, the difference of cost $\Delta \text{cost}_{t+1}$ is always nonnegative. Moreover, we also have, by the third  property of Lemma \ref{lem:MC_def} (and the values given in Table 1) that as long as $\delta_t\neq 0$, $|S_t\cap (\delta_t, 1]|\geq 1$, $|S_t\cap \{0\}|\neq \emptyset$  and $w_t\geq \delta_t$, the expected value of $\Delta \text{cost}_{t+1}$ is at least the increase in $\delta$. A telescoping sum over all time steps yields the result.

We also give a simple lower bound on the expected difference of cost for matching requests $r_1,\ldots, r_m$. 

\begin{restatable}{rLem}{lemcostrm}
\label{lem:costrm} For all $m\in [n]$,
$\E\Big[ \sum_{t= 1}^m (cost_t(\H^{m-1}) - cost_t(\H^m))|r_m\in [0,y_0]\Big]\geq -n^{-1/5}.$
\end{restatable}

\paragraph{Main technical lemma.}
We are now ready to present the main technical lemma of this part,  which is a lower bound on the probability that the gap $\delta$ ever exceeds $y_{i-1}$ for all $i$ sufficiently small. 

\lemPdeltamax*

\begin{proof} Fix $m\in \{1,\ldots, c_1n\}$ and $i\in [\io]$.
For simplicity, we write $t_i$ to denote $t_{I_i}$, the time at which $I_i$ is depleted during the execution of $\H^m$, and we write $t_{y_i}$ to denote $t_{(0,y_i]}$, the time at which $(0, y_i]$ is depleted. 

In the remainder of the proof, we condition on the fact that the sequence of requests is regular. In particular, by Lemma \ref{lem:lb1}, we have that
\begin{equation}
\label{eq:reg}
    m< t_{1}<\ldots< t_{\io}\leq n-n^{c_3}.
\end{equation}

We start by lower bounding the probability that $\delta_{t_{y_i} -1}>0$ conditioning on the variables $\delta_m,S_m$. First, note that if $m< t_{y_i} \leq t^{d}$, then by definition of $t^d$, we have that $\delta_{t_{y_i} -1}>0$. In addition, by definition of $t_i, t_{y_i}$, and since $I_i =(y_{i-1}, y_i] \subseteq (0,y_i]$, we have $t_i\leq t_{y_i}$. Since by (\ref{eq:reg}), we have $m<t_i$, we get that $m<t_{y_i}$. Finally, since $R$ is regular, we also have, by Lemma \ref{lem:depleted_before}, that $\min(t_{y_i}, t_{\{0\}}) = t_{y_i}$. Hence, if $\min(t_{y_i}, t_{\{0\}}) \leq \min(t^d, t_{\{0\}})$, then $t_{y_i} = \min(t_{y_i}, t_{\{0\}})\leq \min(t^d, t_{\{0\}})\leq t^d$. Therefore, we have
\begin{align}
    &\mathbb{P}(\delta_{t_{y_i}  -1}>0|\reg, \delta_m, S_m) \nonumber\\
    &\geq  \mathbb{P}( m< t_{y_i} \leq t^{d}|\reg, \delta_m, S_m)\nonumber\\
    &= \mathbb{P}( t_{y_i} \leq t^{d}|\reg, \delta_m, S_m) \nonumber\\
    &\geq \mathbb{P}\Big(\min(t_{y_i}, t_{\{0\}}) \leq \min(t^d, t_{\{0\}})|\reg, \delta_m, S_m\Big)\nonumber\\
    &\geq \mathbb{P}\Big(\min(t_{y_i}, t_{\{0\}}) \leq \min(t^d, t_{\{0\}})| \delta_m, S_m\Big) - \hp &\text{(R is regular w.h.p.~by Lemma \ref{lem:reg})}\nonumber\\
    &\geq \frac{\delta_m}{y_i} - \hp. &\text{(Lemma \ref{lem:simple_version_process})}
    &\label{eq:lb_1_tyi}
\end{align}

Next, we assume that $\delta_{t_{y_i} -1}>0$ and we lower bound $\max_{t\in \{m,\ldots,  \min(n-n^{c_3}, t_{\{0\}})\}} \delta_t$. By definition of $\delta$, we have that for all $t\geq 0$, either $\delta_t = 0$ or $\delta_t = \min\{x>0|x\in S_{t}\}$. Since we assumed $\delta_{t_{y_i} -1}>0$, we thus have 
\begin{equation}
\label{eq:delta_yidef}
  \delta_{t_{y_i}-1} = \min\{x>0|x\in S_{t_{y_i}-1}\}.  
\end{equation}

Now, by (\ref{eq:reg}), we have that for all $j\leq i-1$, $t_j<t_i$ (i.e., $(y_{j-1}, y_j]$ is depleted before $(y_{i-1}, y_i]$). Recalling that $t_i = \min \{t\geq m : S_t\cap (y_{i-1}, y_i] = \emptyset\}$ and that $t_{y_i} = \min \{t\geq m : S_t\cap (0, y_i] = \emptyset\}$, we get that $t_i = t_{y_i}$ and that $(0,y_{i-1}]  \cap S_{t_{y_i}-1} = \left(\bigsqcup_{j=0}^{i-1} (y_{j-1},y_{j}] \right)\cap S_{t_{y_i}-1} = \emptyset$. Hence, 
\begin{equation}
\label{eq:min_Syi}
    \min\{x>0|x\in S_{t_{y_i}-1}\} \geq y_{i-1}.
\end{equation}

Combining (\ref{eq:min_Syi}) and (\ref{eq:delta_yidef}), we get that $\delta_{t_{y_i}-1} \geq y_{i-1}$. In addition, since $m<t_i\leq n-n^{c_3}$ by (\ref{eq:reg}) and $t_{y_i} =t_i$ as argued above, we have that $m<t_{y_i}\leq n-n^{c_3}$. Since $R$ is regular, we also have, by Lemma \ref{lem:depleted_before}, that $t_{y_i}\leq t_{\{0\}}$. We deduce that $\max_{t\in \{m,\ldots,  \min(n-n^{c_3}, t_{\{0\}})\}} \delta_t\geq \delta_{t_{y_i}-1}\geq y_{i-1}$.
As a result,
\begin{align*}
    \P\Big(\max_{t\in \{m,\ldots,  \min(n-n^{c_3}, t_{\{0\}})\}} \delta_t \geq y_{i-1}|\reg, \delta_m, S_m\Big)\geq \P(\delta_{t_{y_i} -1}>0|\reg,\delta_m, S_m).
\end{align*}

Combining this with (\ref{eq:lb_1_tyi}), we finally obtain
\[\P\Big(\max_{t\in \{m,\ldots,  \min(n-n^{c_3}, t_{\{0\}})\}} \delta_{t} \geq y_{i-1}|\reg, \delta_m, S_m\Big) \geq\frac{\delta_m}{y_i}-n^{-\Omega(\log(n))} .\]
\end{proof}

\paragraph{Concluding the proof.} We now present the proof of Lemma \ref{lem:m_cn2}, that we restate below for convenience.
\lemmnc*

\begin{proof}
Let $m\in  [n]$. Since $\H^m$ and $\H^{m-1}$ make the same decisions at all time steps when $r_m\in (y_0,1]$, it is immediate that $
\mathbb{E}[cost(\H^{m-1})- cost(\H^m)|r_m \in (y_0,1]]=0$, which shows the third point of the lemma. 

We now show the first two points. By Lemma \ref{cor:costLdelta_gamma}, we have that

\begin{align}
\label{eq:firsteq_LBcost}
    \mathbb{E}\Big[ \sum_{t= m+1}^n (cost_t(\H^{m-1}) - cost_t(\H^m))|\delta_m,S_m\Big]\geq \frac{1}{2}\mathbb{E}\Big[&\max_{t \in \{0,\ldots, \min(t_{\{0\}}, t_w)-m\}}\delta_{t+m}- \delta_{m} |\delta_m,S_m\Big]\nonumber\\
    & \qquad\qquad\qquad - \P(t^d>t_{\{0\}}|\delta_m,S_m).
\end{align}

Thus, we first get
\begingroup
\allowdisplaybreaks
\begin{align*}
    &\mathbb{E}[cost(\H^{m-1}) - cost(\H^m)|r_m \in [0,y_0]] \\
    &=\mathbb{E}\Big[ \sum_{t= m+1}^n (cost_t(\H^{m-1}) - cost_t(\H^m))|r_m \in [0,y_0]\Big] \\
    &\qquad\qquad+ \mathbb{E}\Big[ \sum_{t= 1}^m (cost_t(\H^{m-1}) - cost_t(\H^m))|r_m \in [0,y_0]\Big]\\
    &\geq\mathbb{E}\Big[ \sum_{t= m+1}^n (cost_t(\H^{m-1}) - cost_t(\H^m))|r_m \in [0,y_0]\Big]- n^{-1/5}\\
    &= \int_{(x,S)\in \mathcal{X}}\mathbb{E}\Big[ \sum_{t= m+1}^n (cost_t(\H^{m-1}) - cost_t(\H^m))|(\delta_m,S_m) = (x,S), r_m \in [0,y_0]\Big]\\
    &\qquad\qquad\qquad\qquad\qquad\qquad\qquad\qquad\qquad\qquad\qquad\qquad \cdot\text{d}\P((x,S)|r_m \in [0,y_0] ) - n^{-1/5}\\
    &= \int_{(x,S)\in \mathcal{X}} \mathbb{E}\Big[ \sum_{t= m+1}^n (cost_t(\H^{m-1}) - cost_t(\H^m))|(\delta_m,S_m) = (x,S)\Big]\cdot\text{d}\P((x,S)|r_m \in [0,y_0] ) - n^{-1/5}\\
    &\geq  \int_{(x,S)\in \mathcal{X}}  [0 -\P(t^d>t_{\{0\}}|(\delta_m,S_m)= (x,S))]\cdot \text{d}\P((x,S)|r_m \in [0,y_0] ) - n^{-1/5}\\
    &=  \int_{(x,S)\in \mathcal{X}}  [0 -\P(t^d>t_{\{0\}}|(\delta_m,S_m)= (x,S)),r_m \in [0,y_0]  ]\cdot \text{d}\P((x,S)|r_m \in [0,y_0] ) - n^{-1/5}\\
    & =  -\P(t^d>t_{\{0\}}|r_m \in [0,y_0])-n^{-1/5}\\
    &= -O(n^{-1/5}),
\end{align*}
\endgroup

where the first inequality is by Lemma \ref{lem:costrm}, the second and fourth equalities are since conditioned on $(\delta_m,S_m)$, $\{(cost_t(\H^m),  cost_t(\H^{m-1}))\}_{t\geq m+1}$ is independent on $r_m$, the second inequality is  by (\ref{eq:firsteq_LBcost}) and the last equality by Lemma \ref{lem:depleted_before}.
This completes the proof of the first part  of Lemma \ref{lem:m_cn2}.

Next, we prove the second point of the lemma by providing a tighter lower bound on (\ref{eq:firsteq_LBcost}) when $m\leq c_1n$. In the remainder of the proof, we consider a fixed $m\in \{1,\ldots, c_1n\}$.

First, we show that $t^w\geq n-n^{c_3}$. Note that if $R$ is regular, then by Lemma \ref{prop:l_nbeta}, we have that for all $t\in [n-n^{c_3}]$, $s_{t,2}-s_{t,1}\leq2\log(n)^4n^{1-2c_3}$. Thus, for $n$ large enough (and since we chose $c_3>4/5$), we have that $s_{t,2}-s_{t,1}< n^{-1/5}$. Since by definition of the instance, it is always the case that $s_{t,1}>y_0 = n^{-1/5}$, we thus have $s_{t,2}-s_{t,1}< s_{t,1}$. In addition, since $t\leq n-n^{c_3}$ and $c_3>4/5$, we have that for $n$ large enough, $S_t\cap (0,1]\geq S_0\cap (0,1] - (n-n^{c_3}) = n-(n^{4/5}+\excess) - (n-n^{c_3}) = n^{c_3} - n^{4/5} - \excess >2$, thus $s_{t,2}\neq \emptyset$. Since $t^w = \min\{t\geq m: s_{t,2} - s_{t,1}> s_{t,1}, \text{ or } s_{t,2}= \emptyset\}$, we thus have $t\leq t^w$. Hence $t^w\geq n-n^{c_3}$.

Therefore,
\begin{align}
\label{eq:n_tw}
    \mathbb{E}\Big(\max_{t \in \{0,\ldots, \min( t_{\{0\}}, t^w)-m\}} &\delta_{t+m} |\delta_m,S_m,\text{R is regular}\Big)\nonumber \\
    &\geq \mathbb{E}\Big(\max_{t \in \{0,\ldots, \min(t_{\{0\}}, n-n^{c_3})-m\}} \delta_{t+m}|\delta_m,S_m,\text{R is regular}\Big).
\end{align}

Next, 
\begin{align}
    \mathbb{E}\Big(&\max_{t \in \{0,\ldots, \min(t_{\{0\}}, n-n^{c_3})-m\}} \delta_{t+m}|\delta_m,S_m,\text{R is regular}\Big) \nonumber\\
    &= \int_0^1 \mathbb{P}\Big(\max_{t \in \{0,\ldots, \min(t_{\{0\}}, n-n^{c_3})-m\}} \delta_{t+m}\geq x|\delta_m,S_m, \text{R is regular}\Big)dx\nonumber\nonumber\\
    &\geq \sum_{i=0}^{\io} \int_{x\in I_i} \mathbb{P}\Big(\max_{t \in \{0,\ldots, \min(t_{\{0\}}, n-n^{c_3})-m\}} \delta_{t+m}\geq x|\delta_m,S_m, \text{R is regular}\Big)dx\nonumber\nonumber\\
    &\geq \sum_{i=0}^{\io} (y_i-y_{i-1})\cdot \mathbb{P}\Big(\max_{t \in \{0,\ldots, \min(t_{\{0\}}, n-n^{c_3})-m\}} \delta_{t+m}\geq y_i|\delta_m,S_m, \text{R is regular}\Big)\nonumber\nonumber\\
     &\geq \sum_{i=0}^{\io} (y_i-y_{i-1})\cdot\frac{\delta_m}{y_{i+1}} -\hp \nonumber\\
    &=n^{-1/5}\cdot\frac{\delta_m}{(3/2)n^{-1/5}} + \sum_{i=1}^{\io} \frac{(3/2)^{i-1}n^{-1/5}}{2}\cdot \frac{\delta_m}{(3/2)^{i+1}n^{-1/5}} -n^{-\Omega(\log(n))}\nonumber\nonumber\\
    &= C \delta_m \log(n) -n^{-\Omega(\log(n))},\label{eq:last_integral}
\end{align}
for some constant $C>0$.
The first inequality holds  since $\bigsqcup_{i=1}^{\io} I_i\subseteq[0,1]$, the second inequality is since $I_i=(y_{i-1}, y_i]$ and the third inequality results from Lemma \ref{lem:P_delta_max}. Finally, the second equality is since $y_i = (3/2)^in^{-1/5}$ for all $i\in \{0,\ldots, \io\}$ and since $y_{-1} = 0$.

Thus, we get
\begin{align*}
    &\mathbb{E}\Big(\max_{t \in \{0,\ldots, \min( t_{\{0\}}, t^w)-m\}} \delta_{t+m} |\delta_m, S_m\Big) \\
    &\geq \mathbb{E}\Big(\max_{t \in \{0,\ldots, \min( t_{\{0\}}, t^w)-m\}} \delta_{t+m} |\delta_m, S_m, \reg\Big)\P(\reg) \\
    &\geq \mathbb{E}\Big(\max_{t \in \{0,\ldots, \min(t_{\{0\}}, n-n^{c_3})-m\}} \delta_{t+m}|\delta_m,S_m,\text{R is regular}\Big)\P(\reg) \\
    &\geq (C \delta_m \log(n) -n^{-\Omega(\log(n))})(1-\hp)\nonumber\\
    &= C \delta_m \log(n) -n^{-\Omega(\log(n))},
    \label{eq:d_max}
\end{align*}
where the second inequality is by (\ref{eq:n_tw}) and the third one by (\ref{eq:last_integral}) and the fact that $R$ is regular with high probability by Lemma \ref{lem:reg}.

Combining this with (\ref{eq:firsteq_LBcost}) gives:
\begin{align*}
    \mathbb{E}\Big[ \sum_{t= m+1}^n (cost_t(\H^{m-1}) - cost_t(\H^m))|\delta_m,S_m\Big]&\geq \frac{1}{2}[C \delta_m \log(n) -n^{-\Omega(\log(n))}-\delta_m] \\
    &\qquad\qquad- \P(t^d>t_{\{0\}}|\delta_m,S_m).
\end{align*}

Finally, similarly as  for the first point, we get
\begin{align*}
    &\mathbb{E}[cost(\H^{m-1}) - cost(\H^m)|r_m \in [0,y_0]] \\
    &\geq \int_{(x,S)\in \mathcal{X}} \mathbb{E}\Big[ \sum_{t= m+1}^n (cost_t(\H^{m-1}) - cost_t(\H^m))|\delta_m,S_m\Big]\cdot\text{d}\P((x,S)|r_m \in [0,y_0] )  - n^{-1/5}\\
    &\geq \int_{(x,S)\in \mathcal{X}} \left(\frac{1}{2}[C x \log(n) -n^{-\Omega(\log(n))}-x]- \P(t^d>t_{\{0\}}|(\delta_m,S_m)= (x,S)))  \right)\\
    &\qquad\qquad\qquad\qquad\qquad\qquad\qquad\qquad\qquad\qquad\qquad\qquad\cdot\text{d}\P((x,S)|r_m \in [0,y_0] )-n^{-1/5}\\
    &\geq \E[\delta_m|r_m \in [0,y_0]]\cdot(\tfrac{1}{2}C\log(n)-1) - \hp - \P(t^d>t_{\{0\}}|r_m \in [0,y_0])-n^{-1/5}\\
    &\geq \E[\delta_m|r_m \in [0,y_0]]\cdot(\tfrac{1}{2}C\log(n)-1) - \hp- O(n^{-1/5})\\
    &\geq \Big(\frac{n^{-1/5}}{4}- \hp\Big)\cdot(\tfrac{1}{2}C\log(n)-1) - \hp- O(n^{-1/5})\\
    &= \Omega(\log(n)n^{-1/5}),
\end{align*}
where the fourth inequality is by Lemma \ref{lem:depleted_before} and the fifth one by Lemma \ref{claim:E_delta_m}. This concludes the proof of the second point and the proof of the lemma. \end{proof}






%% file: appendix_LB_full_proof.tex
\section{Missing Analysis from Section \ref{sec:full_proof}}
\label{appendix:full_LB}

In this section, we re-state and prove all statements that were claimed, but not proved in Section~\ref{sec:full_proof}, as well as provide some auxiliary facts and definitions.

\subsection{Missing analysis from Section \ref{subsec:preliminaries}}
\factnumber*

\begin{proof}
Let $x,y\in [0,1]$ such that $x\leq y$ and let $k,j\in \mathbb{N}$ be such that $\frac{k}{\tilde{n}}\leq x\leq \frac{k+1}{\tilde{n}}$ and $\frac{j-1}{\tilde{n}}\leq y\leq \frac{j}{\tilde{n}}$. Then by construction of $\mathcal{I}_n$, the number of servers in the interval $[x,y]$ is in $\{j-k -1, j-k, j-k+1 \}$, and by definition of $k,j$, we have $j-k-2\leq \tilde{n}(y-x)\leq j-k$. Hence, for all $z> \tilde{n}(y-x)+3$ or $z< \tilde{n}(y-x)-1$, we have $z\notin \{j-k -1, j-k, j-k+1 \}$; hence $|S_0\cap [x,y]|\in [\tilde{n}(y-x)-1, \tilde{n}(y-x)+3]$.
\end{proof}

\lemlnbeta*

\begin{proof} Note that if the statement of the lemma holds for $t = n-n^{c_3}$, then it holds for all $t\in [n- n^{c_3}]$ since $S_t\supseteq S_{n-n^{c_3}}$ when $t\leq n-n^{c_3}$. Hence it suffices to consider the case $t = n-n^{c_3}$.

Now, consider $s, s'\in S_{n-n^{c_3}}\cap (0,1]$ such that $s'-s>2\log(n)^4n^{1-2c_3}$. Note that if there exists $s''\in S_{n-n^{c_3}}$ such that $s<s''<s'$, we are done. In the remainder of the proof, we show that there is such an $s''$.  By definition of $\H^m$, each request is either matched greedily, or it is matched to $0$. Hence, for all $j\in [n-n^{c_3}]$, if $r_j\notin (s, s')$, then $s_{\H^m}(r_j)\notin (s, s')$ (since $r_j$ is closer to either $s$ or $s'$ than any point in $(s, s')$, and both $s$ and $s'$ are available when $r_j$ arrives). Similarly, by the greediness of $\H^m$ for requests $r>y_0$ and since $s> y_0$ (by definition of $S_0$ and since $S_{n-n^{c_3}}\subseteq S_0$), if $r_j\in (s, s')$, then $s_{\H^m}(r_j)\in (s, s')$ for all $j \leq n - n^{c_3}$. Therefore,
\begin{equation*}
|\{ j\in [n-n^{c_3}] : s_{\H^m}(r_j)\in (s, s')\}| = |\{ j\in [n-n^{c_3}] : r_j\in (s, s')\}|.
\end{equation*}

Note that $d^+((s, s'))\cdot n^{1-2c_3}\geq  (s'-s)n^{1-2c_3}\geq 2\log(n)^4n^{1-2c_3} (n-n^{c_3}) = \Omega(1)$. Hence, by applying the first regularity condition with $t=0$, $t' = n-n^{c_3}$, and $(d,d') = d^+((s, s'))$, we have

\begin{align*}
    |\{ j\in [n-n^{c_3}] :\; r_j \in (s, s')\}|&\leq |\{ j\in [n-n^{c_3}] :\; r_j \in d^+((s, s'))\}| \\
    &\leq (n-n^{c_3})\cdot|d^+((s, s'))| + \log(n)^2\sqrt{(n-n^{c_3})\cdot|d^+((s, s'))|}\\
    &\leq (n-n^{c_3})(s'-s + \tfrac{2}{n}) + \log(n)^2\sqrt{(n-n^{c_3})(s'-s + \tfrac{2}{n})}\\
    &\leq (n-n^{c_3})(s'-s) + \log(n)^2\sqrt{(n-n^{c_3})(s'-s)} + 4\log(n)^2.
\end{align*}

By combining the previous inequality with the previous equality, we obtain that
\begin{equation}
    \label{eq:prop_lnbeta}
    |\{ j\in [n-n^{c_3}] : s_{\H^m}(r_j)\in (s, s')\}| \leq (n-n^{c_3})(s'-s) + \log(n)^2\sqrt{(n-n^{c_3})(s'-s)} + 4\log(n)^2.
\end{equation}

Now, since we assumed $s'-s>2\log(n)^4n^{1-2c_3}$, we have, for $n$ large enough and since $c_3\in (4/5,1)$, that $\left(1-4\log(n)^2n^{1/2-c_3}/(1-n^{-1/5})- (4\log(n)^2 +1)/((s'-s)n^{c_3})\right)^{-1}\leq 2$, thus we also have
\begin{align*}
    (s'-s)&>2\log(n)^4n^{1-2c_3}\\
    & \geq 2\log(n)^4\frac{(n-n^{c_3})}{n^{2c_3}}\\&\geq \log(n)^4\frac{(n-n^{c_3})}{n^{2c_3}(1-4\log(n)^2n^{1/2-c_3}/(1-n^{-1/5})- (4\log(n)^2 +1)/((s'-s')n^{c_3})}\\
    &= \log(n)^4\frac{(n-n^{c_3})}{(n^{c_3}-4\log(n)^2n^{1/2}/(1-n^{-1/5}) - (4\log(n)^2 +1)/(s'-s)) ^2}\\
    &= \log(n)^4\frac{(n-n^{c_3})}{((n-4\log(n)^2n^{1/2}/(1-n^{-1/5}) - (4\log(n)^2 +1)/(s'-s)) - (n-n^{c_3}))^2}\\
    &=\log(n)^4\frac{(n-n^{c_3})}{(\tilde{n} - (4\log(n)^2 +1)/(s'-s)-(n-n^{c_3}))^2},
\end{align*}
 
which implies that $(s'-s)^2(\tilde{n} - (4\log(n)^2 +1)/(s'-s) - (n-n^{c_3}))^2>\log(n)^4(n-n^{c_3})(s'-s)$. By taking the square root on both sides and reorganizing the terms, this gives
\[
(s'-s)\tilde{n} -1 > (n-n^{c_3})(s'-s) + \log(n)^2\sqrt{(n-n^{c_3})(s'-s)} + 4\log(n)^2.
\]

Combining this with (\ref{eq:prop_lnbeta}) and using that $(s'-s)\tilde{n}-1 \leq |S_0\cap (s,s')|$ by  Lemma\ref{fact:instance}, we obtain
\[
|\{ j\in [n-n^{c_3}] | s_{\H^m}(r_j)\in (s, s')\}|< (s'-s)\tilde{n}-1 \leq |S_0\cap (s,s')|,
\]
Hence, 
\begin{align*}
    |S_{n-n^{c_3}}\cap (s,s') | = |S_0\cap (s,s') \setminus \{ j\in [n-n^{c_3}] | s_{\H^m}(r_j)\in (s, s')\}| > 0.
\end{align*}
Thus, there exists $s'' \in S_{n-n^{c_3}}$ such that $s< s'' < s'$, which concludes the proof.
\end{proof}

\subsection{Missing analysis from Section \ref{subsec:analysis_of_S}}

\IiIi*

\begin{proof} We treat separately the cases $i=1$ and $i>1$. The case $i=1$ is immediate since by construction of the instance, $I_0\cap S_0 = (0, n^{-1/5}] \cap S_0 =  \emptyset$ whereas $I_1\cap S_0 = (n^{-1/5}, \frac{3}{2} n^{-1/5}] \cap S_0  \neq \emptyset$, which implies $t_0 = 0  <t_1$.

Next, assume $i\in \{2,\ldots, d_1\log(n)\}$. By definition of $t_{i-1}$, we have that $(y_{i-2},y_{i-1}]\cap S_{t_{i-1} - 1} \neq \emptyset$, which implies that $s_{t_{i-1} - 1,1} \leq y_{i-1}.$ In addition, because of the assumptions and by Lemma~\ref{claim:dlogn}, we have $ t_{i-1} - 1 < n - (1-c_2)^{i-1}n \leq n - n^{c_3}$. Hence, by applying Lemma~\ref{prop:l_nbeta} with $t=t_{i-1}$, we deduce $(y_{i-1},y_{i-1}  + 2\log(n)^4n^{1-2c_3}]\cap S_{t_{i-1} - 1}\neq \emptyset$. Then, since $c_3>3/4$ and since we assumed $n$ sufficiently large, we have that $y_{i-1}  + 2\log(n)^4n^{1-2c_3} = (3/2)^{i-1}n^{-1/5} + 2\log(n)^4n^{1-2c_3} \leq (3/2)^{i}n^{-1/5} = y_{i}$. Thus, we get that $(y_{i-1},y_{i}]\cap S_{t_{i-1} - 1} \neq \emptyset$. By definition of $t_{i},$ this implies that $t_{i} > t_{i-1} -1$. In addition, since $I_{i-1}$ and $I_{i}$ are disjoint, at most one of $I_{i-1}$ and $I_{i}$ can be depleted at each time step and we have that $t_{i-1} \neq t_{i}$. We conclude that $t_{i} > t_{i-1}$.
\end{proof}

\propfi*

\begin{proof}
Let $i\in [d_1\log(n)]$. We first upper bound $|\{j\in [\overline{t_i}]:r_j\in I_i, s_{\H^m}(r_j)>y_i\}|$. 

Since $\overline{t_i} = \min (t_i,t_{i-1} + c_2(n-t_{i-1}))\leq t_i$, we have that $(y_{i-1},y_{i}]\cap S_{\overline{t_i} - 1} \neq \emptyset$ by definition of $t_i$, which implies that $s_{\overline{t_i} - 1,1} \leq y_i$. Now, using the definition of $\overline{t_i}$, the assumption that $t_{i-1}\leq n - (1-c_2)^{i-1}n$,  and  Lemma~\ref{claim:dlogn}, we have $\overline{t_i} - 1 \leq t_{i-1} + c_2(n-t_{i-1}) -1 = c_2 n + t_{i-1}(1-c_2)-1< n - (1-c_2)^in \leq n - n^{c_3}$. Hence, by Lemma \ref{prop:l_nbeta}, applied with $t = \overline{t_i} - 1$, we obtain that there is $ s\in (y_{i} - 2\log(n)^4n^{1-2c_3},y_{i}]\cap S_{\overline{t_i}- 1}$. We let $s$ be such a server. Then, by the greediness of $\H^m$ for all requests $r> y_0$, we have that for all $j\in [\overline{t_i}]$, if $r_j\leq y_i - 2\log(n)^4n^{1-2c_3}$, then $s_{\H^m}(r_j)\leq  s\leq y_i$ (since $r_j$ is closer to $s$ than any point in $(s, y_i]$, and $s$ is available when $r_j$ arrives). 

Now, note that since we assumed $c_3>3/4$ and $n$ large enough, we have $y_i - 2\log(n)^4n^{1-2c_3} \geq  (3/2)^in^{-1/5}- 2\log(n)^4n^{-1/4}>(3/2)^{i-1}n^{-1/5} = y_{i-1}$. 

Hence, we can write $I_i = (y_{i-1}, y_i - 2\log(n)^4n^{1-2c_3}]$ $\cup (y_i - 2\log(n)^4n^{1-2c_3}, y_i]$. Since we have shown that $s_{\H^m}(r_j)\leq y_i$ for all $j$ such that $r_j\leq y_i - 2\log(n)^4n^{1-2c_3}$, we thus obtain 
\begin{align}
  |\{j\in [\overline{t_i}]&:r_j\in I_i, s_{\H^m}(r_j)>y_i\}| \nonumber\\
  & =|\{j\in [\overline{t_i}]:r_j\in (y_i - 2\log(n)^4n^{1-2c_3}, y_i], s_{\H^m}(r_j)>y_i\}| \nonumber \\
  & \leq|\{j\in [\overline{t_i}]:r_j\in (y_i - 2\log(n)^4n^{1-2c_3}, y_i]\}| \nonumber \\
  &\leq |\{j\in [n]:r_j\in d^+([y_i - 2\log(n)^4n^{1-2c_3}, y_i])\}| \nonumber \\
    &\leq |d^+([y_i - 2\log(n)^4n^{1-2c_3}, y_i])| \cdot n + \log^2(n) \sqrt{|d^+([y_i - 2\log(n)^4n^{1-2c_3}, y_i])| \cdot n} \nonumber \\
    &\leq (2\log(n)^4n^{1-2c_3} + 2/n) \cdot n + \log^2(n) \sqrt{(2\log(n)^4n^{1-2c_3} + 2/n)  \cdot n} \nonumber \\
  &= \tilde{O}(\sqrt{n})\label{eq:upbound_R_fi}.
\end{align}
where the third inequality is by the second regularity condition, applied with  $t=0$, $t'=n$, $[d,d'] = d^+([y_i - 2\log(n)^4n^{1-2c_3}, y_i])$, which satisfy the condition $(t-t')(d-d') = n\cdot 2\log(n)^4n^{1-2c_3}= \Omega(1)$ since $c_3<1$. The fourth inequality is by definition of $d^+(\cdot)$, and the fifth is since $c_3>3/4$.

Next, we upper bound $|\{j\in [\overline{t_i}]:r_j\in I_i, s_{\H^m}(r_j)\leq y_{i-1}\}|$.  We treat separately the cases where $j\in [t_{i-1}]$ and $j\in \{t_{i-1}+1, \overline{t_i}\}$.

First, consider the case $j\in [t_{i-1}]$. If $i=1$, then by construction of the instance, $I_0\cap S_{0} = \emptyset$, hence $t_{0} = 0$ and we have the trivial identity $|\{j\in [t_{0}]:r_j\in I_1, s_{\H^m}(r_j)>y_1\}| = 0$. Now, for $i>1$, by definition of $t_{i-1}$, we have that $(y_{i-2},y_{i-1}]\cap S_{t_{i-1} - 1} \neq \emptyset$, which implies that $s_{t_{i-1} - 1,1} \leq y_{i-1}$. Since by assumption and by Lemma~\ref{claim:dlogn}, we have $t_{i-1}-1<n - (1-c_2)^{i-1}n \leq n-n^{c_3}$, we obtain, by applying Lemma \ref{prop:l_nbeta} at time $t = t_{i-1}-1$ and by a similar argument as in  \eqref{eq:upbound_R_fi}:
\begin{align}
  |\{j\in [t_{i-1}]:r_j\in I_i, s_{\H^m}(r_j)\leq y_{i-1}\}| & \leq|\{j\in [n]:r_j\in d^+([y_{i-1}, y_{i-1} + 2\log(n)^4n^{1-2c_3}])\}| = \tilde{O}(\sqrt{n})\label{eq:upbound_R_fi_2}
\end{align}

Now, for all $j\in \{t_{i-1}+1, \overline{t_i}\}$, since $j\geq t_{i-1}+1$ and $t_0<\ldots< t_{i-1}$ by assumption, we have that $(0,y_{i-1}]\cap S_{j-1} = \left(\bigcup_{k\in [i-1]}(y_{k-1}, y_k] \right)\cap S_{j-1} = \emptyset$ by definition of $t_0,\ldots, t_{i-1}$; and since  $j \leq \overline{t_i}\leq t_i$, we have that $(y_{i-1}, y_i]\cap S_{j-1}\neq \emptyset$ by definition of $t_i$. Hence, if $r_j\in I_i$, we either have $s_{\H^m}(r_j) > y_{i-1}$ or $s_{\H^m}(r_j) =0$. By the greediness of $\H^m$ for all $r> n^{-1/5}$ and since $|y_i - r_j|\leq |y_i - y_{i-1}| = |(3/2)^{i}n^{-1/5} - (3/2)^{i-1}n^{-1/5}| = 1/2(3/2)^{i-1}n^{-1/5} \leq |r_j-0|$ for any $r_j\in I_i$, we have that  $s_{\H^m}(r_j)> y_{i-1}$ for any $r_j\in I_i$. 
Hence, we get that
\[|\{j\in \{t_{i-1}+1, \overline{t_i}\}:r_j\in I_i, s_{\H^m}(r_j)\leq y_{i-1}\}| =0.\]
Combining this with (\ref{eq:upbound_R_fi_2}), we obtain:
\begin{equation}
\label{eq:upbound_R_fi_3}
 |\{j\in [\overline{t_i}]:r_j\in I_i, s_{\H^m}(r_j)\leq y_{i-1}\}| =  \tilde{O}(\sqrt{n}).   
\end{equation}
Finally, from (\ref{eq:upbound_R_fi}) and (\ref{eq:upbound_R_fi_3}), we get
\begin{align*}
\label{eq:upbound_R_fi_3}
    &|\{j\in [\overline{t_i}]:r_j\in I_i, s_{\H^m}(r_j)\notin I_i\}|\\
    &=  |\{j\in [\overline{t_i}]:r_j\in I_i, s_{\H^m}(r_j)>y_i\}| + |\{j\in [\overline{t_i}]:r_j\in I_i, s_{\H^m}(r_j)\leq y_{i-1}\}|
    =\tilde{O}(\sqrt{n}).
\end{align*}
\end{proof}

\propgi*

\begin{proof}
Consider $j \in \{t_{i-1} + 1 +c_1(n - t_{i-1}),\ldots ,\overline{t_i}\}$ and assume that $r_j \in [\tfrac{3}{4} y_{i-1}, y_{i-1}]$. Since $j\geq t_{i-1}+1$ and $t_0<\ldots< t_{i-1}$ by assumption, we have $j > t_\ell$ for all $\ell \in [i-1]$. Thus, by definition of $t_0,\ldots, t_{i-1}$, we have that $(0,y_{i-1}]\cap S_{j-1} = \left(\bigcup_{k\in [i-1]}(y_{k-1}, y_k] \right)\cap S_{j-1} = \emptyset$. In addition, since  $j \leq \overline{t_i}\leq t_i$, we have that $(y_{i-1}, y_i]\cap S_{j-1}\neq \emptyset$ by definition of $t_i$. Thus, by the greediness of $\H^m$ for all  $j\geq t_{i-1} + 1 +c_1(n - t_{i-1})\geq c_1n\geq m$, and since $r_j \in [\tfrac{3}{4} y_{i-1}, y_{i-1}]$, we either have $s_{\H^m}(r_j)\in (y_{i-1}, y_i]$ or $s_{\H^m}(r_j) = 0$. 
 Now, since $(y_{i-1}, y_i]\cap S_{j-1}\neq \emptyset$, and since for any $s\in (y_{i-1}, y_i]\cap S_{j-1}$, we have 
 \[
|s-r_j|\leq |y_{i}- \tfrac{3}{4} y_{i-1}| = (3/2)^{i-1}n^{-1/5}[\tfrac{3}{2}-\tfrac{3}{4}] = |\tfrac{3}{4} y_{i-1}|\leq |r_j-0|,
\]
we must have $s_{\H^m}(r_j) \in (y_{i-1}, y_i]$. Hence, 
\begin{align}
&|\{j \in \{t_{i-1} + c_1(n - t_{i-1}),\ldots ,\overline{t_i}\}:r_j\in [\tfrac{3}{4} y_{i-1}, y_{i-1}]\}|\nonumber\\
& = |\{j \in \{t_{i-1} + c_1(n - t_{i-1}),\ldots ,\overline{t_i}\}:r_j\in [\tfrac{3}{4} y_{i-1}, y_{i-1}], s_{\H^m}(r_j)\in (y_{i-1}, y_i]\}|\label{eq:34}.
\end{align}

Now, since we assumed that the sequence of requests is regular, by applying the first regularity condition with $t = t_{i-1} + c_1(n - t_{i-1})$, $t' = \overline{t_i}$ and $[d,d'] = d^-([\tfrac{3}{4} y_{i-1}, y_{i-1}])$), we have that
\begin{align}
|\{j \in &\{t_{i-1}+1 + c_1(n - t_{i-1}),\ldots ,\overline{t_i}\}:r_j\in [\tfrac{3}{4} y_{i-1}, y_{i-1}]\}\nonumber\\
&\geq|\{j \in \{t_{i-1}+1 + c_1(n - t_{i-1}),\ldots ,\overline{t_i}\}:r_j\in d^-([\tfrac{3}{4} y_{i-1}, y_{i-1}])\}|\nonumber\\
&\geq d^-([\tfrac{3}{4} y_{i-1}, y_{i-1}])\cdot (\overline{t_i} - t_{i-1} - c_1 (n - t_{i-1})-1) \\&\qquad\qquad\qquad\qquad\qquad\qquad -\log(n)^2\sqrt{d^-([\tfrac{3}{4} y_{i-1}, y_{i-1}])\cdot (\overline{t_i} - t_{i-1} - c_1 (n - t_{i-1})-1)}.\nonumber\\
&\geq (y_{i-1}-\tfrac{3}{4} y_{i-1}-2/n) (\overline{t_i} - t_{i-1} - c_1 (n - t_{i-1})-1) \nonumber\\
&\qquad\qquad\qquad\qquad\qquad\qquad-\log(n)^2\sqrt{(y_{i-1}-\tfrac{3}{4} y_{i-1}-2/n) (\overline{t_i} - t_{i-1} - c_1 (n - t_{i-1})-1)}.\nonumber\\
&\geq (\overline{t_i} - t_{i-1} - c_1 (n - t_{i-1})) \frac{y_{i-1}}{4}- \tilde{O}(\sqrt{n})\label{eq:LBti}.
\end{align}
By combining (\ref{eq:LBti}) and (\ref{eq:34}), and noting that $|I_i| = (3/2)^{i}n^{-1/5} - (3/2)^{i-1}n^{-1/5} = (3/2)^{i-1}n^{-1/5}\cdot \frac{1}{2} =  \frac{y_{i-1}}{2}$, we finally obtain that
\begin{align*}
&|\{j \in \{t_{i-1} + c_1(n - t_{i-1}),\ldots ,\overline{t_i}\}:r_j\in [\tfrac{3}{4} y_{i-1}, y_{i-1}], s_{\H^m}(r_j)\in (y_{i-1}, y_i]\}|\\
&\geq \frac{1}{2}(\overline{t_i} - t_{i-1} - c_1 (n - t_{i-1}))|I_i|- \tilde{O}(\sqrt{n}).
\end{align*}
\end{proof}

\propmt*

\begin{proof} First, note that, since $t_1$ is the time at which $I_1$ is depleted, we have, using  Lemma \ref{fact:instance}, that
\begin{equation}
\label{eq:upt1}
    |\{j\in [t_1]: s_{\H^m}(r_j) \in I_1\}| = |S_0\cap I_1|\geq |I_1|\tilde{n} -1.
\end{equation}
 In particular, $t_1 \geq |\{j\in [t_1]: s_{\H^m}(r_j) \in I_1\}|\geq |I_1|\tilde{n} -1$; thus for all $i\geq 0$, we have 
\begin{equation}
\label{eq:t1_lbtrivial}
    t_1|I_i|= \Omega(|I_1||I_i| \tilde{n}) = \Omega(1).
\end{equation}

Next, we upper bound $|\{j\in [t_1] : s_{\H^m}(r_j)\in I_1\}|$. Since $t_1<t_2$ by assumption, for all $j\in [t_1]$, we have that $(y_1,y_2]\cap S_{j-1}\neq \emptyset$ by definition of $t_2$. Hence, by the greediness of $\H^m$ for all requests $r> y_0$, if $r_j> y_2$, we have $s_{\H^m}(r_j)> y_1$ (since $r_j$ is closer to any $s\in (y_1,y_2]\cap S_{j-1}$ than any point in $[0,y_1]$). We thus get
\[
|\{j\in [t_1]\;| s_{\H^m}(r_j)\in I_1\}|\leq |\{j\in [t_1]: r_j \leq y_2\}| 
\]
Now, note that by (\ref{eq:t1_lbtrivial}), we have $t_1|d^+(I_i)| = \Omega(1)$ for all $i$. Hence, by applying the second regularity condition with $t=0$, $t'= t_1$, and $[d,d'] = d^+(I_0), d^+(I_1), d^+(I_2)$, respectively, we get
\begin{align*}
|\{j\in [t_1]&: r_j \leq y_2\}| \\
&\leq |\{j\in [t_1]: r_j \in d^+(I_0)\}| + |\{j\in [t_1]: r_j \in d^+(I_1)\}| + |\{j\in [t_1]: r_j \in d^+(I_2)\}|\\
&\leq (d^+(I_0) + d^+(I_1)+ d^+(I_2)) t_1 + \log(n)^2(\sqrt{d^+(I_0)t_1}+ \sqrt{d^+(I_1)t_1} + \sqrt{d^+(I_2)t_1})\\
&\leq (|I_0| + |I_1| + |I_2| + 6/n) t_1 + \log(n)^2(\sqrt{|I_0|t_1 + 2/n}+ \sqrt{|I_1|t_1 + 2/n} + \sqrt{|I_2|t_1 + 2/n})\\
&\leq (|I_0| + |I_1| + |I_2|)t_1 + \tilde{O}(\sqrt{n})\\
& = \frac{9}{2}|I_1| t_1 + \tilde{O}(\sqrt{n}),
\end{align*}
where the first inequality is since $(0,y_2] = I_0\cup I_1\cup I_2$, and the equality is since $|I_1| = \frac{n^{-1/5}}{2}$ and $|I_0\cup I_1\cup I_2| = (3/2)^2n^{-1/5}$.
Hence, by combining the two previous inequalities, we obtain
\begin{equation}
\label{eq:lbt1}
    |\{j\in [t_1]\;| s_{\H^m}(r_j)\in I_1\}| \leq \frac{9}{2}|I_1| t_1 + \tilde{O}(\sqrt{n}).
\end{equation}

Combining (\ref{eq:upt1}) and (\ref{eq:lbt1}), and reorganizing the terms, we get
\[
t_1\geq \frac{2\tilde{n}}{9} - \tilde{O}(\sqrt{n}/|I_1|) =\frac{2(n - \excess/(1-n^{-1/5}))}{9} - \tilde{O}(\sqrt{n}/n^{-1/5})=  \frac{2n}{9} - \tilde{O}(n^{7/10}). 
\]
Hence, since we chose $c_1< 2/9$, and since we assumed $n$ sufficiently large, we have $t_1> c_1n$.
\end{proof}

\lemmto*
\begin{proof}
Note that by definition of $t_1$ and since we assumed $c_1n<t_1$, we have that for all $j\in [c_1n]$, $S_{j}\cap I_1\neq \emptyset$. Hence, by the greediness of $\H^m$ for all requests $r> y_0$, if $r_j> y_1$, we have $s_{\H^m}(r_j)> y_0$ (since $r_j$ is closer to any $s\in (y_0,y_1]\cap S_{j-1}$ than to the servers at location $0$). We thus get
\begin{equation}
\label{eq:c*m}
   |\{j\in [c_1n]\;| s_{\H^m}(r_j)=0\}|\leq |\{j\in [c_1n]\;| s_{\H^m}(r_j)\in [0,y_0)\}|\leq |\{j\in [c_1n]: r_j \leq y_1\}| . 
\end{equation}

Now, note that we have $c_1n|d^+(I_j)| = \Omega(1)$ for $j=0,1$. Hence, by applying the second regularity condition with $t=0$, $t'= c_1n$, and $[d,d'] = d^+(I_0), d^+(I_1)$, respectively, we get
\begin{align*}
|\{j\in [c_1n]: r_j \leq y_1\}|&\leq |\{j\in [c_1n]: r_j \in d^+(I_0)\}| + |\{j\in [c_1n]: r_j \in d^+(I_1)\}|\\
&\leq (d^+(I_0) + d^+(I_1)) c_1n + \log(n)^2(\sqrt{d^+(I_0)c_1n}+ \sqrt{d^+(I_1)c_1n} )\\
&\leq (|I_0| + |I_1|)c_1n + \tilde{O}(\sqrt{n})\\
& = c_1(3/2)n^{-1/5}\cdot n + \tilde{O}(\sqrt{n}),\\
&< n^{4/5}\\
&\leq |S_0\cap \{0\}|,
\end{align*}
where the fourth inequality is since we set $c_1<2/3$ and since we assumed $n$ large enough and the last one by definition of the instance. Hence, combining this with (\ref{eq:c*m}), we get 
$|S_{c_1n}\cap\{0\}| = |S_0\cap \{0\}| - |\{j\in [c_1n]\;| s_{\H^m}(r_j)=0\}|>0$.
By definition of $t_{\{0\}}$, we deduce that $t_{\{0\}}> c_1n$.
\end{proof}

\subsection{Missing analysis from Section \ref{subsec:lower_bounc_cost}}

In the following, we write  $\mathcal{N}(r_t) = \{\max\{z\in S_{t-1}: z\leq r_t\}, \min\{z\in S_{t-1}: z\geq r_t\}\}$ and $\mathcal{N}(r_t)' = \{\max\{z\in S_{t-1}': z\leq r_t\}, \min\{z\in S_{t-1}': z\geq r_t\}\}$  to denote the servers in $S_t$ and $S_t'$ which
are either closest on the left or closest on the right to $r_t$. We also write $s(r_t)$ and $s'(r_t)$ to denote the servers to which $r_t$ is matched by $\H^m$ and $\H^{m-1}$, respectively.

\lemconfigSS*

\begin{proof}
It is immediate that $S_t = S_t'$ for all $t\in \{0,\ldots, m-1\}$ since $\H^m$ and $\H^{m-1}$ make the same matching decisions until time $m-1$. 

Next, we show that either $S_m =  S_m'$, or $\{s>0|s\in S_m\} \neq \emptyset$ and $ S_m' = S_m\cup \{0\}\setminus\{s_{m,1}\}$. We consider different cases depending on the location of request $r_m$.
\begin{itemize}
    \item \textbf{Case 1: $r_m\in (n^{-1/5},1]$ or ($r_m\in [0, n^{-1/5}]$ and $S_{m-1}\cap \{0\} = \emptyset$).} In this case, both $\H^m$ and $\H^{m-1}$ match $r_m$ greedily. Since we also have $S_{m-1}= S_{m-1}'$, we get $S_m = S_m'$.
    \item \textbf{Case 2: $r_m\in [0, n^{-1/5}]$ and $S_{m-1}\cap \{0\} \neq \emptyset$.} In this case, $\H^{m}$ matches $r_m$ to $0$, i.e. $s(r_m) = 0$, while $\H^{m-1}$ matches $r_m$ greedily. Note that $S_{m-1}'\cap (0,n^{-1/5}]\subseteq S_0'\cap (0,n^{-1/5}]=\emptyset$, hence $s_{m-1,1}' = \min\{s>0: s\in S_{m-1'}\}\geq n^{-1/5}$ and we thus have $\mathcal{N}'(r_m) \subseteq \{0, s_{m-1,1}'\}$. Since $\H^{m-1}$ matches $r_m$ greedily, we get that $s(r_m)' \in \{0, s_{m-1,1}'\}$. 
    
    We now consider two cases: \\
    (1) $s'(r_m) = 0$. In this case, we have $s'(r_m) = s(r_m)= 0$, hence $S_m = S_m'$. 
    
    (2) $s'(r_m) = s_{m-1,1}'$. In this case, we have $s(r_m)= 0$ and $ s'(r_m) = s_{m-1,1}' =  \min\{s>0: s\in S_{m-1}'\} = \min\{s>0: s\in S_{m-1}\} = \min\{s>0: s\in S_{m}\cup\{0\}\} = \min\{s>0: s\in S_{m}\} = s_{m,1}$.
    Hence $S_{m}' = S_{m-1}'\setminus\{s'(r_m)\}= S_{m-1}'\setminus\{s_{m,1}\} =  S_{m-1}\setminus\{s_{m,1}\} = S_{m}\cup\{0\}\setminus\{s_{m,1}\}$.
\end{itemize}

Hence we either have that $S_m = S_m'$ or that $S_m' = S_{m}\cup\{0\}\setminus\{s_{m,1}\}$. Now, we show by induction on $t$ that for all $t\in \{m, \ldots, n\}$,  we either have that $S_t = S_t'$ or that $S_t' = S_{t}\cup\{0\}\setminus\{s_{t,1}\}$.

Fix $t\in \{m, \ldots, n-1\}$. If $S_t = S_t'$, it is immediate that $S_{t+1} = S_{t+1}'$ and we are done. We now assume that  $S_t' = S_{t}\cup\{0\}\setminus\{s_{t,1}\}$.
We thus have that $S_t' = S_{t}\cup\{g_t^L\}\setminus\{g_t^R\}$ with $g_t^L = 0, g_t^R = s_{t,1}$. To get the values of $S_{t+1}, S_{t+1}'$, we apply the third part of Lemma \ref{lem:structural_hybrid}, noting that we have here $g_t^L = 0, g_t^R = s_{t,1},s_t^L = 0,  s_t^R = s_{t,2}, d_t^L = |g_t^L-s_t^L| = 0, d_t^R = |s_t^R- g_t^R| = s_{t_2} - s_{t_1}$. We enumerate below all possible values of $S_{t+1},S_{t+1}'$ by reporting the values given in Tables \ref{tab:delta_general}, \ref{tab:delta_special} and \ref{tab:delta_special2} (note that the roles of $S_t$ and $S_t'$ are reversed since $S_t' = S_{t}\cup\{g_t^L\}\setminus\{g_t^R\}$ here instead of $S_t = S_{t}'\cup\{g_t^L\}\setminus\{g_t^R\}$ as in the statement of Lemma \ref{lem:structural_hybrid}).
\begin{itemize}
    \item \textbf{Case 1: $s_{t,2}\neq \emptyset$ and $S_t\cap \{0\}\neq \emptyset$.} In this case, the values of $S_{t+1},S_{t+1}'$ are obtained by using Table \ref{tab:delta_general}. There are three possible cases: (1) $S_{t+1} = S_{t+1}'$ (Column 4 of Table \ref{tab:delta_general}) (2) $S_{t+1}' = S_{t+1}\cup\{0\}\setminus\{s_{t,1}\}$  and $s_{t,1}\in S_{t+1}$ (Column 2,3,6,7) (3) $S_{t+1}' = S_{t+1}\cup\{0\}\setminus\{s_{t,2}\}$ and $s_{t,1}\notin S_{t+1}$ (Column 5). 
    \item \textbf{Case 2: $S_t\cap \{0\}= \emptyset$ and $s_{t,2}\neq \emptyset$.} In this case, the values of $S_{t+1},S_{t+1}'$ are obtained by using Table \ref{tab:delta_special}. There are three possible cases: (1) $S_{t+1} = S_{t+1}'$ (Column 2) (2) $S_{t+1}' = S_{t+1}\cup\{0\}\setminus\{s_{t,1}\}$  and $s_{t,1}\in S_{t+1}$ (Column 4,5) (3) $S_{t+1}' = S_{t+1}\cup\{0\}\setminus\{s_{t,2}\}$ and $s_{t,1}\notin S_{t+1}$ (Column 3).
     \item \textbf{Case 3: $s_{t,2}=\emptyset$ and $S_t\cap \{0\}\neq\emptyset$.} In this case, the values of $S_{t+1},S_{t+1}'$ are obtained by using Table \ref{tab:delta_special2}. There are two possible cases: (1) $S_{t+1} = S_{t+1}'$ (Column 5) (2) $S_{t+1}' = S_{t+1}\cup\{0\}\setminus\{s_{t,1}\}$  and $s_{t,1}\in S_{t+1}$ (Column 2,3,4).
       \item \textbf{Case 4: $s_{t,2}= \emptyset$ and $S_t\cap \{0\}= \emptyset$.} From Lemma \ref{lem:structural_hybrid}, we get $S_{t+1} = S_{t+1}'$.
\end{itemize}

In all cases, we get that either (1) $S_{t+1} = S_{t+1}'$, (2) $S_{t+1}' = S_{t+1}\cup\{0\}\setminus\{s_{t,1}\}$ and $s_{t,1}\in S_{t+1}$ or (3) $S_{t+1}' = S_{t+1}\cup\{0\}\setminus\{s_{t,2}\}$ and $s_{t,1}\notin S_{t+1}$.  If case (2) holds, and since $s_{t,1}\in S_{t+1}$, note that $s_{t+1,1} = \min\{s\in S_{t+1}:s>0\} = s_{t,1}$, and if case (3) holds, since $s_{t,1}\notin S_{t+1}$, note that $s_{t+1,1} = \min\{s\in S_{t+1}:s>0\} = \min\{s\in S_{t}\setminus\{s_{t,1}\}:s>0\} = s_{t,2}$. In all cases, we have that either $S_{t+1} = S_{t+1}'$, or $S_{t+1}' = S_{t+1}\cup\{0\}\setminus\{s_{t+1,1}\}$, which concludes the inductive case and the proof. \end{proof}

\lemMCdef*

\begin{proof} In the following, we consider a fixed $t\in \{m,\ldots, n-1\}$. We start by the proof of Point 1.

\noindent\textbf{Proof of Point 1.}
By definition of $\delta$, if $\delta_t = 0$, then $S_t=S_t'$. Since both $\H^m$ and $\H^{m-1}$ match all requests $r_{t+1},\ldots, r_n$ greedily, it is immediate that $S_j = S_j'$ for all $j\geq t$, which also implies that $\delta_j = 0$ and $\Delta\text{cost}_{j+1} = 0$ for all $j\geq t$.

We now show points 2,3,4,5. First, note that by Lemma \ref{lem:configSS'}, we have that either $S_t = S_t'$ or $S_t' = S_{t}\cup\{0\}\setminus\{s_{t,1}\}$. Since all properties follow immediately when  $S_t = S_t'$, we assume in the following that $S_t' = S_{t}\cup\{0\}\setminus\{s_{t,1}\}$.

\noindent\textbf{Proof of Point 2.}
Assume that $S_{t}\cap \{0\}\neq \emptyset$. Then, $S_t' \cap [0,1] = (S_t\cup\{0\}\setminus\{s_{t,1}\})\cap [0,1]\subseteq S_t\cap [0,1]$, hence, for any value of $r_{t+1}\in [0,1]$, we have, by definition of the process:
\begin{align*}
    \text{cost}_{t+1}(\H^m) = |r_{t+1}- s(r_{t+1})| &=  \min_{s\in S_t\cap [0,1]} |r_{t+1}-s| \\
    &\leq \min_{s\in S_t'\cap [0,1]} |r_{t+1}-s| 
    = |r_{t+1}- s'(r_{t+1})| = \text{cost}_{t+1}(\H^{m-1}).
\end{align*}
Thus, $\Delta\text{cost}_{t+1} = \text{cost}_t(\H^{m-1}) - \text{cost}_t(\H^{m})\geq 0$.

\begin{figure}
    \centering
    \includegraphics[scale=0.38]{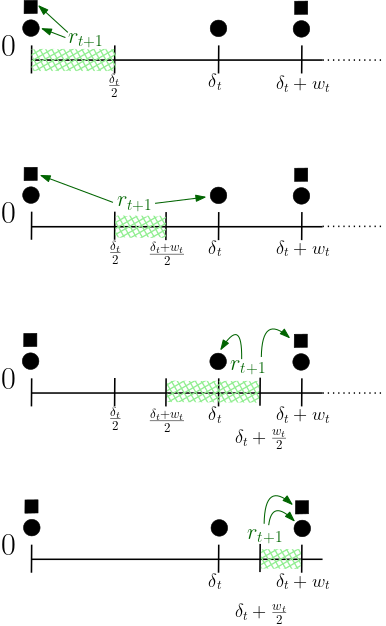}
    \caption{Illustration of the different cases in Point 3 of Lemma \ref{lem:MC_def}.}
    \label{fig:process_lb}
\end{figure}

\noindent\textbf{Proof of point 3.} In the remainder of this paragraph, we condition on the variables $(\delta_t, S_t)$ and we assume that $S_{t}\cap \{0\}\neq \emptyset$, $\delta_t\neq 0$ and $|S_t\cap (\delta_t,1]|\geq 1$.

To get the values of $(\delta_{t+1}, S_{t+1})$ depending on the location of $r_{t+1}$, we apply the third point of Lemma \ref{lem:structural_hybrid}, by noting that we have in this case $g_t^L = 0, g_t^R = s_{t,1},s_t^L = 0,  s_t^R = s_{t,2}, d_t^L = |g_t^L-s_t^L| = 0, d_t^R = |s_t^R- g_t^R| = s_{t_2} - s_{t_1} = w_t$. The values given in Table \ref{tab:MC_gamma_S} are thus directly reported from Table \ref{tab:delta_general} (see Figure \ref{fig:process_lb} for an illustration of the different cases).

Next, we give a lower bound on the expected value of  $\Delta\text{cost}_{t+1}$ depending on $r_{t+1}$. Note that, since $S_t\cap \{0\}\neq \emptyset$, we already have that $\Delta\text{cost}_{t+1}\geq 0$ from Point 2 and we can fill the corresponding values in Table \ref{tab:MC_gamma_S}. We thus only need to refine the lower bound on $\Delta\text{cost}_{t+1}$ in the case $r_{t+1} \in [\tfrac{\delta_t+w_t}{2}, \delta_t +\tfrac{w_t}{2}]$ and $w_t\leq \delta_t$. To ease the exposition, we let $\mathcal{E}_t$ be the set of all $(\delta, S)$ that  satisfy the assumptions of the third point of the lemma (i.e., $\mathcal{E}_t = \{(\delta, S)\in [0,1]\times [0,1]^{n-t}: S\cap \{0\}\neq \emptyset$, $\delta\neq 0$ and $|S\cap (\delta,1]|\geq 1\}$). Since $s'(r_{t+1}) = \delta_t + w_t$, we have
        \begin{align*}
             &\E[\text{cost}_{t+1}' \;|(\delta_t, S_t),(\delta_t, S_t)\in \mathcal{E}_t, w_t\leq \delta_t, r_{t+1} \in [\tfrac{\delta_t+w_t}{2}, \delta_t +\tfrac{w_t}{2}]] \\
             &= \E[|(\delta_t+w_t) - r_{t+1}|\;|(\delta_t, S_t),(\delta_t, S_t)\in \mathcal{E}_t, w_t\leq \delta_t, r_{t+1} \in [\tfrac{\delta_t+w_t}{2}, \delta_t +\tfrac{w_t}{2}]]\\
             &=  \E[(\delta_t+w_t) - r_{t+1}|(\delta_t, S_t),(\delta_t, S_t)\in \mathcal{E}_t,w_t\leq \delta_t, r_{t+1} \in [\tfrac{\delta_t+w_t}{2}, \delta_t +\tfrac{w_t}{2}]]\\
             &= \frac{w_t}{2} +  \E[(\delta_t+\tfrac{w_t}{2}) - r_{t+1}|(\delta_t, S_t),(\delta_t, S_t)\in \mathcal{E}_t,w_t\leq \delta_t, r_{t+1} \in [\tfrac{\delta_t+w_t}{2}, \delta_t +\tfrac{w_t}{2}]]\\
            &= \frac{w_t}{2} + \frac{(\delta_t +\tfrac{w_t}{2} - \tfrac{\delta_t+w_t}{2})}{2}\\
            &= \frac{w_t}{2} + \frac{\delta_t}{4},
        \end{align*}

        and since since $s(r_{t+1}) = \delta_t$, we have
        \begin{align*}
             &\E[\text{cost}_{t+1} \;|(\delta_t, S_t),(\delta_t, S_t)\in \mathcal{E}_t, w_t\leq \delta_t, \;r_{t+1} \in [\tfrac{\delta_t+w_t}{2}, \delta_t +\tfrac{w_t}{2}]]\\
             &=  \E[|\delta_t- r_{t+1}|\;|(\delta_t, S_t),(\delta_t, S_t)\in \mathcal{E}_t,w_t\leq \delta_t,  r_{t+1} \in [\tfrac{\delta_t+w_t}{2}, \delta_t +\tfrac{w_t}{2}]]\\
            &\leq \frac{(\delta_t +\tfrac{w_t}{2}) -\tfrac{\delta_t+w_t}{2}}{2}\\
             &= \frac{\delta_t}{4}, 
        \end{align*}
        where the inequality is since $\delta_t \in [\tfrac{\delta_t+w_t}{2}, \delta_t +\tfrac{w_t}{2}]$ when $w_t\leq \delta_t$.
        
        Hence,
        \[\E[\Delta\text{cost}_{t+1}| (\delta_t, S_t),(\delta_t, S_t)\in \mathcal{E}_t, w_t\leq \delta_t, r_{t+1} \in [\tfrac{\delta_t+w_t}{2}, \delta_t +\tfrac{w_t}{2}]]\geq \left((\frac{w_t}{2} + \frac{\delta_t}{4}) - \frac{\delta_t}{4}\right) \geq \frac{w_t}{2}.\]
        
\noindent\textbf{Proof of point 4.} By assumption, we have $\delta_t = \min\{s>0: s\in S_t\}\in S_t$. Now, assume that $s(r_{t+1}) \neq \delta_t$. Then, whatever the value of  $s'(r_{t+1})$, we have that $\delta_t \notin (S_t\cup \{0\}\setminus\{\delta_t\})\setminus \{s'(r_{t+1})\}$, whereas $\delta_t \in S_t\setminus \{s(r_{t+1})\}$. Thus, $S_{t+1}' = (S_t\cup \{0\}\setminus\{\delta_t\})\setminus \{s'(r_{t+1})\}\neq S_t\setminus \{s(r_{t+1})\} = S_{t+1}$, and by definition of $\delta$, we get $\delta_{t+1} = \min \{s>0: s\in S_{t+1}\} = \min\{s>0: s\in S_t\setminus\{s(r_{t+1})\}\} = \delta_t$. By contraposition, if $\delta_{t+1}\neq \delta_t$, then we must have $s(r_{t+1}) = \delta_t$, and we thus get $S_{t+1} = S_t\setminus\{s(r_{t+1})\} = S_t\setminus\{\delta_t\}$.

\noindent\textbf{Proof of point 5.}  We condition on the variables $(\delta_t, S_t)$ and assume that $S_t\cap \{0\}=\emptyset$ and $\delta_t\neq 0$. We first show that $\delta_{t+1} \neq 0$ if and only if $s'(r_{t+1})\neq 0$. 
\begin{itemize}
    \item $\Leftarrow$: Assume that $s'(r_{t+1})\neq 0$. Since $0\in (S_t\cup\{0\}\setminus\{s_{t,1}\})\setminus\{s'(r_{t+1})\}$, and $0\notin S_t$ by assumption, we have that whatever the value of  $s(r_{t+1})$, $S_{t+1}'= (S_t\cup\{0\}\setminus\{s_{t,1}\})\setminus\{s'(r_{t+1})\}\neq S_t\setminus\{s(r_{t+1})\} = S_{t+1}$. Thus, by construction of $\delta$, we get $\delta_{t+1} \neq 0$. 
    \item $\Rightarrow$: Assume, by contrapositive, that $s'(r_{t+1}) = 0$. Since by assumption, $S_t \cap\{0\} = \emptyset$, we have that $S_t = \{s_{t,1}, s_{t,2}\}\cup (S_t\cap (s_{t,2}, 1])$. Now, since $s'(r_{t+1}) = \argmin_{s\in S_t\cup\{0\}\setminus\{s_{t,1}\}}|s-r_{t+1}|$, we have that for all $s\in S_t\cap [s_{t,2},1]$, $|r_{t+1}-0|\leq |r_{t+1}-s|$; thus, if $r_{t+1}\geq s_{t,1}$, we have that $|r_{t+1}-s_{t,1}|\leq |r_{t+1}-0|\leq |r_{t+1}-s|$, and if $r_{t+1}\leq s_{t,1}$, it is immediate that for all $s\in S_t\cap [s_{t,2},1]$, $|s_{t,1} - r_{t+1}|\leq |s- r_{t+1}|$. Hence, we get that $s(r_{t+1}) = \argmin_{s\in S_t} |r_{t+1} - s| = s_{t,1}$, which immediately implies that $S_{t+1} = S_{t+1}'$, from which we deduce $\delta_{t+1}=0$.
\end{itemize}

We now show that $\Delta\text{cost}_{t+1}\geq 0$ when $s'(r_{t+1})\neq 0$. Since $s'(r_{t+1})\in S_t\cup\{0\}\setminus\{s_{t,1}\}$ and $s'(r_{t+1})\neq 0$, we have $s'(r_{t+1})\in S_t$. Thus,  $|s(r_{t+1})-r_{t+1}| = \min_{s\in S_t}|s-r_{t+1}|\leq |s'(r_{t+1})-r_{t+1}|$, and we deduce $\Delta\text{cost}_{t+1} = |s'(r_{t+1})-r_{t+1}| - |s(r_{t+1})-r_{t+1}|\geq 0$.

Hence, we have shown that $\delta_{t+1} \neq 0$ if and only if $s'(r_{t+1})\neq 0$, and that if $s'(r_{t+1})\neq 0$, then $\Delta\text{cost}_{t+1}\geq 0$. Using that it is always the case that $\Delta\text{cost}_{t+1}\in [-1,1]$, we get:
\begin{align*}
   & \E[\mathbbm{1}_{S_t\cap \{0\}=\emptyset, \delta_t\neq 0}\cdot\Delta\text{cost}_{t+1}|(\delta_t, S_t)]\\ &= \mathbbm{1}_{S_t\cap \{0\}=\emptyset, \delta_t\neq 0}\cdot\Big(\E[\Delta\text{cost}_{t+1}|(\delta_t, S_t), \delta_{t+1}\neq 0]\P(\delta_{t+1}\neq 0|(\delta_t, S_t)) \\
   &\qquad\qquad\qquad\qquad\qquad\qquad\qquad\qquad+ \E[\Delta\text{cost}_{t+1}|(\delta_t, S_t), \delta_{t+1}= 0]\P(\delta_{t+1}= 0|(\delta_t, S_t))\Big)\\
   &\geq 0  - \mathbbm{1}_{S_t\cap \{0\}=\emptyset, \delta_t\neq 0}\cdot\P(\delta_{t+1}=0|(\delta_t, S_t)),
\end{align*}
which concludes the proof of the fifth point and the proof of Lemma \ref{lem:MC_def}.
\end{proof}

\lemsimpleversionprocess*

\begin{proof} For all $j\in \{m,\ldots, n-1\}$, we  define the auxiliary stopping times: $t^{j,d} = \min \{t\geq j: \delta_t = 0\}$, $t^j_{\{0\}} = \min \{t\geq j: S_t\cap \{0\} = \emptyset\}$ and $t^j_{(0,y]} = \min \{t\geq j: S_t\cap \{(0,y]\} = \emptyset\}$. To ease the presentation, we write $\underline{t^j_y}$ and $\underline{t^{j,d}}$ instead of $\min(t^j_{\{0\}}, t^j_{[0,y)})$ and $\min(t^j_{\{0\}}, t^{j,d})$.

We now show by downward induction on $j$ that for any $j\in \{m,\ldots, n\}$, any pair $(x,S)$ with $x\in [0,1]$ and with $S$ a set of $n-j$ arbitrary servers in $[0,1]$ such that either $x= 0$ or  $x = \min\{s\in S: s>0\}$, and any $y\in (x,1]$, we have: 
\[
\P_{R}\Big(\underline{t^j_y}\leq \underline{t^{j,d}}|\delta_j=x,S_j =S\Big)\geq \frac{x}{y}.
\]

We first show the base case, which is for $j=n$. The only valid pair of $(x,S)$ is $(0, \emptyset)$, and it is immediate that for any $y\in (0,1]$, we have
\[
\P_{R}\Big(\underline{t^{n}_y}\leq \underline{t^{n,d}}|\delta_{n}=0,S_{n} =\emptyset\Big) = 1\geq  \frac{x}{y}.
\]

Next, let $j\in \{m,\ldots, n-1\}$, and assume that for any pair $(x,S)$ with $x\in [0,1]$ and with $S$ a set of $n-(j+1)$ arbitrary servers in $[0,1]$ such that either $x= 0$ or  $x = \min\{s\in S: s>0\}$, and any $y\in (x,1]$, we have 
\[
\P_{R}\Big(\underline{t^{j+1}_y}\leq \underline{t^{j+1,d}}|\delta_{j+1}=x,S_{j+1} =S\Big)\geq \frac{x}{y}.
\]

Now, consider some arbitrary pair $(x,S)$ with $x\in [0,1]$ and with $S$ a set of $n-j$ arbitrary servers in $[0,1]$ such that either $x= 0$ or  $x = \min\{s\in S: s>0\}$, and some arbitrary $y\in (x,1]$, and assume that $\delta_{j} = x$ and $S_{j} = S$. 

We first consider the case where $x>0$, $|(x,y]\cap S|\geq 1$ and $S\cap \{0\}\neq \emptyset$.

First, note that since $S_j \cap \{0\}\neq \emptyset$, $\delta_j =x >0$ and $S_j\cap (0,y]\supseteq \{x\} \neq \emptyset$, we have that $\underline{t^{j}_y}, \underline{t^{j,d}}\geq j+1$. We deduce the following proposition: for any $r\in [0,1]$, letting $\chi(x,S,r)$ and $T(x,S,r)$ be the value of $\delta_{j+1}$ and $S_{j+1}$ assuming that $\delta_j = x, S_j = S$ and $r_{j+1} = r$, we have
\begin{align}
    &\P_{R}\Big(\underline{t^{j}_y}\leq \underline{t^{j,d}}|\delta_j=x,S_j =S, r_{j+1} = r\Big)\nonumber\\
    & = \P_{R}\Big(\underline{t^{j}_y}\leq \underline{t^{j,d}}|\delta_j=x,S_j =S, \delta_{j+1} = \chi(x,S,r), S_{j+1} = T(x,S,r), r_{j+1} = r\Big)\nonumber\\
    & = \P_{R}\Big(\underline{t^{j+1}_y}\leq \underline{t^{j+1,d}}| \delta_j=x,S_j =S, \delta_{j+1} = \chi(x,S,r), S_{j+1} = T(x,S,r), r_{j+1} = r\Big)\nonumber\\
     & = \P_{R}\Big(\underline{t^{j+1}_y}\leq \underline{t^{j+1,d}}| \delta_{j+1} = \chi(x,S,r), S_{j+1} = T(x,S,r)\Big)\nonumber\\
    &\geq \frac{\chi(x,S,r)}{y},\label{eq:markov_2}
\end{align}

where the second equality is since $\underline{t^{j}_y}, \underline{t^{j,d}}\geq j+1$, the third equality is since conditioned on $S_{j+1}, \delta_{j+1}$, we have that $\{(\delta_t, S_t)\}_{t\in \{j+1, \ldots, n\}}$ is independent on $r_{j+1}, S_j, \delta_j$, and the  inequality is by the induction hypothesis.

We now enumerate five different cases depending on request $r_{j+1}$. Since we assumed that $|(x,y]\cap S|\geq 1$, $S\cap \{0\}\neq \emptyset$ and $\delta_j =x\neq 0$, we have by Lemma \ref{lem:MC_def} that the values of $\chi(x,S,r_{j+1})$  are the one given in Table \ref{tab:MC_gamma_S}, with $s^R = \min\{s\in S: s> x\}$ and $w = s^R-x$.

\begin{itemize}
    \item Case 1: $r_{j+1}\in [0,\tfrac{x}{2}]$. Then,  from Table \ref{tab:MC_gamma_S}, we  have  $\chi(x, S, r_{j+1}) = \delta_j = x$. Thus, by using (\ref{eq:markov_2}), we get
    \[ \P_{R}\Big(\underline{t^{j}_y}\leq \underline{t^{j,d}}|\delta_j=x,S_j =S, r_{j+1}  \in [0,\tfrac{x}{2}]\Big)  \geq \frac{x}{y}.\]

    \item Case 2: $r_{j+1}\in [\tfrac{x}{2}, \tfrac{x+w}{2}]$. Trivially, we have
    \[ \P_{R}\Big(\underline{t^{j}_y}\leq \underline{t^{j,d}}|\delta_j=x,S_j =S, r_{j+1} \in [\tfrac{x}{2}, \tfrac{x+w}{2}]\Big)  \geq 0.\]

    \item Case 3: $r_{j+1}\in  [\tfrac{x+w}{2},x+\tfrac{w}{2}]$. Then,  from Table \ref{tab:MC_gamma_S}, we  have  $\chi(x, S, r_{j+1}) = s^R$. 
    Thus, by using (\ref{eq:markov_2}), we get
    \[ \P_{R}\Big(\underline{t^{j}_y}\leq \underline{t^{j,d}}|\delta_j=x,S_j =S, r_{j+1} \in [\tfrac{x+w}{2},x+\tfrac{w}{2}]\Big)  \geq \frac{s^R}{y} = \frac{x+w}{2}.\]

    \item Case 4: $r_{j+1}\in  [x+\tfrac{w}{2},x+w]$. Then, from Table \ref{tab:MC_gamma_S}, we have  $\chi(x, S, r_{j+1}) = x$. 
    
    Thus, by using (\ref{eq:markov_2}), we get
    \[ \P_{R}\Big(\underline{t^{j}_y}\leq \underline{t^{j,d}}|\delta_j=x,S_j =S, r_{j+1} \in [x+\tfrac{w}{2},x+w]\Big)  \geq \frac{x}{y}.\]

    \item Case 5: $r_{j+1}\in  [x+w,1]$. Then, from Table \ref{tab:MC_gamma_S}, we have that $\chi(x, S, r_{j+1}) = x$. Thus, by using (\ref{eq:markov_2}), we get
    \[ \P_{R}\Big(\underline{t^{j}_y}\leq \underline{t^{j,d}}|\delta_j=x,S_j =S, r_{j+1}\in  [x+w,1]\Big)  \geq \frac{x}{y}.\]
    
\end{itemize}

By combining the five cases above, we get
\begin{align*}
    &\P_{R}\Big(\underline{t^{j}_y}\leq \underline{t^{j,d}}|\delta_j=x,S_j =S\Big) \\
    \hspace{0.1cm} &\geq \P(r_{j+1}\in [0,\tfrac{x}{2}])\cdot\frac{x}{y} 
    + 0 
    + \P(r_{j+1}\in  [\tfrac{x+w}{2},x+\tfrac{w}{2}])\cdot\frac{x+w}{y}\\
    \hspace{0.1cm} &+ \P(r_{j+1}\in  [x+\tfrac{w}{2},x+w])\cdot\frac{x}{y}
    + \P(r_{j+1}\in  [x+w,1])\cdot\frac{x}{y}\\
    \hspace{0.1cm} &= \frac{x}{2}\cdot\frac{x}{y} +  \frac{x}{2}\cdot\frac{x+w}{y} + \frac{w}{2}\cdot\frac{x}{y} + (1- (x+w))\cdot\frac{x}{y}\\
    \hspace{0.1cm} &=  \frac{x}{y}\cdot\Big(\frac{x}{2} +  \frac{x+w}{2} + \frac{w}{2}+ (1- (x+w))\Big)\\
    \hspace{0.1cm} & = \frac{x}{y}.
\end{align*}

It remains to show the inductive case when either $x=0$, $|(x,y]\cap S|=0$, or $S\cap \{0\}= \emptyset$. Note that if $x=0$, it is immediate that 
\[
\P_{R}\Big(\underline{t^{j}_y}\leq \underline{t^{j,d}}|\delta_j=x,S_j =S\Big)\geq 0 = \frac{x}{y},
\]
and if $S\cap \{0\}=\emptyset$, then $\underline{t^{j}_y} = \underline{t^{j,d} }= j$ and we have 
\[
\P_{R}\Big(\underline{t^{j}_y}\leq \underline{t^{j,d}}|\delta_j=x,S_j =S\Big)=1 \geq \frac{x}{y}.
\]

Finally, if $|(x,y]\cap S|=0$ and $x>0$, then 
\[S_j\cap (0,y] = (S\cap (0,x])  \cup (S\cap (x,y]) = \{x\} \cup \emptyset  = \{\delta_j\},
\]
where the second equality is since $\min\{s>0|s\in S\} = x$ and the assumption that $|S\cap (x,y]| = 0$.  Now, we have
\begin{align*}
   \underline{t^{j,d}} &:=\min\{t\geq j: \delta_{t} = 0 \text{ or }  S_t\cap \{0\} = \emptyset\} \\
  &\geq \min\{t\geq j: \delta_{t} \neq \delta_j \text{ or }  S_t\cap \{0\} = \emptyset\}\\
    &= \min \{t\geq j: S_t = S_{t-1}\setminus\{\delta_j\} \text{ or }  S_t\cap \{0\} = \emptyset\}\\
   &= \min \{t\geq j:  S_t\cap (0,y] = \emptyset \text{ or }  S_t\cap \{0\} = \emptyset\}\\
   &=: \underline{t^{j}_y}, 
\end{align*}
where the inequality is since $ \delta_j  = x >0$, the second equality is from the fourth point of Lemma \ref{lem:MC_def}, and the third equality is since $S_j\cap (0,y] = \{\delta _j\}$. Hence we also get
\[
\P_{R}\Big(\underline{t^{j}_y}\leq \underline{t^{j,d}}|\delta_j=x,S_j =S\Big)=1 \geq \frac{x}{y}.
\]

Hence, in all possible cases, we have shown that 
\[
\P_{R}\Big(\underline{t^{j}_y}\leq \underline{t^{j,d}}|\delta_j=x,S_j =S\Big)\geq \frac{x}{y},
\]
which concludes the inductive case. We conclude the proof by applying the above inequality with $j = m$.
\end{proof}

\claimEdeltam*

\begin{proof}
\textbf{Point 1.} By Lemma \ref{lem:configSS'}, we have that $S_{m-1} = S_{m-1}'$. Now, if $r_m\in (y_0, 1]$, then both $\H^m$ and $\H^{m-1}$ match $r_m$ greedily. Hence, $S_m = S_m'$, which, by definition of $\delta$, implies that $\delta_m =0$. By contraposition, if $\delta_m>0$, we must have $r_m\in [0,y_0]$. \\

\textbf{Point 2.}
It is immediate that $\delta_m = 0$ when $r_m\in (y_0,1]$.
Now, if $r_m\in [0,y_0]$, then either $S_m\cap \{0\} = \emptyset$ or $S_m\cap \{0\} \neq \emptyset$. In the first case, $r_m$ is matched greedily by both $\H^m$ and $\H^{m-1}$ and we get $S_m = S_m'$ and $\delta_m= 0$. In the second case, we first have, by definition of $\H^m$, that $s(r_m) = 0$. Then, by the greediness of $\H^{m-1}$ for $r_m$, we get $|s(r_m)' - r_m|\leq |r_m-0|$, which implies $s(r_m)'\leq 2r_m$. Hence, $\delta_m = s(r_m') - 0 \leq 2r_m\leq 2y_0 = 2n^{-1/5}$. We conclude that for all $m\in [n]$, $\delta_m\in [0,2n^{-1/5}]$.

\textbf{Point 3.}
Fix $m\leq c_1n$. We first show that if $R$ is regular and $r_m\in [\frac{3}{4}y_0, y_0]$, we have that $\delta_m= s_{m-1,1}\geq n^{-1/5}$.

Assume that the sequence of requests is regular and consider $r_m\in [\frac{3}{4}y_0, y_0]$. We start by showing that $s(r_m)' = s_{m-1,1}$ and $s(r_m) = 0$. Since $m\leq c_1n$, we have  by Lemma \ref{lem:lb1} that $m<t_1= \min\{t\geq 0: S_t\cap I_1 = \emptyset\}$. Hence, $S_{m-1}\cap [n^{-1/5}, (3/2)n^{-1/5}]\neq \emptyset$. Since $S_{m-1}\cap [0,n^{-1/5}] \subseteq S_{0}\cap [0,n^{-1/5}]= \emptyset$, we thus have $s_{m-1,1}' = s_{m-1,1}\in [n^{-1/5}, (3/2)n^{-1/5}]$. Since $r_m\in [\frac{3}{4}y_0, y_0] = [\frac{3}{4}n^{-1/5}, n^{-1/5}]$, we deduce that $\mathcal{N}'(r_m)\subseteq\{0,s_{m-1,1}\}$. Now, note that
\[|s_{m-1,1}-r_m|\leq |(3/2)n^{-1/5} - r_m| \leq |(3/2)n^{-1/5} - (3/4)n^{-1/5}| = (3/4)n^{-1/5} \leq  |r_m-0|.\]
 Since $\H^{m-1}$ matches $r_m$ greedily, we get that $s'(r_m) = s_{m-1,1}$.

 On the other hand, $\H^{m}$ follows $\mathcal{A}$ for matching $r_m$. Note that by Lemma \ref{lem:lb1}, we have that $m<t_{\{0\}}$, hence $S_{m-1}'\cap \{0\} = S_{m-1}\cap \{0\}\neq \emptyset$. Since $r_m\in [0,n^{-1/5}]$, we get by definition of $\A$ that $s(r_m) = 0$.

Since $s(r_m) = 0$ and $s(r_m') = s_{m-1,1}$, we deduce that $S_m = S_{m-1}\setminus\{0\} \neq S_{m-1}\setminus\{s_{m-1,1}\} =S_{m-1}'\setminus\{s_{m-1,1}\} = S_m'$, hence, by definition of $\delta$, we get: $\delta_m = s_{m,1} = s_{m-1,1}\geq n^{-1/5}$. We have thus shown that if $R$ is regular and $r_m\in [\frac{3}{4}y_0, y_0]$, then $\delta_m \geq n^{-1/5}$. As a result,
\begin{align*}
    &\mathbb{E}[\delta_m|r_m \in [0,y_0]]\geq n^{-1/5}\P(\delta_m\geq n^{-1/5}|r_m \in [0,y_0])\\
    &\geq n^{-1/5}\P(r_m\in [\tfrac{3}{4}y_0, y_0], \reg|r_m \in [0,y_0])\\
    &= n^{-1/5}(\P(r_m\in [\tfrac{3}{4}y_0, y_0]|r_m \in [0,y_0]) - \P(r_m\in [\tfrac{3}{4}y_0, y_0], \text{R is not regular}|r_m \in [0,y_0]))\\
    &\geq n^{-1/5}(\P(r_m\in [\tfrac{3}{4}y_0, y_0]|r_m \in [0,y_0]) - \P(\text{R is not regular}|r_m \in [0,y_0]))\\
    &= n^{-1/5}(\P(r_m\in [\tfrac{3}{4}y_0, y_0]|r_m \in [0,y_0]) - \P(r_m \in [0,y_0], \text{R is not regular})/\P(r_m \in [0,y_0]))\\
    &\geq n^{-1/5}(\P(r_m\in [\tfrac{3}{4}y_0, y_0]|r_m \in [0,y_0]) - \P(\text{R is not regular})/\P(r_m \in [0,y_0]))\\
    &= \frac{1}{4}n^{-1/5}-\hp,
\end{align*}
where the last equality holds since $R$ is regular with high probability by Lemma \ref{lem:reg}.
\end{proof}

\lemdepletedbefore*

\begin{proof}
We first assume that the requests sequence is regular, and we show that there is no $t\in [n]$ such that $s_{t,1}<1/2$ and $S_t\cap \{0\} = \emptyset$. 

Assume by contradiction that there is such a $t$. Since $s_{t,1}$ is available at time $t$, we have that for all $i\in [t]$, $s_{t,1}$ is available when request $r_i$ arrives. In addition, recall that $\H^m$ either matches each request to $0$, or matches it greedily, and it matches a request $r$ to $0$ only if $r\leq y_0$. Since by definition of the instance, $s_{t,1}> y_0$, we get that there is no $i\in [t]$ such that $r_i< s_{t,1}$ and $s(r_i)\geq s_{t,1}$ (since $r_i$ is closer to $s_{t,1}$ than any other server $s>s_{t,1}$ and $s_{t,1}$ is available when $r_i$ arrives) and there is no $i\in [t]$ such that $r_i\geq s_{t,1}$ and $s(r_i)<s_{t,1}$. Hence, 
\[|\{i\in [t]: r_i\in [0, s_{t,1}]| = |i\in [t]: s(r_i)\in [0, s_{t,1}]\}|.\]

In addition, since $s_{t,1} = \min\{s>0:s\in S_t\}$ and since we assumed that $S_t\cap \{0\} = \emptyset$, we have $[0,s_{t,1})\cap S_t =S_t\cap \{0\}= \emptyset$, hence all servers in $[0,s_{t,1})$ have been matched to some request before time $t$ and we have $|i\in [t]: s(r_i)\in [0, s_{t,1}]| = |S_0\cap [0,s_{t,1})|$. Let $d^+ = \min\{j/n:j\in [n], j/n>s_{t,1}\}$. We get
\begin{align*}
&\quad \;|\{i\in [n]: r_i\in [0,d^+]\}|\\
&\geq 
|\{i\in [t]: r_i\in [0,d^+]\}|\\
&\geq|\{i\in [t]: r_i\in [0, s_{t,1}]\}|\\
&=|S_0\cap [0,s_{t,1})|\\
&\geq|S_0\cap [0,n^{-1/5}]| + |S_0\cap [n^{-1/5},d^+-1/n)|\\
&\geq [n^{4/5} + \excess] + [(d^+-1/n-n^{-1/5})\tilde{n} - 1]\\
&= [n^{4/5} + \excess] + [(d^+-1/n-n^{-1/5})(n - \excess/(1-n^{-1/5})) - 1]\\
&= d^+n +  \excess(1 - (d^+-1/n-n^{-1/5})/(1-n^{-1/5})) - 2\\
&= d^+n +  \excess(1 - (1/2-n^{-1/5})/(1-n^{-1/5})) - 2\\
&> d^+n +  \log{n}^2\sqrt{d^+n},
\end{align*}
where the fourth inequality is by definition of the instance and by  Lemma~\ref{fact:instance}, and the fourth equality since $d^+ \leq s_{t,1}+1/n\leq 1/2+1/n$ and the last inequality since $d^+\leq 1/2$.
Hence, the second regularity condition from Definition \ref{def:reg} is not satisfied for $t=0,t'=n$ and $d =0, d'=d^+$. Thus, if $R$ is regular, then there is no $t\in [n]$ such that $s_{t,1}<1/2$ and $S_t\cap \{0\} = \emptyset$.

On the way, we deduce the following equation, that will be used in the proof of the second part of the lemma.
\begin{align}\label{eq:12S0}
    &\P(\exists t\in [n]: s_{t,1}<1/2 \text{ and } S_t\cap \{0\} = \emptyset|r_m\in [0,y_0])\nonumber\\
    &\qquad\qquad\leq \P(\text{R is not regular}|r_m\in [0,y_0])
     \leq \P(\text{R is not regular})/\P(r_m\in [0,y_0])
    = \hp,
\end{align}
where the last inequality holds since $R$ is regular with high probability by Lemma \ref{lem:reg}.

Now, we assume that $R$ is regular and we show that for all $i\in [d_1\log(n)]$, $t_{(0,y_i]}\leq t_{\{0\}}$. Note that if $t_{(0,y_i]}> t_{\{0\}}$, then, by definition of $t_{(0,y_i]}$, we have that $S_{t_{\{0\}}}\cap (0,y_i]\neq \emptyset$, which implies that $s_{t_{\{0\}},1}\leq y_i< 1/2$. By definition of $t_{\{0\}}$, we also have  $S_{t_{\{0\}}}\cap\{0\} = \emptyset$. This contradicts the fact that there is no $t\in [n]$ such that $s_{t,1}<1/2$ and $S_t\cap \{0\} = \emptyset$. Hence, we have $t_{(0,y_i]}\leq t_{\{0\}}$, which concludes the proof of the first part of the lemma.

Next, we show that $t^d\leq t_{\{0\}}$ with high probability. First, note that by Lemma \ref{lem:worst_case}, we have $
\P(\max_{t\in [n]}\delta_t\geq 1/2\;|\;\delta_m)\leq 2\delta_m.$
Since by Lemma \ref{claim:E_delta_m}, we have that for all $m\in [n]$, $\delta_m\leq 2n^{-1/5}$, we get
\begin{equation}
\label{eq:maxdelta}
    \P(\max_{t\in [n]}\delta_t\geq 1/2|r_m\in [0,y_0])\leq 4n^{-1/5}.
\end{equation}

Hence, we have
\begin{align*}
    &\P(t^d>t_{\{0\}}|r_m\in [0,y_0]) \\
    &= \P(t^d>t_{\{0\}},\; \delta_{t_{\{0\}}}\geq 1/2|r_m\in [0,y_0]) + \P(t^d>t_{\{0\}},\;
    \delta_{t_{\{0\}}}<1/2|r_m\in [0,y_0])\\
    &\leq\P(\delta_{t_{\{0\}}}\geq 1/2|r_m\in [0,y_0]) +  \P(\delta_{t_{\{0\}}} < 1/2, \delta_{t_{\{0\}}} >0, S_{t_{\{0\}}} \cap \{0\}= \emptyset|r_m\in [0,y_0]) \\
    &\leq\P(\delta_{t_{\{0\}}}\geq 1/2|r_m\in [0,y_0]) +  \P(s_{t_{\{0\}},1}<1/2, S_{t_{\{0\}}} \cap \{0\}= \emptyset|r_m\in [0,y_0]) \\
    &\leq\P(\max_{t\in [n]}\delta_t\geq 1/2|r_m\in [0,y_0]) +  \P(\exists t\in [n]: s_{t,1}<1/2 \text{ and } S_t\cap \{0\} = \emptyset|r_m\in [0,y_0])\\
    &\leq  4n^{-1/5}+ \hp\\
    & = O(n^{-1/5}),
\end{align*}
where the first inequality holds since by definition of $t^d$, if $t^d>t_{\{0\}}$, then $\delta_{t_{\{0\}}} >0$, and since we always have $S_{t_{\{0\}}} \cap \{0\}= \emptyset$ by definition of $t_{\{0\}}$. The second inequality holds since by definition of $\delta$, if $\delta_{t_{\{0\}}}>0$, then $\delta_{t_{\{0\}}} = s_{t_{\{0\}},1}$. The last inequality is by (\ref{eq:maxdelta}) and (\ref{eq:12S0}).

\end{proof}

\lemcostdeltagamma*

\begin{proof} We  analyse the difference of cost between $\H^m$ and $\H^{m-1}$ for all requests $r_{m+1}, \ldots, r_{n}$. We consider in the following paragraphs some time steps $t\geq m$ and we omit to mention this condition throughout the proof.

We start by some preliminary notational considerations. We first recall that $w_t = s_{t,2}-s_{t,1}$. Also, note that for all $t\in \mathbb{N}$, we have $S_{t+1}\subseteq S_t$; thus if $S_t\cap \{0\}= \emptyset$, then $S_{t+1}\cap \{0\}= \emptyset$. Hence, by definition of $t_{\{0\}}$, we first have that $S_t\cap \{0\}\neq \emptyset$ if and only if $t\leq t_{\{0\}}-1$. Hence, $\{t\leq t_{\{0\}}-1\}$ is entirely determined by the value of $S_t$ and

\[
\quad \mathbbm{1}_{S_t\cap \{0\}\neq \emptyset} = \mathbbm{1}_{t\leq t_{\{0\}}-1}.
\]
Now, assume that we also have $t\leq t^w-1$ and  $t\leq t^d-1$. Since $t\leq t^d-1$, we first have $\delta_t \neq 0$, which also implies, by definition of $\delta_t$, that $\delta_t = s_{t,1}$. Since $t\leq t^w-1$, we deduce that $\delta_t = s_{t,1}\geq s_{t,2} - s_{t,1} = w_t$. Finally, we trivially have $|S_t\cap (\delta_t,1]|= |S_t\cap (s_{t,1},1]|\geq 1$. Hence, if $t\leq t^w-1$ and  $t\leq t^d-1$, then we have that $\delta_t\neq 0$, $|S_t\cap (\delta_t,1]| \geq 1$ and  $ w_t\leq \delta_t$. Therefore, we get

\begin{equation}
\label{eq:indicators}
     \mathbbm{1}_{S_t\cap \{0\}\neq \emptyset, \delta_t\neq 0, |S_t\cap (\delta_t,1]|\geq 1, w_t\leq \delta_t} \geq \mathbbm{1}_{t\leq \min(t^w, t^d, t_{\{0\}})-1}.
\end{equation}

\noindent\textbf{Lower bound on  $\E[\Delta\text{cost}_{t+1}]$ in the case where $t\leq t_{\{0\}}-1$.}
For all $t\leq t_{\{0\}}-1$ we have $S_t\cap \{0\}\neq \emptyset$; thus, by the  second point of Lemma \ref{lem:MC_def}, we have that $\Delta\text{cost}_{t+1}\geq 0$. Hence,
\[\E[\mathbbm{1}_{t\leq t_{\{0\}}-1}\cdot\Delta \text{cost}_{t+1}|(\delta_t, S_t)] \geq 0.\]
Assume that we further have that $\delta_t\neq 0$ and $|S_t\cap (\delta_t,1]|\geq 1$ (i.e., the assumptions of the third point of Lemma \ref{lem:MC_def} are satisfied) and that $w_t\leq \delta_t$. Then, we obtain:
\begin{align*}
    &\E[\mathbbm{1}_{t\leq t_{\{0\}}-1, \delta_t\neq 0, |S_t\cap (\delta_t,1]|\geq 1, w_t\leq \delta_t}\cdot\Delta \text{cost}_{t+1}|(\delta_t, S_t)]\\
    &=\E[\mathbbm{1}_{t\leq t_{\{0\}}-1, \delta_t\neq 0, |S_t\cap (\delta_t,1]|\geq 1, w_t\leq \delta_t}\cdot\Delta \text{cost}_{t+1}|(\delta_t, S_t), S_t\cap \{0\}\neq \emptyset, \delta_t\neq 0, |S_t\cap (\delta_t,1]|\geq 1, w_t\leq \delta_t]\\
    &\qquad \cdot \P(S_t\cap \{0\}\neq \emptyset, \delta_t\neq 0, |S_t\cap (\delta_t,1]|\geq 1, w_t\leq \delta_t|(\delta_t, S_t)) + 0\\
    &=\mathbbm{1}_{t\leq t_{\{0\}}-1, \delta_t\neq 0, |S_t\cap (\delta_t,1]|\geq 1, w_t\leq \delta_t}\cdot\E[\Delta \text{cost}_{t+1}|(\delta_t, S_t),S_t\cap \{0\}\neq \emptyset, \delta_t\neq 0, |S_t\cap (\delta_t,1]|\geq 1, w_t\leq \delta_t]\\
    &\geq\mathbbm{1}_{t\leq t_{\{0\}}-1, \delta_t\neq 0, |S_t\cap (\delta_t,1]|\geq 1, w_t\leq \delta_t}\\
    &\qquad\cdot \Big( 0 + \E[ \Delta \text{cost}_{t+1}|(\delta_t, S_t), S_t\cap \{0\}\neq \emptyset, \delta_t\neq 0, |S_t\cap (\delta_t,1]|\geq 1, w_t\leq \delta_t, r_{t+1}\in [\tfrac{\delta_t+w_t}{2}, \delta_t +\tfrac{w_t}{2}]]\\
    &\qquad\qquad\cdot\P(r_{t+1}\in [\tfrac{\delta_t+w_t}{2}, \delta_t +\tfrac{w_t}{2}])\Big)\\
    &= \mathbbm{1}_{t\leq t_{\{0\}}-1, \delta_t\neq 0, |S_t\cap (\delta_t,1]|\geq 1, w_t\leq \delta_t}\cdot\frac{w_t}{2}\cdot\P(r_{t+1}\in [\tfrac{\delta_t+w_t}{2}, \delta_t +\tfrac{w_t}{2}]),
\end{align*}
where the inequality is by inspecting all possible cases given in Table \ref{tab:MC_gamma_S} and since $r_{t+1}$ is independent of $(\delta_t, S_t)$, and the last equality is since $w_t\leq \delta_t$ and by inspecting the corresponding case in Table \ref{tab:MC_gamma_S}.

By combining the two previous inequalities, we get
\begin{align*}
    &\E[\mathbbm{1}_{t\leq t_{\{0\}}-1}\cdot\Delta \text{cost}_{t+1}|(\delta_t, S_t)]\\
    &=  \E[\mathbbm{1}_{t\leq t_{\{0\}}-1, \delta_t\neq 0, |S_t\cap (\delta_t,1]|\geq 1, w_t\leq \delta_t}\cdot\Delta \text{cost}_{t+1}|(\delta_t, S_t)] \\
    &\qquad\qquad\qquad\qquad\qquad+ \E[\mathbbm{1}_{t\leq t_{\{0\}}-1, \delta_t\neq 0, |S_t\cap (\delta_t,1]|\geq 1, w_t\leq \delta_t}\cdot\Delta \text{cost}_{t+1}|(\delta_t, S_t)]\\
    &\geq \mathbbm{1}_{t\leq t_{\{0\}}-1, \delta_t\neq 0, |S_t\cap (\delta_t,1]|\geq 1, w_t\leq \delta_t}\cdot\frac{w_t}{2}\cdot\P(r_{t+1}\in [\tfrac{\delta_t+w_t}{2}, \delta_t +\tfrac{w_t}{2}]) + 0.
\end{align*}

Note that conditioning on $(\delta_t, S_t)$, we have that $\Delta \text{cost}_{t+1}$ is independent of $(\delta_m, S_m)$. Thus, first conditioning on $(\delta_m, S_m)$, then applying the tower law, we get:
\begin{align}
    \E[&\mathbbm{1}_{t\leq t_{\{0\}}-1}\cdot\Delta \text{cost}_{t+1}|(\delta_m, S_m)]\nonumber\\
    & = \E[\E[\mathbbm{1}_{t\leq t_{\{0\}}-1}\cdot\Delta \text{cost}_{t+1}|(\delta_t, S_t), (\delta_m, S_m)]\;|(\delta_m, S_m)]\nonumber\\
    & = \E[\E[\mathbbm{1}_{t\leq t_{\{0\}}-1}\cdot\Delta \text{cost}_{t+1}|(\delta_t, S_t)]\;|(\delta_m, S_m)]\nonumber\\
    &\geq \E[\mathbbm{1}_{t\leq t_{\{0\}}-1, \delta_t\neq 0, |S_t\cap (\delta_t,1]|\geq 1, w_t\leq \delta_t}\cdot\frac{w_t}{2}|(\delta_m, S_m)]\cdot\P(r_{t+1}\in [\tfrac{\delta_t+w_t}{2}, \delta_t +\tfrac{w_t}{2}])\nonumber\\
    &= \E[\mathbbm{1}_{t\leq t_{\{0\}}-1, \delta_t\neq 0, |S_t\cap (\delta_t,1]|\geq 1, w_t\leq \delta_t}\cdot\frac{w_t}{2}|(\delta_m, S_m)]\cdot\E[\mathbbm{1}_{r_{t+1}\in [\tfrac{\delta_t+w_t}{2}, \delta_t +\tfrac{w_t}{2}]}|(\delta_m, S_m)]\nonumber\\
    &= \E[\mathbbm{1}_{t\leq t_{\{0\}}-1, \delta_t\neq 0, |S_t\cap (\delta_t,1]|\geq 1, w_t\leq \delta_t, r_{t+1}\in [\tfrac{\delta_t+w_t}{2}, \delta_t +\tfrac{w_t}{2}]}\cdot\frac{w_t}{2}|(\delta_m, S_m)],\nonumber\\
    &= \E[\mathbbm{1}_{t\leq t_{\{0\}}-1, \delta_t\neq 0, |S_t\cap (\delta_t,1]|\geq 1, w_t\leq \delta_t,  r_{t+1}\in [\tfrac{\delta_t+w_t}{2}, \delta_t +\tfrac{w_t}{2}]}\cdot\frac{\delta_{t+1}-\delta_t}{2}|(\delta_m, S_m)]\nonumber\\
    &\geq \E[\mathbbm{1}_{t\leq \min(t^w, t^d, t_{\{0\}})-1, r_{t+1}\in [\tfrac{\delta_t+w_t}{2}, \delta_t +\tfrac{w_t}{2}]}\cdot\frac{\delta_{t+1}-\delta_t}{2}|(\delta_m, S_m)],\label{eq:LBcost1}
\end{align}
where the third equality uses that $r_{t+1}$ is independent of $(\delta_m, S_m)$, and the fourth equality holds since $r_{t+1}$ and $(\delta_t,S_t)$ are independent, which implies that $r_{t+1}$ and $\{t\leq t_{\{0\}}-1, \delta_t\neq 0, |S_t\cap (\delta_t,1]|\geq 1, w_t\leq \delta_t\}$ are independent. The last equality is by inspecting the case $r_{t+1}\in [\tfrac{\delta_t+w_t}{2}, \delta_t +\tfrac{w_t}{2}]$ in Table  \ref{tab:MC_gamma_S}. Finally, the last inequality is from (\ref{eq:indicators}).

Next, if $t\leq \min(t^w, t^d, t_{\{0\}})-1$ (which,  by (\ref{eq:indicators}), implies in particular that the assumptions of the third point of Lemma \ref{lem:MC_def} are satisfied), we get by inspecting all possible cases given in Table \ref{tab:MC_gamma_S} that $r_{t+1}\in [\tfrac{\delta_t+w_t}{2}, \delta_t +\tfrac{w_t}{2}]$ if and only if $\delta_{t+1}\neq \delta_t$ and $\delta_{t+1}\neq 0$. Thus, 
\begin{equation}
\label{eq:LBcost2}
  \mathbbm{1}_{t\leq \min(t^w, t^d, t_{\{0\}})-1, r_{t+1}\in [\tfrac{\delta_t+w_t}{2}, \delta_t +\tfrac{w_t}{2}]} = \mathbbm{1}_{t\leq \min(t^w, t^d, t_{\{0\}})-1, \delta_{t+1}\neq \delta_t, \delta_{t+1}\neq 0}.  
\end{equation}

In addition, by definition of $t^d$ and by the first point of Lemma \ref{lem:MC_def}, we have that $\delta_{t+1}\neq 0$ if and only if $t\leq t^d - 2$, thus 

\begin{equation}
\label{eq:LBcost3}
    \mathbbm{1}_{t\leq \min(t^w, t^d, t_{\{0\}})-1, \delta_{t+1}\neq \delta_t, \delta_{t+1}\neq 0} = \mathbbm{1}_{t\leq \min(t^w, t^d-1, t_{\{0\}})-1, \delta_{t+1}\neq \delta_t}.
\end{equation}

Hence, by combining (\ref{eq:LBcost1}), (\ref{eq:LBcost2}) and (\ref{eq:LBcost3}), we get that for all $t\in \mathbb{N}$, 

\begin{equation}
\label{eq:LBcost4}
  \E[\mathbbm{1}_{t\leq t_{\{0\}}-1}\cdot\Delta \text{cost}_{t+1}|(\delta_m, S_m)]\geq  \E[\mathbbm{1}_{t\leq \min(t^w, t^d-1, t_{\{0\}})-1, \delta_{t+1}\neq \delta_t}\cdot\frac{\delta_{t+1}-\delta_t}{2}|(\delta_m, S_m)].  
\end{equation}

\noindent\textbf{Lower bound on  $\E[\Delta\text{cost}_{t+1}]$ in the case where $t\geq t_{\{0\}}$.}
 Note that from the first point of Lemma \ref{lem:MC_def}, if $\delta_t = 0$, then $\Delta\text{cost}_{t+1} =0$. In addition, recall that $S_t\cap \{0\}=\emptyset$ when  $t\geq t_{\{0\}}$. Hence, using
the fifth point of Lemma \ref{lem:MC_def}, we get 

\begin{align*}
E[\mathbbm{1}_{t\geq t_{\{0\}} }\cdot\Delta\text{cost}_{t+1}|(\delta_t, S_t)]
    &=\E[\mathbbm{1}_{t\geq t_{\{0\}}, \delta_t\neq 0}\cdot\Delta\text{cost}_{t+1}|(\delta_t, S_t)] + \E[\mathbbm{1}_{t\geq t_{\{0\}}, \delta_t= 0}\cdot\Delta\text{cost}_{t+1}|(\delta_t, S_t)]\\
    &=
    \E[\mathbbm{1}_{t\geq t_{\{0\}}, \delta_t\neq 0}\cdot\Delta\text{cost}_{t+1}|(\delta_t, S_t)] \\
    &\geq-\mathbbm{1}_{t\geq t_{\{0\}}, \delta_t\neq 0}\cdot\P(\delta_{t+1}=0|(\delta_t, S_t))\\
    &=-\mathbbm{1}_{t\geq t_{\{0\}}, \delta_t\neq 0}\cdot E[\mathbbm{1}_{\delta_{t+1}=0}|(\delta_t, S_t)]\\
    &= -\E[\mathbbm{1}_{t\geq t_{\{0\}}, \delta_t\neq 0,\delta_{t+1}=0}|(\delta_t, S_t)].
\end{align*}
Note that by the first point of Lemma \ref{lem:MC_def} and by definition of $t^d$, we have that $\delta_t\neq 0,\delta_{t+1}=0$ if and only if $\delta_{t+1} = t^d$. Thus, 
$\mathbbm{1}_{t\geq t_{\{0\}}, \delta_t\neq 0,\delta_{t+1}=0} = \mathbbm{1}_{ t\geq t_{\{0\}}, \delta_{t+1}=t^d }$. 

By applying the tower law on a similar way as above, we conclude:
\begin{equation}
\label{eq:LBcost5}
  \E[\mathbbm{1}_{t\geq t_{\{0\}}}\cdot\Delta \text{cost}_{t+1}|(\delta_m, S_m)]\geq  -\E[\mathbbm{1}_{t\geq t_{\{0\}}, \delta_{t+1}=t^d}|(\delta_m, S_m)].  
\end{equation}
\noindent\textbf{Concluding the proof of Lemma \ref{cor:costLdelta_gamma}.} We lower bound the difference of costs for matching requests $r_{m+1},\ldots, r_n$:
\begin{align*}
    &\E\Bigg[\sum_{t=m+1}^{n} cost_t(\H^{m-1}) - cost_t(\H^m) |\delta_m,S_m\Bigg]\\
    &=\mathbb{E}\Bigg[ \sum_{t = m}^{n-1} \Delta\text{cost}_{t+1}|(\delta_m, S_m)\Bigg]\\
    &=\mathbb{E}\Bigg[ \sum_{t = m}^{n-1} \mathbbm{1}_{t\leq t_{\{0\}}-1}\cdot \Delta\text{cost}_{t+1}|(\delta_m, S_m)\Bigg]+ \mathbb{E}\Bigg[ \sum_{t = m}^{n-1} \mathbbm{1}_{t\geq t_{\{0\}}}\cdot \Delta\text{cost}_{t+1}|(\delta_m, S_m)\Bigg]\\
    &\geq \E\Bigg[\sum_{t = m}^{n-1} \mathbbm{1}_{t\leq \min(t^w, t^d - 1, t_{\{0\}})-1, \delta_{t+1}\neq \delta_t}\cdot\frac{\delta_{t+1}-\delta_t}{2}\Bigg|(\delta_m, S_m)\Bigg] \\
    &\qquad\qquad\qquad\qquad\qquad\qquad\qquad\qquad\qquad\qquad\qquad -\mathbb{E}\Bigg[ \sum_{t = m}^{n-1} \mathbbm{1}_{t\geq t_{\{0\}},\;\delta_{t+1}= t^d} |(\delta_m, S_m)\Bigg]\\
    &= \E\Bigg[\sum_{t = m}^{ \min( t^w, t^d-1, t_{\{0\}})-1}\mathbbm{1}_{\delta_{t+1} \neq \delta_t} \cdot\frac{\delta_{t+1}-\delta_t}{2}\Bigg|(\delta_m, S_m)\Bigg] - \P(t^d> t_{\{0\}}|(\delta_m, S_m))\\
     &= \frac{1}{2}\mathbb{E}\Bigg[\delta_{\min( t^w, t^d-1, t_{\{0\}})}- \delta_m \Bigg|(\delta_m, S_m)\Bigg]- \P(t^d> t_{\{0\}}|(\delta_m, S_m)),
\end{align*}
where the inequality is by (\ref{eq:LBcost4}) and (\ref{eq:LBcost5}).

In addition, note that by construction of the process, for all $t\in \mathbb{N}$, we have that $S_{t+1}\subseteq S_t$; hence $s_{t,1}\leq s_{t+1,1}$. Then, for all $t\leq t^d-2$, we have that $\delta_t, \delta_{t+1}\neq 0$, hence by construction, we have $\delta_t = s_{t,1}$ and $\delta_{t+1} =  s_{t+1,1}$. Thus, we get that $\delta_t\leq \delta_{t+1}$. Hence, we have that $\delta_{\min(t^w, t^d-1, t_{\{0\}})} = \max_{t \in \{0,\ldots, \min(t^w,t^d-1, t_{\{0\}})\}}\delta_t$. In addition, since by the first point of Lemma \ref{lem:MC_def}, we have that $\delta_t = 0$ for all $t\geq t^d$, we get that $\max_{t \in \{0,\ldots, \min(t^w,t^d-1, t_{\{0\}})\}}\delta_t = \max_{t \in \{0,\ldots, \min(t^w,t_{\{0\}})\}}\delta_t$. Therefore, 

\begin{align*}
    \E\Bigg[\sum_{t=m+1}^{n} cost_t(\H^{m-1}) - cost_t(\H^m) |\delta_m,S_m\Bigg] &\geq \frac{1}{2}\mathbb{E}\Bigg[ \max_{t \in \{0,\ldots, \min(t^w, t_{\{0\}})\}}\delta_t
    - \delta_m \Bigg|(\delta_m, S_m)\Bigg]\\
    & - \P(t^d>t_{\{0\}}|(\delta_m, S_m)).\label{eq:total_cost_geq_m}
\end{align*}

\end{proof}

\lemcostrm*

\begin{proof}
Since $\H^m$ and $\H^{m-1}$ both follow $\mathcal{A}$ for the first $m-1$ requests, it is immediate that
\begin{equation*}
\label{eq:sum_0}
    \E\Bigg[\sum_{t=1}^{m-1} cost_t(\H^{m-1}) - cost_t(\H^m) |\delta_m,S_m\Bigg] = 0.
\end{equation*}
We now lower bound the cost of matching request $r_m$. We consider two cases: \\
(1) If $r_m\in (y_0, 1]$ or ($r_m\in [0,y_0]$ and $S_m\cap \{0\} = \emptyset)$, then both $\H^m$ and $\H^{m-1}$ match $r_m$ greedily. Since $S_{m-1} = S_{m-1}'$, we get $cost_m(\H^{m-1}) = cost_m(\H^m)$.\\
(2) If ($r_m\in [0,y_0]$ and $S_m\cap \{0\} \neq \emptyset)$, then $s(r_m) = 0$, hence $cost_m(\H^{m-1}) - cost_m(\H^m) \geq - cost_m(\H^m) = -|r_m-0| \geq - y_0 = -n^{-1/5}$. 

In both cases, 
\begin{equation*}
\label{eq:delta_m0}
     \E[cost_m(\H^{m-1}) - cost_m(\H^m)|\delta_m,S_m] \geq - n^{-1/5}.
\end{equation*}
Combining the two above equations concludes the proof.
\end{proof}

%% file: main.bbl
\begin{thebibliography}{50}
\providecommand{\natexlab}[1]{#1}
\providecommand{\url}[1]{\texttt{#1}}
\expandafter\ifx\csname urlstyle\endcsname\relax
  \providecommand{\doi}[1]{doi: #1}\else
  \providecommand{\doi}{doi: \begingroup \urlstyle{rm}\Url}\fi

\bibitem[Akbarpour et~al.(2022)Akbarpour, Alimohammadi, Li, and
  Saberi]{Akbarpour21}
Mohammad Akbarpour, Yeganeh Alimohammadi, Shengwu Li, and Amin Saberi.
\newblock The value of excess supply in spatial matching markets.
\newblock \emph{Proceedings of the 23rd ACM Conference on Economics and
  Computation}, page~62, 2022.

\bibitem[Anari et~al.(2023)Anari, Charikar, and
  Ramakrishnan]{anari2023distortion}
Nima Anari, Moses Charikar, and Prasanna Ramakrishnan.
\newblock Distortion in metric matching with ordinal preferences.
\newblock In \emph{Proceedings of the 24th ACM Conference on Economics and
  Computation}, pages 90--110, 2023.

\bibitem[Antoniadis et~al.(2014)Antoniadis, Barcelo, Nugent, Pruhs, and
  Scquizzato]{Antoniadis}
Antonios~Foivos Antoniadis, Neal Barcelo, Michael Nugent, Kirk Pruhs, and
  Michele Scquizzato.
\newblock A o(n) -competitive deterministic algorithm for online matching on a
  line.
\newblock In \emph{Workshop on Approximation and Online Algorithms}, 2014.

\bibitem[Arnosti(2022)]{arnosti2022greedy}
Nick Arnosti.
\newblock Greedy matching in bipartite random graphs.
\newblock \emph{Stochastic Systems}, 12\penalty0 (2):\penalty0 133--150, 2022.

\bibitem[Arthur et~al.(2009)Arthur, Manthey, and R{\"o}glin]{arthur2009k}
David Arthur, Bodo Manthey, and Heiko R{\"o}glin.
\newblock K-means has polynomial smoothed complexity.
\newblock In \emph{2009 50th Annual IEEE Symposium on Foundations of Computer
  Science}, pages 405--414. IEEE, 2009.

\bibitem[Balkanski et~al.(2022)Balkanski, Faenza, and
  Kubik]{balkanski2022simultaneous}
Eric Balkanski, Yuri Faenza, and Mathieu Kubik.
\newblock The simultaneous semi-random model for {TSP}.
\newblock In \emph{International Conference on Integer Programming and
  Combinatorial Optimization}, pages 43--56. Springer, 2022.

\bibitem[Bansal et~al.(2007)Bansal, Buchbinder, Gupta, and Naor]{bansal}
Nikhil Bansal, Niv Buchbinder, Anupam Gupta, and Joseph~Seffi Naor.
\newblock An o(log2k)-competitive algorithm for metric bipartite matching.
\newblock In \emph{Proceedings of the 15th Annual European Conference on
  Algorithms}, ESA'07, page 522–533, Berlin, Heidelberg, 2007.
  Springer-Verlag.
\newblock ISBN 3540755195.

\bibitem[Brown(2016)]{Lyft}
Timothy Brown.
\newblock Matchmaking in lyft line — part 1.
\newblock \emph{Lyft Engineering}, 2016.
\newblock URL \url{https://tinyurl.com/3sdrw7yc}.

\bibitem[Caragiannis et~al.(2016)Caragiannis, Filos-Ratsikas, Frederiksen,
  Hansen, and Tan]{caragiannis2016truthful}
Ioannis Caragiannis, Aris Filos-Ratsikas, S{\o}ren Kristoffer~Stiil
  Frederiksen, Kristoffer~Arnsfelt Hansen, and Zihan Tan.
\newblock Truthful facility assignment with resource augmentation: An exact
  analysis of serial dictatorship.
\newblock In \emph{Web and Internet Economics: 12th International Conference,
  WINE 2016, Montreal, Canada, December 11-14, 2016, Proceedings 12}, pages
  236--250. Springer, 2016.

\bibitem[Chatziafratis et~al.(2017)Chatziafratis, Roughgarden, and
  Vondrak]{chatziafratis2017stability}
Vaggos Chatziafratis, Tim Roughgarden, and Jan Vondrak.
\newblock Stability and recovery for independence systems.
\newblock In \emph{25th Annual European Symposium on Algorithms (ESA 2017)}.
  Schloss Dagstuhl-Leibniz-Zentrum fuer Informatik, 2017.

\bibitem[Conforti and Cornu{\'e}jols(1984)]{conforti1984submodular}
Michele Conforti and G{\'e}rard Cornu{\'e}jols.
\newblock Submodular set functions, matroids and the greedy algorithm: tight
  worst-case bounds and some generalizations of the rado-edmonds theorem.
\newblock \emph{Discrete applied mathematics}, 7\penalty0 (3):\penalty0
  251--274, 1984.

\bibitem[Csaba and Pluhár(2007)]{Csaba}
Bela Csaba and András Pluhár.
\newblock A randomized algorithm for the on-line weighted bipartite matching
  problem.
\newblock \emph{Journal of Scheduling}, 11, 07 2007.
\newblock \doi{10.1007/s10951-007-0037-5}.

\bibitem[Devanur et~al.(2011)Devanur, Jain, Sivan, and
  Wilkens]{devanur2011near}
Nikhil~R Devanur, Kamal Jain, Balasubramanian Sivan, and Christopher~A Wilkens.
\newblock Near optimal online algorithms and fast approximation algorithms for
  resource allocation problems.
\newblock In \emph{Proceedings of the 12th ACM conference on Electronic
  commerce}, pages 29--38, 2011.

\bibitem[Englert et~al.(2014)Englert, R{\"o}glin, and
  V{\"o}cking]{englert2014worst}
Matthias Englert, Heiko R{\"o}glin, and Berthold V{\"o}cking.
\newblock Worst case and probabilistic analysis of the 2-opt algorithm for the
  {TSP}.
\newblock \emph{Algorithmica}, 68\penalty0 (1):\penalty0 190--264, 2014.

\bibitem[Englert et~al.(2016)Englert, R{\"o}glin, and
  V{\"o}cking]{englert2016smoothed}
Matthias Englert, Heiko R{\"o}glin, and Berthold V{\"o}cking.
\newblock Smoothed analysis of the 2-opt algorithm for the general {TSP}.
\newblock \emph{ACM Transactions on Algorithms (TALG)}, 13\penalty0
  (1):\penalty0 1--15, 2016.

\bibitem[Feige(2021)]{feige2021introduction}
Uriel Feige.
\newblock Introduction to semirandom models.
\newblock \emph{Beyond the Worst-Case Analysis of Algorithms}, page 189, 2021.

\bibitem[Feldman et~al.(2010)Feldman, Henzinger, Korula, Mirrokni, and
  Stein]{feldman2010online}
Jon Feldman, Monika Henzinger, Nitish Korula, Vahab~S Mirrokni, and Cliff
  Stein.
\newblock Online stochastic packing applied to display ad allocation.
\newblock In \emph{European Symposium on Algorithms}, pages 182--194. Springer,
  2010.

\bibitem[Frieze et~al.(1990)Frieze, McDiarmid, and Reed]{Frieze}
Alan Frieze, Colin McDiarmid, and Bruce Reed.
\newblock Greedy matching on the line.
\newblock \emph{SIAM Journal on Computing}, 19\penalty0 (4):\penalty0 666--672,
  1990.
\newblock \doi{10.1137/0219045}.
\newblock URL \url{https://doi.org/10.1137/0219045}.

\bibitem[Fuchs et~al.(2003)Fuchs, Hochstättler, and Kern]{Fuchs2003}
Bernhard Fuchs, Winfried Hochstättler, and Walter Kern.
\newblock Online matching on a line.
\newblock \emph{Electronic Notes in Discrete Mathematics}, 13:\penalty0
  49–51, 03 2003.
\newblock \doi{10.1016/S1571-0653(04)00436-6}.

\bibitem[Goel and Mehta(2008)]{goel2008online}
Gagan Goel and Aranyak Mehta.
\newblock Online budgeted matching in random input models with applications to
  adwords.
\newblock In \emph{SODA}, volume~8, pages 982--991, 2008.

\bibitem[Gupta and Lewi(2012)]{GuptaL12}
Anupam Gupta and Kevin Lewi.
\newblock The online metric matching problem for doubling metrics.
\newblock In Artur Czumaj, Kurt Mehlhorn, Andrew~M. Pitts, and Roger
  Wattenhofer, editors, \emph{Automata, Languages, and Programming - 39th
  International Colloquium, {ICALP} 2012, Warwick, UK, July 9-13, 2012,
  Proceedings, Part {I}}, volume 7391 of \emph{Lecture Notes in Computer
  Science}, pages 424--435. Springer, 2012.

\bibitem[Gupta et~al.(2019)Gupta, Guruganesh, Peng, and Wajc]{GuptaGPW19}
Anupam Gupta, Guru Guruganesh, Binghui Peng, and David Wajc.
\newblock Stochastic online metric matching.
\newblock In Christel Baier, Ioannis Chatzigiannakis, Paola Flocchini, and
  Stefano Leonardi, editors, \emph{46th International Colloquium on Automata,
  Languages, and Programming, {ICALP} 2019, July 9-12, 2019, Patras, Greece},
  volume 132, pages 67:1--67:14, 2019.

\bibitem[Gupta et~al.(2020)Gupta, Krishnaswamy, and
  Sandeep]{Gupta2020PERMUTATIONSB}
Varun Gupta, Ravishankar Krishnaswamy, and Sai Sandeep.
\newblock Permutation strikes back: The power of recourse in online metric
  matching.
\newblock In \emph{APPROX-RANDOM}, 2020.

\bibitem[Holden et~al.(2021)Holden, Peres, and Zhai]{Holden21}
Nina Holden, Yuval Peres, and Alex Zhai.
\newblock {Gravitational allocation for uniform points on the sphere}.
\newblock \emph{The Annals of Probability}, 49\penalty0 (1):\penalty0 287 --
  321, 2021.
\newblock \doi{10.1214/20-AOP1452}.
\newblock URL \url{https://doi.org/10.1214/20-AOP1452}.

\bibitem[Jackson(2019)]{UberEats}
Joab Jackson.
\newblock How {U}ber {E}ats uses machine learning to estimate delivery times?
\newblock \emph{The New Stack}, 2019.
\newblock URL
  \url{https://thenewstack.io/how-uber-eats-uses-machine-learning-to-estimate-delivery-times/}.

\bibitem[Kalyanasundaram and Pruhs(1993)]{KalyanasundaramP93}
Bala Kalyanasundaram and Kirk Pruhs.
\newblock Online weighted matching.
\newblock \emph{J. Algorithms}, 14\penalty0 (3):\penalty0 478--488, 1993.

\bibitem[Kalyanasundaram and Pruhs(2000)]{KalyanasundaramP00}
Bala Kalyanasundaram and Kirk Pruhs.
\newblock The online transportation problem.
\newblock \emph{{SIAM} J. Discret. Math.}, 13\penalty0 (3):\penalty0 370--383,
  2000.
\newblock \doi{10.1137/S0895480198342310}.
\newblock URL \url{https://doi.org/10.1137/S0895480198342310}.

\bibitem[Kanoria(2021)]{Kan_arxiv}
Yash Kanoria.
\newblock Dynamic spatial matching, 2021.
\newblock URL \url{https://arxiv.org/abs/2105.07329}.

\bibitem[Kanoria(2022)]{Kanoria21}
Yash Kanoria.
\newblock Dynamic spatial matching.
\newblock In \emph{Proceedings of the 23rd ACM Conference on Economics and
  Computation}, EC '22, page 63–64. Association for Computing Machinery,
  2022.
\newblock ISBN 9781450391504.

\bibitem[Karp et~al.(1990)Karp, Vazirani, and Vazirani]{karp1990optimal}
Richard~M Karp, Umesh~V Vazirani, and Vijay~V Vazirani.
\newblock An optimal algorithm for on-line bipartite matching.
\newblock In \emph{Proceedings of the twenty-second annual ACM symposium on
  Theory of computing}, pages 352--358, 1990.

\bibitem[Khuller et~al.(1994)Khuller, Mitchell, and Vazirani]{khuller1994line}
Samir Khuller, Stephen~G Mitchell, and Vijay~V Vazirani.
\newblock On-line algorithms for weighted bipartite matching and stable
  marriages.
\newblock \emph{Theoretical Computer Science}, 127\penalty0 (2):\penalty0
  255--267, 1994.

\bibitem[Koutsoupias and Nanavati(2004)]{Koutsoupias}
Elias Koutsoupias and Akash Nanavati.
\newblock The online matching problem on a line.
\newblock In Roberto Solis-Oba and Klaus Jansen, editors, \emph{Approximation
  and Online Algorithms}, pages 179--191, Berlin, Heidelberg, 2004. Springer
  Berlin Heidelberg.
\newblock ISBN 978-3-540-24592-6.

\bibitem[K{\"u}nnemann and Manthey(2015)]{kunnemann2015towards}
Marvin K{\"u}nnemann and Bodo Manthey.
\newblock Towards understanding the smoothed approximation ratio of the 2-opt
  heuristic.
\newblock In \emph{International Colloquium on Automata, Languages, and
  Programming}, pages 859--871. Springer, 2015.

\bibitem[Li et~al.(2020)Li, Fang, Zeng, Maag, Tong, and Zhang]{li2020two}
Yiming Li, Jingzhi Fang, Yuxiang Zeng, Balz Maag, Yongxin Tong, and Lingyu
  Zhang.
\newblock Two-sided online bipartite matching in spatial data: experiments and
  analysis.
\newblock \emph{GeoInformatica}, 24\penalty0 (1):\penalty0 175--198, 2020.

\bibitem[Manthey and R{\"o}glin(2013)]{manthey2013worst}
Bodo Manthey and Heiko R{\"o}glin.
\newblock Worst-case and smoothed analysis of k-means clustering with bregman
  divergences.
\newblock \emph{Journal of Computational Geometry (Old Web Site)}, 4\penalty0
  (1):\penalty0 94--132, 2013.

\bibitem[Mastin and Jaillet(2013)]{mastin2013greedy}
Andrew Mastin and Patrick Jaillet.
\newblock Greedy online bipartite matching on random graphs.
\newblock \emph{arXiv preprint arXiv:1307.2536}, 2013.

\bibitem[Megow and N{\"{o}}lke(2020)]{Megow}
Nicole Megow and Lukas N{\"{o}}lke.
\newblock Online minimum cost matching with recourse on the line.
\newblock In Jaroslaw Byrka and Raghu Meka, editors, \emph{Approximation,
  Randomization, and Combinatorial Optimization. Algorithms and Techniques,
  {APPROX/RANDOM} 2020, August 17-19, 2020, Virtual Conference}, volume 176 of
  \emph{LIPIcs}, pages 37:1--37:16, 2020.

\bibitem[Mehta(2013)]{Mehta}
Aranyak Mehta.
\newblock Online matching and ad allocation.
\newblock \emph{Foundations and Trends in Theoretical Computer Science}, 8
  (4):\penalty0 265--368, 2013.
\newblock URL \url{http://dx.doi.org/10.1561/0400000057}.

\bibitem[Mehta et~al.(2013)]{mehta2013online}
Aranyak Mehta et~al.
\newblock Online matching and ad allocation.
\newblock \emph{Foundations and Trends{\textregistered} in Theoretical Computer
  Science}, 8\penalty0 (4):\penalty0 265--368, 2013.

\bibitem[Meyerson et~al.(2006)Meyerson, Nanavati, and Poplawski]{Myerson}
Adam Meyerson, Akash Nanavati, and Laura Poplawski.
\newblock Randomized online algorithms for minimum metric bipartite matching.
\newblock In \emph{Proceedings of the Seventeenth Annual ACM-SIAM Symposium on
  Discrete Algorithm}, SODA '06, page 954–959, USA, 2006. Society for
  Industrial and Applied Mathematics.
\newblock ISBN 0898716055.

\bibitem[Nayyar and Raghvendra(2017)]{Nayyar}
Krati Nayyar and Sharath Raghvendra.
\newblock An input sensitive online algorithm for the metric bipartite matching
  problem.
\newblock In \emph{2017 IEEE 58th Annual Symposium on Foundations of Computer
  Science (FOCS)}, pages 505--515, 2017.
\newblock \doi{10.1109/FOCS.2017.53}.

\bibitem[Peserico and Scquizzato(2021)]{peserico2021matching}
Enoch Peserico and Michele Scquizzato.
\newblock Matching on the line admits no $o(\sqrt{log n})$-competitive
  algorithm.
\newblock In \emph{48th International Colloquium on Automata, Languages, and
  Programming (ICALP 2021)}, 2021.

\bibitem[Pokutta et~al.(2020)Pokutta, Singh, and
  Torrico]{pokutta2020unreasonable}
Sebastian Pokutta, Mohit Singh, and Alfredo Torrico.
\newblock On the unreasonable effectiveness of the greedy algorithm: Greedy
  adapts to sharpness.
\newblock In \emph{International Conference on Machine Learning}, pages
  7772--7782. PMLR, 2020.

\bibitem[Raghvendra(2016)]{Raghvendra16}
Sharath Raghvendra.
\newblock A robust and optimal online algorithm for minimum metric bipartite
  matching.
\newblock In Klaus Jansen, Claire Mathieu, Jos{\'{e}} D.~P. Rolim, and Chris
  Umans, editors, \emph{Approximation, Randomization, and Combinatorial
  Optimization. Algorithms and Techniques, {APPROX/RANDOM} 2016, September 7-9,
  2016, Paris, France}, 2016.

\bibitem[Raghvendra(2018)]{raghvendra2018optimal}
Sharath Raghvendra.
\newblock Optimal analysis of an online algorithm for the bipartite matching
  problem on a line.
\newblock In \emph{34th International Symposium on Computational Geometry (SoCG
  2018)}, 2018.

\bibitem[Rubinstein and Zhao(2022)]{rubinstein2022budget}
Aviad Rubinstein and Junyao Zhao.
\newblock Budget-smoothed analysis for submodular maximization.
\newblock In \emph{13th Innovations in Theoretical Computer Science Conference
  (ITCS 2022)}. Schloss Dagstuhl-Leibniz-Zentrum f{\"u}r Informatik, 2022.

\bibitem[Sinharay(2010)]{Beta_distrib}
S.~Sinharay.
\newblock Continuous probability distributions.
\newblock In Penelope Peterson, Eva Baker, and Barry McGaw, editors,
  \emph{International Encyclopedia of Education (Third Edition)}, pages
  98--102. Elsevier, Oxford, third edition edition, 2010.
\newblock ISBN 978-0-08-044894-7.
\newblock \doi{https://doi.org/10.1016/B978-0-08-044894-7.01720-6}.
\newblock URL
  \url{https://www.sciencedirect.com/science/article/pii/B9780080448947017206}.

\bibitem[Tong et~al.(2016)Tong, She, Ding, Chen, Wo, and Xu]{TongSDCWX16}
Yongxin Tong, Jieying She, Bolin Ding, Lei Chen, Tianyu Wo, and Ke~Xu.
\newblock Online minimum matching in real-time spatial data: Experiments and
  analysis.
\newblock \emph{Proc. {VLDB} Endow.}, 9\penalty0 (12):\penalty0 1053--1064,
  2016.

\bibitem[Tsai et~al.(1994)Tsai, Tang, and Chen]{TsaiTC94}
Ying~The Tsai, Chuan~Yi Tang, and Yunn~Yen Chen.
\newblock Average performance of a greedy algorithm for the on-line minimum
  matching problem on euclidean space.
\newblock \emph{Inf. Process. Lett.}, 51\penalty0 (6):\penalty0 275--282, 1994.

\bibitem[Xu et~al.(2019)Xu, Shi, Cheng, Dickerson, Sankararaman, Srinivasan,
  Tong, and Tsepenekas]{xu2019unified}
Pan Xu, Yexuan Shi, Hao Cheng, John Dickerson, Karthik~Abinav Sankararaman,
  Aravind Srinivasan, Yongxin Tong, and Leonidas Tsepenekas.
\newblock A unified approach to online matching with conflict-aware
  constraints.
\newblock In \emph{Proceedings of the AAAI Conference on Artificial
  Intelligence}, volume~33, pages 2221--2228, 2019.

\end{thebibliography}
